%% file: main.tex
\documentclass[11pt,a4paper]{article}

\usepackage{amsmath}

\usepackage{times}
\usepackage{bm}
\usepackage[authoryear, round]{natbib}
\def\T{{ \mathrm{\scriptscriptstyle T} }}

\usepackage{amsfonts} 
\newcommand{\n}{\leavevmode\newline}

\newcommand{\NN}{\mathbb{N}}
\newcommand{\PP}{\mathbb{P}}

\newcommand{\RR}{\mathbb{R}}
\newcommand{\EE}{\mathbb{E}}

\newcommand{\leqst}{\leq_{\text{st}}}
\newcommand{\normm}[1]{\left\lVert#1 \right\rVert}
\newcommand\roww[1]{_{#1,\boldsymbol{\cdot}}}

\newcommand{\N}{\text{N}}
\newcommand\ind{\mathbbm{1}}
\usepackage{enumitem}
\usepackage{bbm}
\usepackage{algpseudocode}
\usepackage{adjustbox}
\usepackage{xr}

\usepackage[utf8]{inputenc}
\usepackage{float}
\usepackage[a4paper, total={6in, 8in}]{geometry}
\usepackage{fullpage} 
\usepackage{parskip} 
\usepackage{tikz} 
\usepackage{amsmath}
\usepackage{hyperref}
\hypersetup{
    colorlinks=true,
    linkcolor=purple,
    filecolor=magenta,      
    urlcolor=pink,
}
\usepackage{amssymb}
\usepackage{enumitem}
\usepackage{bm}
\usepackage{rotating}
\usepackage{bbm}
\usepackage{adjustbox}
\usepackage{framed}
\usepackage[lofdepth,lotdepth]{subfig}
\usepackage{amsthm}
\usepackage{listings}

\usepackage{mathtools}

\mathtoolsset{showonlyrefs=true}
\usepackage{algorithm}
\usepackage{algpseudocode} 

\hypersetup{
    colorlinks   = true, 
    urlcolor     = blue, 
    linkcolor    = blue, 
    citecolor   = black 
}

\newtheoremstyle{mitheorem}
{11pt} 
{11pt} 
{\itshape} 
{} 
{\bfseries} 
{.} 
{.5em} 
{} 

\newtheoremstyle{midef}
{11pt} 
{11pt} 
{\upshape} 
{} 
{\bfseries} 
{.} 
{.5em} 
{} 

\theoremstyle{mitheorem}
\newtheorem{theorem}{Theorem}[section]
\newtheorem{lemma}[theorem]{Lemma}

\newtheorem{proposition}[theorem]{Proposition}
\theoremstyle{midef}

\newtheorem{rec}[theorem]{Plan}

\renewcommand{\emptyset}{\varnothing}
\renewcommand{\epsilon}{\varepsilon}
\renewcommand{\phi}{\varphi}

\newcommand\ceq\coloneqq 
\newcommand\eps\epsilon

\newcommand\pt\partial

\newcommand{\iid}{\overset{\text{i.i.d.}}{\sim}}
\usepackage{pdfpages}
\makeatletter
\renewcommand*\env@matrix[1][*\c@MaxMatrixCols c]{%
 \hskip -\arraycolsep
 \let\@ifnextchar\new@ifnextchar
 \array{#1}}
\makeatother

\begin{document}

\title{Efficient sparsity adaptive changepoint estimation}
\author{Per August Jarval Moen$^{1}$ \and Ingrid Kristine Glad$^{1}$ \and Martin Tveten$^{2}$}

\date{$^1$Department of Mathematics, University of Oslo \\
$^2$Norwegian Computing Center}

\maketitle

\begin{abstract}
\input{sections/abstract}
\end{abstract}

\input{sections/introduction}

\input{sections/problemdescription}
\input{sections/methodandresults}
\input{sections/numericalstudies}

\input{sections/dataexample}

\section*{Acknowledgement}
We thank Idris Eckley, Arnoldo Frigessi and Nils Lid Hjort for constructive discussions and feedback. We also thank Camilla Feurst and Statkraft for providing the data used in our real-life example. This project has been partially funded by the centre BigInsight, Norwegian Research Council, project number 237718.

\bibliography{preliminary}

\section*{Appendices}
\addcontentsline{toc}{section}{Appendices}
\renewcommand{\thesubsection}{\Alph{subsection}}
\input{sections/proofs}

\input{sections/appendixA}
\input{sections/appendixB}
\input{sections/appendixC}

\input{sections/appendixD}
\input{sections/lemmas}

\end{document}

%% file: sections/abstract.tex
\noindent We propose a computationally efficient and sparsity adaptive procedure for estimating changes in unknown subsets of a high-dimensional data sequence. Assuming the data sequence is Gaussian, we prove that the new method successfully estimates the number and locations of changepoints with a given error rate and under minimal conditions, for all sparsities of the changing subset. Our method has computational complexity linear up to logarithmic factors in both the length and number of time series, making it applicable to large data sets. Through extensive numerical studies we show that the new methodology is highly competitive in terms of both estimation accuracy and computational cost. The practical usefulness of the method is illustrated by analysing sensor data from a hydro power plant, and an efficient R implementation is available.

%% file: sections/introduction.tex
\section{Introduction}
During the last decades, new technology has made it possible to gather data in larger quantities from an ever wider range of sources. Data can often display non-stationarities in the form of distributional changes over time, leading to incorrect statistical inferences if not accounted for. Inference on changepoints may also be of interest in it self. For instance, \citet{cunen_statistical_2020} search for changes in the number of battle deaths in interstate wars between 1816 and 2007, \citet{gao_surface_2020} study monitoring of the temperature of transplant organs, and \citet{10.1214/21-AOAS1508} use a changepoint detection algorithm for condition monitoring of a subsea pump.

In this paper, we study the problem of detection and estimation of an unknown number of changes in the mean of high-dimensional data. By \textit{detection}, we refer to testing for the presence of one or more changepoints in the data. By \textit{estimation}, we refer to estimation of the location(s) of the changepoint(s). This problem is well understood in the literature for univariate data,. Several computationally efficient algorithms have been proposed during the last decade, including Pruned Exact Linear Time of \citet{killick_optimal_2012}, Wild Binary Segmentation of \citet{fryzlewicz_wild_2014}, Narrowest Over Threshold of \citet{baranowski_narrowest-over-threshold_2019} and Seeded Binary Segmentation of \citet{kovacs_seeded_2022}. Notably, these methods have been shown to achieve near optimal performance, in a minimax sense, see \citet{wang_univariate_2019}.

Several methods for the multivariate change in mean problem have also been proposed, although this problem is less studied than the univariate setting. 
The Inspect method of \citet{wang_high_2018} uses sparse projections of CUSUM statistics and a variant of Wild Binary Segmentation to detect and localize multiple sparse changes in the mean. \citet{cho_multiple-change-point_2015} propose the Sparsified Binary Segmentation algorithm based on thresholding and aggregating CUSUM statistics over coordinates, in combination with Binary Segmentation. The Double CUSUM method of \citet{cho_change-point_2016} uses test statistics based on ordered CUSUMs, in combination with ordinary Binary Segmentation. The SUBSET method of \citet{tickle_computationally_2021} uses a penalized likelihood approach, in combination with the Wild Binary Segmentation search procedure.

In this work, we present a novel multiple changepoint estimation algorithm, which we call ESAC (\textbf{E}fficient \textbf{S}parsity \textbf{A}daptive \textbf{C}hangepoint estimator). The method is designed to detect and estimate the locations of an unknown number of changes in the mean of high-dimensional data sequences. An important feature of ESAC is that the subset of data components that undergo a change need not be known --- it can be anything from a single changing component to a small subset to all components. We refer to the size of the changing subset as the \textit{sparsity} of the change. ESAC comes with strong theoretical guarantees, and is in particular adaptive to all sparsities of changes and all distances between changepoints. Still, the worst-case computational cost of ESAC is linear in the number of observations, $n$, as well as the number of components, $p$, save for logarithmic factors. Via simulations, we demonstrate that ESAC is highly competetive in terms of statistical accuracy and running time.
%
%
%

The novelty of our work is threefold. Our first contribution is to modify a sparsity adaptive test statistic proposed by \cite{liu_minimax_2021} to make it suitable for testing for changepoints in a multiple changepoint setup. In particular, our proposed test facilitates control over its family-wise error rate, which is necessary in the multiple changepoint situation. Our second contribution is to propose a novel estimator for the location of a single changepoint. The estimator comes with strong theoretical guarantees, and may be of independent interest. Lastly and most importantly, we combine our proposed test statistic and changepoint estimator into a multiple changepoint estimation algorithm, ESAC, using a slight variant of Seeded Binary Segmentation \citep{kovacs_seeded_2022} and Narrowest-over-Threshold \citep{baranowski_narrowest-over-threshold_2019} selection of changepoints. ESAC is also efficiently implemented in an R package \textbf{HDCD} \citep{moen_hdcd_2023}, available on The Comprehensive R Archive Network (\href{https://cran.r-project.org}{cran.r-project.org}). Efficient implementations of Inspect \cite{wang_high_2018} and the method of \cite{pilliat_optimal_2022} are also available in the package.

Most similar to ESAC is the multiple changepoint detection procedure of \citet{pilliat_optimal_2022} for Gaussian changes in mean. The theoretical guarantees, for instance, are the same for ESAC and the method of \citet{pilliat_optimal_2022}. Still, there are important distinctions between the two methods, which we highlight here. As opposed to ESAC, the method proposed by \citet{pilliat_optimal_2022} is based on their novel "bottom up" search. Their approach segments the data into disjoint segments chosen as narrow as possible from a predefined grid, where for each interval, a test statistic must have detected a changepoint. To ensure a disjoint segmentation, they merge overlapping segments of equal length whenever a changepoint is detected in both. From this segmentation, changepoint locations are estimated by taking midpoints of the segments. Consequently, \citeauthor{pilliat_optimal_2022}'s method only requires a test for a changepoint, and not a location estimator. In practice, this generality comes at a cost of changepoints being crudely estimated or not being detected at all, whenever the signal strength is low. This is illustrated in our simulation studies, which feature empirical comparisons between ESAC, the method of \cite{pilliat_optimal_2022} and other proposed methods.

The paper is organized as follows. In Section \ref{secproblemdesc} we give a formal description of the model assumed throughout the paper. In Section \ref{singletest} we present a test statistic for a single changepoint that facilitates control over its family-wise error rate. In Section \ref{locsingle} we propose an estimator for the location of a single changepoint, also stating its finite sample estimation error rate with comparisons to other methods. In Section \ref{fullprocedure} we propose ESAC, our proposed  multiple changepoint estimation procedure. In Section \ref{subsectheory} we present theoretical results regarding the statistical and computational properties of ESAC and compare these to other methods.  In Section \ref{secnumerical} we study the empirical performance of ESAC and other methods via simulations, including for misspecified models. In Section \ref{dataex} we apply ESAC to sensor data from a Swedish hydro power plant. In Appendix \ref{secproofs} we prove our main theoretical results. In the remained of the Appendix we discuss implemention of ESAC in practice, provide more simulation results, and prove auxiliary lemmas for our main results.

We use the following notation throughout the paper. For any vector $y \in \mathbb{R}^d$ we let $y_j$ denote its $j$-th component, $\normm{y}_2$ denote its Euclidean norm and $\normm{y}_0$ denote the number of non-zero entries in $y$. For any matrix $X \in \RR^{p \times n}$ we let $X_{i,v}$ denote its $(i,v)$th element, $X_v \in \RR^p$ denote its $v$th column, $X_{\roww{i}} \in \RR^{n}$ denote its $i$th row, $\normm{X}_F^2 = \sum_{i=1}^p \sum_{v=1}^n X_{i,v}^2$ denote the squared Frobenius norm of $X$, and $\normm{X}_1 = \sum_{i=1}^p \sum_{v=1}^n |X_{i,v}|$ denote the entry-wise $\ell_1$ norm of $X$. For any pair of matrices $X,Y \in \RR^{p \times n}$, we let $\left\langle X,Y\right\rangle = \text{tr}\left(X^{\T} Y\right)$ denote their trace inner product. For any positive integer $I$ we define $[I]= \{1, \ldots, I\}$. For any pair of real numbers $x,y$, we define $x\vee y= \max \left\{  x,y \right\}$ and $x \wedge y = \min \left\{ x,y \right\}$. For any pair of random variables $X,Y$, we let $X \leqst Y$ mean that $Y$ stochastically dominates $X$, i.e. $\PP \left(X \leq t \right) \geq \PP \left( Y \leq t\right)$ for all $t \in \RR$. 
For any $x \in \RR$, we let $\lfloor x \rfloor$ denote the largest integer no larger than $x$, and $\lceil x \rceil$ denote the smallest integer no smaller than $x$. 

We find it useful to adopt the notation of \cite{baranowski_narrowest-over-threshold_2019} to denote integer intervals. For any pair of integers $s,e$ such that $s\leq e-2$, we let $(s,e)$ denote the open integer interval $\{s+1, \ldots, e-1\}$ and let $(s,e]$ denote the left-open and right closed integer interval $\{s+1, \ldots, e\}$.

%% file: sections/problemdescription.tex
\section{Problem description}\label{secproblemdesc}
To motivate our method and allow for theoretical analysis, we consider the following model for the remainder of the paper. In Section \ref{secnumerical} we assess performance under various misspecified models. Suppose we observe $n\geq 2$ independent multivariate Gaussian variables
\begin{equation}
X_{v} = \mu_{v} + W_{v} \label{datamodel},
\end{equation}
where $\mu_{v} \in \RR^p$ and $W_{v} \sim \N_p(0, \sigma^2I)$ for $v \in [n]$. 
Assume that there are $J\geq 0$ changepoints $0<\eta_1 < \ldots < \eta_J<n$ such that
$$
\mu_{v}\neq \mu_{v+1} \text{ if and only if } v = \eta_j \text{ for some } j \in [J].
$$
Let $\theta_j = \mu_{\eta_{j+1}} - \mu_{\eta_j}$ denote the change in mean occurring at the $j$th changepoint, and let $\phi_j = \normm{\theta_j}_2$ be the $\ell_2$-norm of the mean-change occuring at changepoint $j$. Further, let $k_j = \normm{\theta_j}_0$ denote the \textit{sparsity} of the $j$th changepoint, i.e. the number of non-zero components of $\theta_j$. Lastly, let $\Delta_j = \min\left(  \eta_j - \eta_{j-1}, \eta_{j+1} - \eta_j  \right)$ denote the minimum distance between the $j$th changepoint and a neighboring changepoint (where we for convenience take $\eta_{0}= 0 $ and $\eta_{J+1}=n$). Our goal is to estimate $J$, the number of changepoints, and their locations $\eta_1 <\ldots<\eta_J$.

In the theoretical analysis to follow, we take $\sigma^2$ to be known. For notational compactness, let $X$, $\mu \in \RR^{p \times n}$ denote the matrices with $X_v$, $\mu_v$ as their $v$th columns, respectively, for $v \in [n]$. 

%% file: sections/methodandresults.tex
\section{Method and results}\label{secadaptive}

\subsection{Single changepoint detection with family-wise error rate control}\label{singletest}

We begin by presenting a statistical test for a single change in mean in some arbitrary interval $(s,e) = \{s+1, \ldots, e-1\}$, where $0\leq s <e \leq n$, and $s\leq e-2$. To simplify the exposition, assume for now that $\sigma=1$, as the data can be normalized to satisfy this assumption. We seek a test statistic which facilitates control over the family-wise error rate when testing for changepoints over multiple intervals. This level of control is needed later on when we define a multiple changepoint algorithm in Section \ref{fullprocedure}.

To construct a test statistic, we build upon the work of \citet{liu_minimax_2021}. They propose an efficient and minimax rate optimal test statistic for testing for a single changepoint within an interval. Unfortunately, this test does not allow for control of the family-wise error rate, and thus needs some modifications, presented next. For any candidate changepoint location $v$ such that ${0\leq s<v<e\leq n}$, we define the CUSUM transformation $T^v_{(s,e]}(y)$ of a vector $y \in \RR^n$ as
\begin{align}
T^v_{(s,e]}(y)&= \left\{\frac{e-v}{(e-s)(v-s)}\right\}^{1/2}\sum_{i=s+1}^v y_i - \left\{\frac{v-s}{(e-s)(e-v)}\right\}^{1/2}\sum_{i=v+1}^e y_i\label{c2def},
\end{align}
To simplify notation, we use $C_{(s, e]}^v(i) = T^v_{(s, e]}(X_{i, \cdot})$ to denote the CUSUM of the $i$th component of the data within the integer interval $(s, e]$, evaluated at candidate changepoint position $v$. 

Given a candidate sparsity level $t \in [p]$, and penalizing function $\gamma(t)$, both to be discussed later, define 
\begin{align}
S_{\gamma, (s,e]}^v(t) = \sum_{i=1}^p \left \{  C^v_{(s,e]} \left(i\right)^2 - \nu_{a(t)}       \right\}\mathbbm{1}\left\{|C^v_{(s,e]}(i)| \geq a(t) \right\}  - \gamma(t) \label{stdef},
\end{align}
where the threshold value $a(t)$ is given by 
\begin{align}
a^2(t) = 4 \log\left( \frac{ep \log n}{t^2}\right) \mathbbm{1}\left\{t \leq (p\log n)^{1/2} \right\}\label{atdef},
\end{align} and
$\nu_{a(t)}$ is a mean-centering term defined by taking $\nu_a = \mathbbm{E}\left( Z^2 \ | \ |Z|\geq a\right)$ for $Z\sim \text{N}(0,1)$ and $a\geq0$. In \eqref{atdef} we abuse notation slightly, writing $e = \exp(1)$ to mean Euler's number. 

We refer to $S_{\gamma, (s,e]}^v(t)$ as a \textit{sparsity-specific penalized score}, which heuristically measures the degree of evidence of a changepoint at $v \in (s,e)$ with sparsity $t$. To see this, we note that the CUSUM is a well known test statistic for a univariate change in mean, rejecting the null hypothesis of a constant mean for large values \citep[see][]{wang_univariate_2019}. In the multivariate case, the statistic \eqref{stdef} aggregates CUSUM values in an effort to borrow information across coordinates. To prevent the noise drowning the signal, the CUSUMs are thresholded at level $a(t)$, tailored specifically for a given sparsity level $t$. The functional form of $a(t)$ is depicted in Figure \ref{r_and_t} , where one observes that $a(t)$ decreases with $t$. Hence, $S_{\gamma, (s,e]}^v(t)$ thresholds CUSUMs less harshly as $t$ grows.  More intuition on $a(t)$ is discussed in the end of this subsection. Also plotted in Figure \ref{r_and_t} is $\nu_{a(t)}$, which serves as a mean-centering term for spuriously large CUSUMs, and satisfies $a^2(t) \leq \nu_{a(t)} \leq a^2(t)+2$. Even after thresholding and mean-centering, however, the sum in \eqref{stdef} need not be small even when no changepoint is present. The role of the penalty function $\gamma(t)$ is therefore to ensure that $S^v_{\gamma, (s,e]}<0$ with high probability whenever no changepoint is present.  

The sparsity-specific penalized score $S_{\gamma, (s,e]}^v(t)$ is defined for a fixed sparsity $t$, while the true sparsity of a changepoint is taken to be unknown. To measure the overall degree of evidence of a changepoint at location $v$, we consider an exponentially increasing grid of candidate sparsity levels
 \begin{equation}
\mathcal{T} = \{1, 2, 4, \ldots, 2^{\log_2 \left\{\lfloor \surd(p \log  n) \rfloor \right\} }\} \cup \{p\} \label{mathcaltdef}.
\end{equation}
This approach is also taken by \citet{liu_minimax_2021}, where the grid $\mathcal{T}$ is slightly smaller. This choice of grid is justified as follows. Whenever a changepoint has true sparsity $k < (p\log n)^{1/2}$, there always exists some $t \in \mathcal{T}$ such that $ t/2 \leq k \leq t$, which turns out to be sufficient for detecting the changepoint. Conversely, when the true sparsity $k$ satisfies $\geq (p\log n)^{1/2}$, it is sufficient to consider $t=p$.

For a candidate changepoint position $s<v<e$, define the \textit{penalized score} $S^v_{\gamma,(s,e]}$ as
\begin{align}
S^v_{\gamma, (s,e]} &= \underset{t \in \mathcal{T}}{\max} \ S^v_{\gamma, (s,e]} (t) \label{sdef},
\end{align}
which heuristically measures the degree of evidence of a changepoint at location $v$, irregardless of the sparsity. As test statistic for a changepoint in the interval $(s,e)$, we take 
\begin{equation}
S_{\gamma, (s,e]} = \ind\left\{ \underset{s<v<e}{\max} S^v_{\gamma, (s,e]} >0\ \right\}\label{testhihi}.
\end{equation}

As for the penalty function $\gamma(t)$, for $t \in [p]$ define
\begin{align}
r(t) = r(t,n,p) &= \begin{cases} (p \log n)^{1/2} & \text{ if } t\geq(p \log n)^{1/2}, \\
t \log \left( \frac{e p \log  n }{t^2}\right) \vee \log n & \text{ otherwise.}
\end{cases}\label{rdef}
\end{align}
With penalty function $\gamma(t) = \gamma_0 r(t)$ for some suitably large constant $\gamma_0>0$, we obtain the following control over the family-wise error rate. 
\begin{proposition}\label{lemmalemmalemma}
Consider the model in Section \ref{secproblemdesc}. For all $s,e$ and $v$ such that $0\leq s < v < e \leq n$, assume that the quantity $S_{\gamma,(s,e]}$ is computed with variance-scaled input matrix $\widetilde{X} = ({1}/{\sigma})X$ and penalty function $\gamma(t) = \gamma_0 r(t)$ for some $\gamma_0>0$.
Let $I$ denote the set of all intervals $(s,e) \subseteq (0,n)$ containing no changepoint, i.e. satisfying $\eta_j \notin (s,e) \ \forall j \in [J]$. For any $\epsilon >0$, there exists a universal choice of $\gamma_0>0$ (depending only on $\epsilon$) such that
$$
\PP\left( \underset{(s,e) \in I}{\max} \ S_{\gamma, (s,e]} >0   \right) \leq \epsilon.
$$
\end{proposition}

Some remarks are in order.
Figure \ref{r_and_t} displays plots of $a^2(t)$, $\nu_{a(t)}$ and $r(t)$ as functions of $t$, for $n=p=500$. As our first remark, we observe that $a(t)$ and $\nu_{a(t)}$ are decreasing in $t$, while $r(t)$ is increasing for all $t\leq (p\log (n) /e)^{1/2}$, but with a bulk when $t$ is close to $(p\log n)^{1/2}$. Several equivalent monotonic functions can be chosen in the place of $r(t$), but we have chosen $r(t)$ due to its simple analytical form. The function $r(t)$ can be seen as the information theoretic detection boundary in terms of the signal-to-noise (SNR) ratio for multiple changes in mean of sparsity $t$ in $p$-dimensional Gaussian vectors with sample size $n$ (see Section \ref{subsectheory} or \citealt{pilliat_optimal_2022}). When $p=1$, we recover the standard penalty used in the univariate changepoint literature for Gaussian changes in mean (see e.g. \citealt{wang_univariate_2019}), as $r(1) = \log n$ in this case. As our second remark, the forms of the threshold $a(t)$ and penalty function $\gamma(t)\propto r(t)$ reflect the two sparsity regimes known in the statistical literature on multivariate mean change detection. In the \text{sparse} case where $t\leq (p\log n)^{1/2}$, the threshold $a(t)$ is non-zero and satisfies $a^2(t) \approx r(t)/t$, which is decreasing with $t$. Meanwhile, in the dense case where $t > (p\log n)^{1/2}$, the threshold satisfies $a(t) = 0$, in which case no thresholding takes place and all CUSUMs contribute to \eqref{stdef}.

\begin{figure}[ht]
\includegraphics[width=\textwidth]{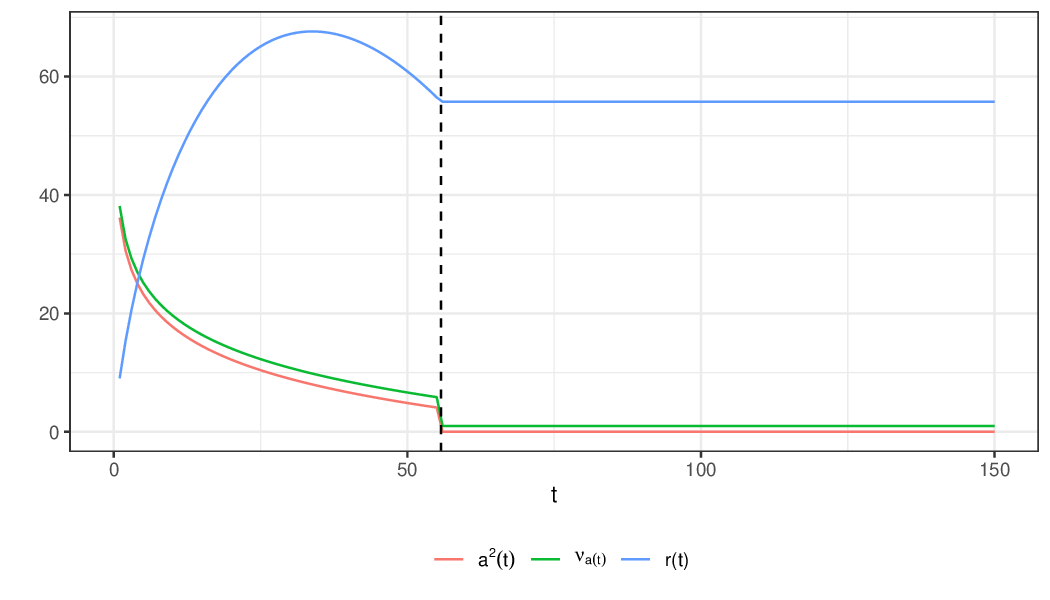}
\caption{Plots of $a^2(t)$ (red), $\nu_{a(t)}$ (green), and $r(t)$ (red) for $n=p=500$. The boundary between the dense and sparse regimes is given by a vertical dashed line at $k = (p\log n)^{1/2}$.}\label{r_and_t}
\end{figure}

Lastly, we remark the following relation between our proposed test statistic and that of \cite{liu_minimax_2021}. To facilitate control over the family-wise error rate, the threshold $a(t)$, the mean-centering term $\nu_{a(t)}$, and the penalty function $\gamma(t) = \gamma_0 r(t)$ grow faster with $n$ than their equivalent counterparts in \citet{liu_minimax_2021}. To recover the test statistic from \citet{liu_minimax_2021}, one must replace $\log n$ by $\log \log (8n)$ in \eqref{stdef}, \eqref{atdef}, \eqref{mathcaltdef}, \eqref{sdef}, \eqref{testhihi} and \eqref{rdef} and use the penalty function $\gamma(t) = \gamma_0 r(t)$ with the modified function $r(t)$.

\subsection{Single changepoint estimation}\label{locsingle}
We now consider the problem of estimating the location of a changepoint within some interval $(s,e)$, assuming the changepoint has already been detected or is known to be present. As before, we assume $\sigma=1$, as the data can be normalized to satisfy this assumption. We further assume that the interval $(s,e)$ contains only a single changepoint. Recalling the model defined in Section \ref{datamodel}, this means that $J=1$ and $[s,e) = [0,n)$. The problem at hand is then to estimate the location $\eta=\eta_1$ of a single changepoint, taking the sparsity $k = k_1$ as unknown. 
As our estimator for $\eta$, we use the test statistic in \eqref{sdef} with a potentially separate penalty function $\lambda(t)$, and maximize it over all candidate changepoints. That is,
\begin{align}
\widehat{\eta}_{\lambda} = \underset{0<v<n}{\arg \max} \ S^v_{\lambda},\label{singlechangeloc}
\end{align}
where we write $S^v_{\lambda}$ shorthand for $S^v_{\lambda, (0,n]}$.
The rationale behind this estimator is that $S^v(\lambda)$ measures the degree of evidence of a changepoint at location $v$. The maximizer is unique with high probability whenever the signal strength is reasonably large. Still, to ensure that $\widehat{\eta}_{\lambda}$ is always well defined, we formally set $\widehat{\eta}_{\lambda}$ be the smallest maximizer, although we suppress this from the notation. 

The finite sample properties of $\widehat{\eta}_{\lambda}$ are given in Theorem \ref{locprop}, which holds whenever $\lambda(t) = \lambda_0 r(t)$ for sufficiently large $\lambda_0>0$. Before stating the theorem, some remarks are in order. The threshold $a(t)$ is as before decreasing in $t$, while a penalty function of the form $\lambda(t) = \lambda_0 r(t)$ is increasing for most values of $t$. For such $\lambda(t)$, the penalized score $S^v_{\lambda}$ involves an implicit model selection in terms of the sparsity. 
An inspection of the proof of Theorem \ref{locprop} reveals the following intuition about the roles of $\lambda(t)$ and $a(t)$; whenever $\lambda_0$ is sufficiently large, the penalty function $\lambda(t) = \lambda_0 r(t)$ cancels all contributions to the sum in \eqref{stdef} from coordinates where there is no signal. Left are the contributions from the affected the coordinates. For a changepoint with a moderate signal strength, these remaining contributions are maximized when $t$ is close to the true sparsity $k$. To see this, note that when is $t$ smaller than the true sparsity $k$, the thresholding is too strict, possibly cancelling the signal from the affected coordinates, either from the thresholding itself or from the mean-centering term. When $t$ is greater $k$, the thresholding is less strict, but the penalty function $\lambda(t)$ may be larger than the signal strength from the affected coordinates. 

The following finite sample result shows that the estimator $\widehat{\eta}_{\lambda}$ is adaptive to the (unknown) sparsity and gives a high-probability upper bound on the estimation error.
\begin{theorem}\label{locprop}
Consider the model in Section \ref{secproblemdesc}, with only one changepoint $\eta$, with sparsity $k$ and $\ell_2$ norm $\phi$. Let $\Delta = \min\left( \eta , n - \eta\right)$. Let $\widehat{\eta}_{\lambda}$ be as in \eqref{singlechangeloc}, when the sparsity-specific score function $S^v_{\lambda}(t) = S^v_{\lambda,(0,n]}$ from \eqref{stdef} is computed with variance-scaled input matrix $\widetilde{X} = ({1}/{\sigma}) X$ and penalty function $\lambda(t) = \lambda_0 r(t)$, where $\lambda_0>0$. Define
\begin{align}
h(t) = h(t,n,p) &= \begin{cases} \left[ p \left\{ \log  n \vee \log \log (ep) \right\}\right]^{1/2} & \text{ if } t\geq(p \log  n)^{1/2}, \\
t \log \left( \frac{e p \log  n }{t^2}\right) \vee \log  n & \text{ otherwise.}
\end{cases}\label{hdef}
\end{align}
There exist a universal choice of $\lambda_0>0$ and universal constants $C_0,C_1>0$ such that, if 
\begin{align}
\frac{\phi^2\Delta}{\sigma^2} \geq C_0 h(k)\label{SNRsingle},
\end{align}
we have that
$$
\PP \left\{ |\widehat{\eta}_{\lambda} - \eta| \leq C_1 \frac{\sigma^2}{\phi^2} h(k) \right\} \geq 1-\frac{1}{n}.
$$
\end{theorem}

The SNR requirement \eqref{SNRsingle} implies that the absolute estimation error of $\widehat{\eta}_{\lambda}$ satisfies $ C_1h(k) {\sigma^2}/{\phi^2} \leq C_0 / C_1 \Delta < \Delta$ whenever the conditions of the Theorem holds. In particular, in an asymptotic regime where $k$, $p$, $\Delta$ and $\phi$ vary with $n$, the quantity $|\widehat{\eta}_{\lambda} - \eta|/\Delta$ converges in probability to $0$ as $n \rightarrow \infty$ whenever $({\phi^2\Delta})/\{{\sigma^2}h(k)\}$ diverges with $n$. Similarly, if ${\phi^2}/\{{\sigma^2}h(k)\}$ diverges with $n$, then $\widehat{\eta}_{\lambda}$ converges in probability to $\eta$ as $n\rightarrow \infty$. Note that Theorem \ref{locprop} requires that the penalty function $\lambda(t)$ has a specific functional form. For practical choices of the penalty function $\lambda(t)$, we refer to Appendix \ref{appendixA}.
 
Some performance comparisons with related methods are in order. In comparison with Theorem \ref{locprop}, the SNR condition for the Inspect method \citep{wang_high_2018} in the single changepoint case is of the form $\phi^2 \Delta/\sigma^2 \geq C v(k, n, p,\Delta)$ for some $C>0$, where $ v(k, n, p,\Delta) = (n/\Delta) k \log \left( p \log n \right)$. Ignoring constants, we see that lthe SNR condition of ESAC is weaker than that of Inspect whenever $k\geq \log n$ and for all values of $k$ whenever $p\geq n/\log n$. Note also that $v(k,n,p,\Delta)$ consists of the factor $n / \Delta$, which is not the case for ESAC. Once the SNR condition for Inspect is satisfied, its error rate is of order $(\sigma^2/\phi^2) \log \log n$, which is smaller than that of ESAC, and especially so whenever $k$ is large. For the Double CUSUM algorithm of \citet[Section~4]{cho_change-point_2016} in the single changepoint case where $\sigma=1$, the asymptotic SNR requirement for consistency implies that $\phi^2\Delta / (k \log^2 n) \rightarrow \infty$. By "consistency" we mean that $|\widehat{\eta}_{\text{DC}} - \eta|/\Delta$ converges to $0$ in probability, where $\widehat{\eta}_{\text{DC}}$ is the Double CUSUM estimate of $\eta$. The (asymptotic) SNR requirement for the Double CUSUM algorithm is thus larger than that of ESAC by a factor of at least $k \log^2(n) / r(k)$, which grows with $k$ and diverges with $n$. Note that theoretical results for the Double CUSUM algorithm only hold whenever $p$ is of the same order as $n^{\omega}$ for some fixed $\omega>0$. For an empirical comparison between the methods in the single changepoint case, we refer to Section \ref{numericsingle}. 

The SNR requirement of ESAC grows much more slowly with $k$ than Inspect and the DC algorithm, which implies that ESAC is able to reasonably estimate dense changepoints under much lower signal strength. Figure \ref{ratessingle} displays the SNR requirement of ESAC, Inspect and DC as a function of $k$, when $n=p=500$. On the left plot, the SNRs are plotted on linear scale, while on a log scale to the right. The boundary between the dense and sparse regimes is indicated by the vertical dashed line at $k= (p\log n)^{1/2}$. As the SNRs are only defined up to constants, each SNR requirement is normalized to have value $1$ for sparsity $k=1$. For ESAC, the normalized SNR ratio is $h(k)/h(1)$, while it is simply $k$ for Inspect and DC. As we can see, the SNR condition of ESAC grows much more slowly with $k$ than Inspect or DC. We remark that the apparent bulk in the SNR requirement of ESAC is a result of keeping the mathematical expression simple, as remarked in Section \ref{singletest}.

\begin{figure}[ht]
\includegraphics[width=\textwidth]{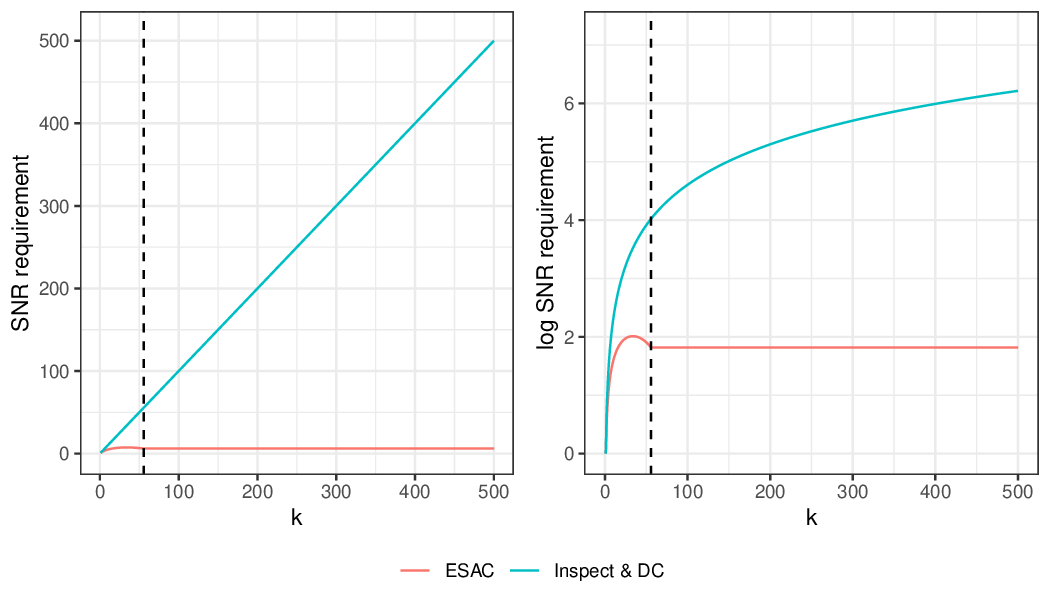}
\caption{Normalized SNR conditions of ESAC (red) and Inspect and the Double CUSUM algorithm (blue), plotted as a function of the sparsity $k$, on a linear scale (left) and log scale (right). The boundary between the dense and sparse regimes is given by a vertical dashed line at $k = (p\log n)^{1/2}$.}\label{ratessingle}
\end{figure}

Lastly, we remark the following. For fixed values of $n$ and $p$, the function $h(k)$ is increasing for most values of $k$. 
Thus, the estimation error and the SNR condition of ESAC tend to grow with the sparsity $k$. As an example, the error rate for estimating a changepoint with sparsity $k=1$ is $(\sigma^2 / \phi^2) \left( \log n \vee \log p\right)$, while the same error rate becomes $(\sigma^2 / \phi^2)\left[ p \left\{ \log n \vee \log \log (ep)\right\}\right]^{1/2}$ for $k=p$. 

\subsection{Detection and estimation of multiple changepoints}\label{fullprocedure}
We now consider the combined problem of detecting and estimating an unknown number of changepoints in the data $X_1, \ldots , X_n$. That is, our goal is to estimate $J$, the number of changepoints, and $(\eta_1, \ldots, \eta_J)^\top$, the changepoint locations. 

Our proposed test statistic from Section \ref{singletest} and changepoint estimator from Section \ref{locsingle} are designed for segments with at most a single changepoint. Hence, a search procedure is needed to allow for multiple changepoint search. Our choice of search procedure is a slight variant of Seeded Binary Segmentation \citep{kovacs_seeded_2022}. In essence, the Seeded Binary Segmentation search procedure generates a deterministic set of intervals (which they call seeded intervals), in which single changepoint candidates are searched for. As a single changepoint may be detected within several distinct intervals, a choice must be made regarding which of these intervals is to be used for estimating its location. We have opted for the Narrowest-Over-Threshold \citep{baranowski_narrowest-over-threshold_2019} choice of changepoints, using the narrowest interval in which a changepoint is detected to estimate its location. Our modification of Seeded Binary Segmentation is minor; in our variant, the generation of intervals is controlled by two parameters, $\alpha$ and $K$. The parameter $K$ controls the distance between the centers of two consecutive intervals of the same length, and the parameter $\alpha$ controls the growth rate of the interval lengths. Our algorithm for generating seeded intervals is found in Appendix \ref{appendixA} (Algorithm \ref{alg:seededintervals}).

Our proposed multiple changepoint estimation procedure, ESAC, is as follows. Let $\mathcal{M} =  \{(s_m, e_m] \ ; \ m\in [M]\}$ denote an enumerated collection of candidate intervals. Let $\gamma(t), \lambda(t)$ denote the penalty functions used in the score statistic \eqref{stdef}, for changepoint detection and estimation, respectively. Given data matrix $X$, our proposed procedure is initiated by calling the recursive algorithm ESAC$(X, (0,n], \mathcal{M}, \emptyset, \gamma, \lambda)$, defined by Algorithm \ref{alg:cap2}. 

For the theoretical analysis of ESAC given in Section \ref{subsectheory}, we find it necessary to consider a slightly modified variant of the algorithm, defined by Algorithm \ref{alg:cap} in Appendix \ref{appendixA}. In this variant, candidate changepoint locations are trimmed away in the recursive step, discarding them from future use to detect or estimate further changepoints. The trimming of changepoints is introduced as a necessary technical step for the proof of Theorem \ref{adaptivetheorem} to go through, specifically to ensure that previously discovered changepoints are not re-discovered.  In practice, we find trimming to be unnecessary and even degrading of performance. An even more modified variant of ESAC, given by Algorithm \ref{alg:cap3} defined in Appendix \ref{appendixA}, takes only the midpoint of an interval as the only candidate changepoint location when testing for a changepoint, in addition to interval trimming.  This results in a substantial decrease in running time at the cost of reduced detection power, although the theoretical results in the next subsection hold for this variant as well. In practical application, we thus recommend using Algorithm \ref{alg:cap2} over Algorithms \ref{alg:cap} and \ref{alg:cap3}. A simulation study comparing the variants of ESAC is found in Appendix \ref{appendixC}.

\begin{algorithm}[H]
\caption{ESAC$\left( X, (s, e], \mathcal{M}, \mathcal{B}, \gamma, \lambda \right)$}\label{alg:cap2}
\textbf{Input: } A matrix of observations $X \in \RR^{p \times n}$, an open integer interval $(s,e)$ in which candidate changepoints are searched for, an enumerated collection $\mathcal{M} =  \{(s_m, e_m] \ ; \ m\in [M]\}$ of $M$ half open integer sub intervals of $(0,n]$, a set of already detected changepoints $\mathcal{B}$, and penalty functions $\gamma(t), \lambda(t)$.  \\
\textbf{Output: } Set $\mathcal{B}$ of already detected changepoints.\\
\begin{tabbing}
\qquad \enspace if {$e-s\leq 1$}\\
\qquad \qquad \textbf{return} $\mathcal{B}$\\
\qquad \enspace set $\mathcal{M}_{(s,e]} = \left\{   m \in [M] \ : \ (s_m, e_m] \subset (s,e]   \right\}$\\
\qquad \enspace set $\mathcal{O}_{(s,e]} = \left\{   m \in \mathcal{M}_{(s,e]} \ : \ \underset{s_m < v <e_m}{\max}  S^{v}_{\gamma, (s_m,e_m]} >0 \right\}$  \\
\qquad \enspace if {$\mathcal{O}_{(s,e]} = \emptyset$}\\
\qquad \qquad \textbf{return} $\mathcal{B}$\\
\qquad \enspace set $l^*= \underset{m \in \mathcal{O}_{(s,e]}}{\text{min}} |e_m - s_m|$\\
\qquad \enspace set $\mathcal{O}_{l^*} =  \left\{   m \in \mathcal{O}_{(s,e]} \ : \  |e_m - s_m| = l^* \right\}$\\
\qquad \enspace set $m^*= \underset{m \in \mathcal{O}_{l^*}}{\text{argmax}} \underset{s_m < v <e_m}{\max} S^v_{\lambda, (s_m, e_m]}$ \\
\qquad \enspace set $v^*= \underset{s_{m^*} <v <e_{m^*}}{\text{argmax}} S^{v}_{\lambda,(s_{m^*},e_{m^*}]}$\\
\qquad \enspace $\mathcal{B} \gets \mathcal{B} \cup \{v^*\}$\\
\qquad \enspace $\mathcal{B}  \gets \text{ESAC} \left( X,(s, v^*],\mathcal{M}, \mathcal{B}, \gamma, \lambda \right)$\\
\qquad \enspace $\mathcal{B} \gets \text{ESAC} \left( X,(v^*,e], \mathcal{M}, \mathcal{B}, \gamma, \lambda \right)$\\
\qquad \enspace \textbf{return} $\mathcal{B}$
\end{tabbing}
\end{algorithm}

\subsection{Theoretical results for the multiple changepoint case}\label{subsectheory}
For the variants of ESAC defined by either Algorithm \ref{alg:cap} or Algorithm \ref{alg:cap3} (both given in in Appendix \ref{appendixA}), we have the following finite-sample statistical result.
\begin{theorem}\label{adaptivetheorem}
Let $X \in \RR^{p\times n}$ follow the model in Section \ref{secproblemdesc}, and let $r(t)$ be defined as in \eqref{rdef}. Let $\mathcal{M}$ denote the set of candidate intervals generated from Algorithm \ref{alg:seededintervals} with parameters $\alpha \leq 2$, $K\geq 2$, and let the penalty function $\gamma(t)$ be defined as $\gamma(t) = \gamma_0 r(t)$. There exists a universal choice of $\gamma_0>0$, such that for some universal constants $C_0, C_1>0$, depending only on $\gamma_0$, and for any choice of $\lambda(t)$,
the following holds. \n\n
Let $\widehat{J}$ and $\widehat{\eta}_1, \ldots, \widehat{\eta}_{\widehat{J}}$ respectively be denote the estimated number of changepoints and sorted changepoint locations from Algorithm \ref{alg:cap} or \ref{alg:cap3}, using penalty functions $\gamma(t)$ and $\lambda(t)$, candidate intervals $\mathcal{M}$, and variance scaled input matrix 
$\widetilde{X} = X/\sigma$. If the SNR condition 
\begin{equation}
\frac{\phi_j^2 \Delta_j}{\sigma^2} \geq C_0 r({k}_j)\label{SNR}
\end{equation}holds for all $j \in [J]$, we have that
\begin{align}
\PP \left\{      \widehat{J} = J \ \cap  \  |\widehat{\eta}_j - \eta_j| \leq C_1 \frac{\sigma^2}{\phi_j^2} r({k}_j)  \ \forall \ j \in [J]    \right\} > 1-\frac{1}{n}\label{theresult}.
\end{align}
\end{theorem}
\n
The explicit values of $\gamma_0, C_0$ and $C_1$ can be found in the proof of Theorem \ref{adaptivetheorem} in Appendix \ref{secproofs}. We remark that these constants have not been optimized. In practice, we recommend choosing the penalty functions $\gamma(t)$ via Monte Carlo simulation or setting $\lambda(t), \gamma(t)$ proportional to a slight variant of $r(t)$. In particular, when using $\lambda(t), \gamma(t) \propto r(t)$, our simulations suggest that the leading constants can be chosen independently of $n$ and $p$, at least for the values of $(n,p)$ we have considered. For further details and recommendations, we refer to Appendix \ref{impdetails}. 

Note that, whenever the SNR condition \eqref{SNR} holds, the localization error of ESAC satisfies $C_1 {\sigma^2} r({k}_j)/{\phi_j}  < \Delta_j$. If the stronger condition ${\phi_j^2 }/\{{\sigma^2 r({k}_j)}\} \rightarrow \infty$ as $n\rightarrow \infty$ holds for all $j$ (allowing all model parameters to vary with $n$), then ${\max}_{j \in [J]} |  \widehat{\eta}_j - \eta_j |$ converges to $0$ in probability as $n \rightarrow \infty$. Note also that Theorem \ref{adaptivetheorem} holds for any choice of the penalty function $\lambda(t)$. This is because the bound on the estimation error in Theorem \ref{adaptivetheorem} relies upon the detection properties of penalized score $S^v_{\gamma, (s,e]}$ rather than the localization properties of our estimator $\widehat{\eta}_{\lambda}$. Specifically, since ESAC uses Narrowest-over-Threshold selection of changepoints \citep{baranowski_narrowest-over-threshold_2019}, it suffices to upper bound the minimum interval width required to detect a changepoint, which for each $j$ is of the order of ${\sigma^2}r({k}_j)/{\phi_j^2}$. This observation is due to \citet{pilliat_optimal_2022}. In practice, we experience that using the estimator $\widehat{\eta}_{\lambda}$ for changepoint localization improves performance compared to using e.g. the mid-point of an interval. For more details and a simulation study, see Appendix \ref{appendixC}.

Some performance comparisons to related methods are in order. Theorem \ref{adaptivetheorem} gives a very similar theoretical guarantee as the method of \citet[Corollary 3]{pilliat_optimal_2022}. When the probability of the desired event in the Corollary is the same as in Theorem \ref{adaptivetheorem}, the method of Pilliat obtains the same error rate under an up to constants equal SNR requirement. 

The Inspect method (\citealt{wang_high_2018}) requires a signal-to-noise condition of the form ${\phi^2 \Delta}/{\sigma^2} \geq C \log(np)  \{ \left({n}/{\Delta}\right)^3 \ \vee \ k \}({n}/{\Delta})$, where $C>0$ is some universal constant, $\Delta = \min_{j=1,\ldots, J} \Delta_j$, $\phi = \min_{j=1, \ldots, J} \phi_j$ and $k = \max_{j=1,\ldots, J} k_j$. The SNR condition of Inspect is (up to constants) larger than \eqref{SNR} by factors of at least ${n}/{\Delta} \geq 1$ and ${k\log (np)}/{r(k)} \geq 1$. The former factor is only close to $1$ whenever there are  few changepoints with large spacing between, while the latter factor is only close to $1$ whenever all changepoints are sparse. For the Double CUSUM algorithm \citep{cho_change-point_2016} when $\sigma=1$, its (asymptotic) SNR condition requires $\phi_j^2 \Delta_j / ( n^{3-5\beta/2} k_j \log^2 n)\rightarrow 0$ whenever $n^{\beta} = \mathcal{O} (\Delta_j)$ for all $j \in [J]$ for some $\beta \in (6/7,1]$, as well as $p$ being of the same order as $n^{\omega}$ for some fixed $\omega>0$. The SNR requirement for the Double CUSUM algorithm is thus larger than that of ESAC by factors of at least $\surd{n}$ and $k_j \log^2(n) / r(k_j)$. The latter factor diverges as $n \rightarrow \infty$ or $k_j \rightarrow \infty$. 

Figure \ref{ratesmulti} displays the SNR requirements of ESAC (red), Inspect (green) and the Double CUSUM algorithm (blue) on a log scale as a function of the sparsity $k$, plotted for different values of $n$ and $p$. Here we have removed the factor $\left({n}/{\Delta}\right)^3 >1$ from the SNR requirement of Inspect, as well as setting $\Delta = n/2$. To the left, we plot the requirements for $n=10^2, 10^3, 10^4$, indicated by solid, dashed and dotted lines, respectively, keeping $p=500$ fixed. Moreover, each SNR requirement is normalized to have value $1$ for sparsity $k=1$ at $n=100$. To the right, we plot the requirements for $p=500,1000,2000$, indicated by solid, dashed and dotted lines, respectively, keeping $n=1000$ fixed. Here, each SNR requirement is normalized to have value $1$ for sparsity $k=1$ for $p=500$.  From Figure \ref{ratesmulti} we see that the SNR requirement of ESAC grows substantially slower with $k$ than Inspect and the Double CUSUM. The SNR requirement of ESAC and Inspect grow substantially slower with $n$ than the Double CUSUM, while the SNR requirement of ESAC grows faster with $p$ than Inspect and the Double CUSUM. As before, the apparent bulk in the SNR requirement of ESAC is a result of keeping $r(t)$ mathematically simple, as remarked in Section \ref{singletest}.

\begin{figure}[ht]
\includegraphics[width=\textwidth]{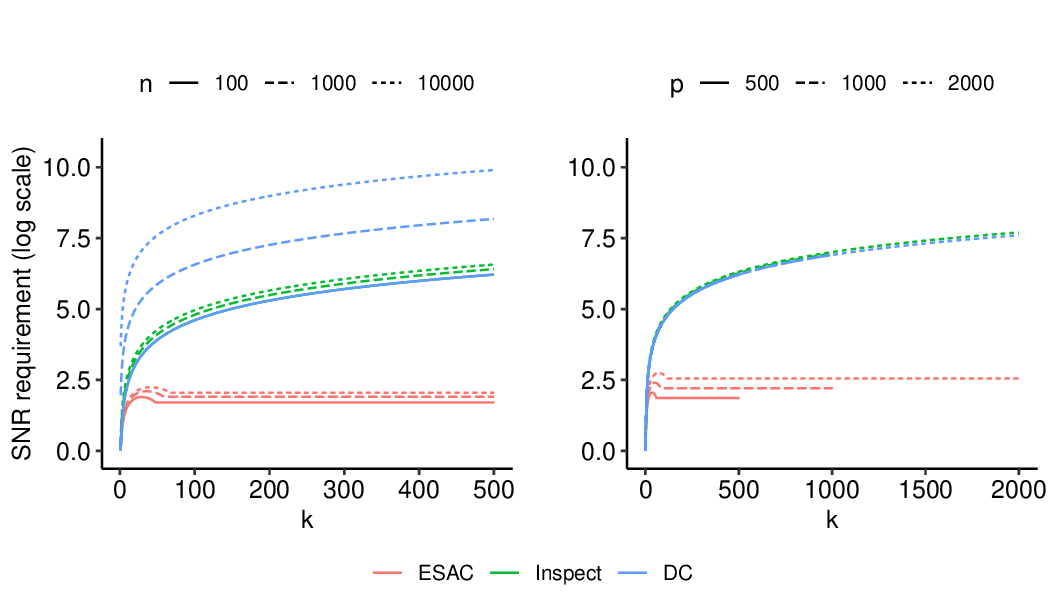}
\caption{Normalized SNR conditions of ESAC (red) and Inspect (green) and the Double CUSUM algorithm (blue) on a log scale, plotted as a function of the sparsity $k$, and varying values of $n$ (left) and $p$ (right). }\label{ratesmulti}
\end{figure}

As for the changepoint location error rates , Inspect obtains a theoretical error rate of ${\sigma^2}({n}/{\Delta})^4 \log(np) / \phi^2$ whenever its SNR condition is satisfied. In comparison, the theoretical error rate of ESAC is at most ${\sigma^2} r(k_j) / {\phi_j^2} r(k_j)$, which is smaller than the error rate of Inspect whenever ${n}/{\Delta}$ is large (short distance between changepoints) or $k$ is large.
For the Double CUSUM algorithm, in the case where $\sigma=1$, the theoretical error rate for each changepoint is at least of the order $\log^2(n) k_j / \phi_j^2$. This is larger than the error rate of ESAC by a factor of at least $k_j \log^2(n) / r(k_j)$, which diverges with $n$ and grows with $k_j$. 

We now turn to optimality considerations. Observe first that the SNR condition \eqref{SNR} is up to constants minimal for identifying $J$, the number of changepoints. Indeed, for any $n$, $p$ and $k \leq p$, an implication of Theorem 2 in \citet{pilliat_optimal_2022} is that
$$
\sup_{P \in Q(n,p,k)} \ P \left( \left|\widehat{\eta}\right|  \neq J\right) \geq 1/4
$$
for all estimators $\widehat{\eta} = \widehat{\eta}(X_1, \ldots, X_n)$ of the changepoint vector $\eta = (\eta_1, \ldots, \eta_J)$, where $Q(n,p,k)$ is the class of all probability distributions of $X_1, \ldots, X_n$ corresponding to the model given in Section \ref{secproblemdesc} for which $k_j \leq k$ and $\phi_j^2 \Delta_j / \sigma^2 \geq c r(k_j)$ for all $j \in [J]$, for some sufficiently small $c>0$. As for the changepoint location error rate, the minimax rate has been shown by \citet{wang_high_2018} to be at least $\sigma^2 / (16\phi^2)$ whenever $\Delta^{-1} \leq \phi^2/\sigma^2 \leq 1$. Hence, at least in this region of the parameter space, the estimator $\widehat{\eta}$ from Section \ref{locsingle} and the full ESAC algorithm have minimax optimal error rates up to factors of $h(k)$ and $r(k)$, respectively, where $k$ is the sparsity of the changepoint in question. Note that, while $r(k)$ and $h(k)$ are constant multiples of $\log n$ whenever $k=1$, they grow substantially with $k$. Hence the error rate of ESAC is only close to minimax rate optimal for small values of the sparsity $k$. 

Finally, we consider the computational cost of ESAC as a function of the size of the data. The following Proposition shows that ESAC has a log-linear computational cost. 

\begin{proposition}\label{compprop}
Consider any input matrix $X \in \RR^{p \times n}$, penalty functions $\gamma(t), \lambda(t)$ and seeded intervals generated by Algorithm \ref{alg:seededintervals} with fixed input parameters $\alpha>1$ and $K \in \NN$. Then the computational complexity of ESAC (either Algorithm \ref{alg:cap2}, \ref{alg:cap} or \ref{alg:cap3}) measured in floating point operations 
is of order $\mathcal{O}\left\{np  \log (p\log n)  \right\}$ in the best case and $\mathcal{O}\left\{np \log n \log (p \log  n ) \right\}$ in the worst case. 
\end{proposition}
In comparison, the computational complexity of the Pilliat algorithm is $\mathcal{O}\{np \log (np)\}$, which is of slightly smaller order than the worst-case complexity of ESAC. For the other multiple changepoint methods like the Double CUSUM, Sparsified Binary Segmentation, SUBSET and Inspect, no specific forms of computational cost are provided in the respective articles. For an empirical comparison of running times, we refer to the next section.

%% file: sections/numericalstudies.tex
\section{Simulations}\label{secnumerical}
We now compare the empirical performance of ESAC with the following state-of-the-art methods for high-dimensional changepoint detection and estimation: a variant of the Inspect method by \citet{wang_high_2018}, the method of \citet{pilliat_optimal_2022} hereby called Pilliat, Sparsified Binary Segmentation  of \citet{cho_multiple-change-point_2015}, the Double CUSUM algorithm of \citet{cho_change-point_2016}, and the SUBSET method by \citet{tickle_computationally_2021}. We introduce a slightly modified variant of Inspect, based on Narrowest-over-Threshold search, mainly to reduce computational cost. The details of our modified Inspect algorithm can be found in Appendix \ref{appendixB}. To run the Sparsified Binary Segmentation and Double CUSUM algorithms, we use the R package \textbf{hdbinseg} \citep{cho_hdbinseg_2018}. To run SUBSET, we use the code from the Github repository of \citet{tickle_subset_2022}. We have implemented the remaining methods ESAC, Pilliat and Inspect in the C programming language, which are found in the R package \textbf{HDCD} \citep{moen_hdcd_2023}, available on CRAN. We remark that our implementations of Inspect and the method \citeauthor{pilliat_optimal_2022} are orders of magnitude faster than their original implementations. Whenever running times are reported, they have been run using R (4.2.1) on a MacOS (12.3) computer with an (ARM) Apple M1 Pro CPU.

For each method in the simulation study, a choice of penalty parameters must be made, which is discussed in each subsection. In all simulations, changes in mean are taken to have magnitudes spread evenly across all affected coordinates. In Appendix \ref{appendixD} we run the same simulations with uneven and random magnitudes, giving similar results. In all simulations 
we assume $\sigma=1$ is unknown. We estimate $\sigma$ separately for each of the $p$ coordinates of the observed time series, and use it to normalize the data before applying each changepoint detection method. As is commonly done in the changepoint literature, we estimate the noise level by the median absolute deviation of first-order differences with scaling factor 1.05 for the Gaussian distribution. 

\subsection{Single changepoint estimation}\label{numericsingle}
In this subsection we consider the algorithms' performance when estimating the location of a single changepoint, assuming that it has already been detected. Our simulations are run with parameters $n \in \left\{ 200,500\right\}$, $p \in \left\{ 100,1000,5000\right\}$ , $k \in \left\{1, \lceil p^{1/3} \rceil, \lceil \surd{(p\log n)} \rceil , p\right\}$. For each configuration of these parameters, we simulate $1000$ data sets and apply the methods considered in the study. For each combination of $n,p,k$, the simulated data sets have a changepoint at $\eta = \lceil n/5 \rceil$ with change-vector $\theta \propto (I_1, \ldots, I_k , 0, \ldots, 0)^{\T}$, where $I_1, \ldots, I_k$ are drawn independently and uniformly from $\{-1,1\}$. For each sample we scale $\theta$ such that $\Delta\phi^2 = ({n}/{5})\normm{\theta}_2^2 = ({5^2}/{2^2})r(k)$, where $k$ is the sparsity of the change and $\Delta = \lceil n/5 \rceil$.

To keep the simulation study simple, we use the authors' recommended non-empirical choices of penalty parameters. We take ESAC to be the estimator given in \eqref{singlechangeloc}, with penalty function $\widetilde{\lambda}(t)$ as defined in Appendix \ref{impdetails}. As for Inspect, we use Algorithm 2 in \citet{wang_high_2018}, with penalty parameter $\lambda = \left\{{\log \left( p \log n\right)/2}\right\}^{1/2}$. 
For the Double CUSUM algorithm we set $\phi = -1$, corresponding to the version presented in Section 4.1 of \citet{cho_change-point_2016}. For Sparsified Binary Segmentation, a default choice of the threshold $\pi_T$ is not available, so we take $\pi_T$ to be the maximum value of the CUSUMs $| T_{[0,n)}^v (Z\roww{i} / \widehat{\sigma}_i)|$ over all values of $0<v<n$ and $i \in [p]$, where $Z_{v,i} {\sim} \N(0, 1)$ independently for $i \in [p]$, $v \in [n]$, and $\widehat{\sigma}_i$ is the median absolute deviation of the noise level in the $i$th series, based on $1000$ Monte Carlo samples. 
Whenever the Sparsified Binary Segmentation estimator is not defined, we set its output to be $1$. For both the Double CUSUM and Sparsified Binary Segmentation algorithm, we have specified $\textit{height}=1$ when calling the respective functions to turn the methods into single changepoint estimators. The Pilliat algorithm is not included in this simulation as there is no straightforward way to modify it into a single changepoint estimator. 

For each method and each configuration of parameters, Table \ref{tablesingleloc} displays the average Mean Squared Error (MSE) and average running time in milliseconds. The Double CUSUM and Sparsified Binary Segmentation methods are abbreviated as DC and SBS, respectively. For each configuration of parameters, the minimum value of both the MSE and the running is indicated in boldface. In terms of statistical accuracy, Table \ref{tablesingleloc} demonstrates that ESAC and SUBSET are the only methods with competitive accuracy across the sparsity regimes, although ESAC has a slight edge over SUBSET. ESAC has the lowest MSE in $14$ out of the $24$ different combinations of parameters (including both dense and sparse regimes), while SUBSET has the lowest MSE in $6$ out of $24$. When averaging the MSE over all the rows, ESAC is the clear winner, with SUBSET in second place. In comparison, the estimation accuracy of Inspect is excellent for $k = \lceil p^{1/3} \rceil$, but deteriorates for higher sparsity levels. The Double CUSUM algorithm displays excellent estimation accuracy when $k=1$, but often not so for dense changepoints (although this seems to vary slightly with $n$ and $p$). Sparsified Binary Segmentation has high estimation accuracy for sparse changepoints (especially when $k=1$), but the accuracy deteriorates for dense changepoints.

In terms of running time, ESAC is the clear winner, with execution time around one fifth of SUBSET, the runner up, and down to $4\%$ of the execution time of Inspect and Double CUSUM for large values of $p$. Note that SUBSET is the only method not implemented in C or C++, giving the other methods an advantage when comparing running times. We also remark that the running time of scaling the data by the median absolute deviations is not included in the running times of ESAC, Inspect and SUBSET, as it would otherwise dominate the running time. The running time of the scaling is included in the running times of the Double CUSUM and Sparsified Binary Segmentation algorithms, as the implementations of these algorithms do not offer an option to disable it.

\input{tables/tablesingleloc/table_even}

\subsection{Multiple changepoint estimation}\label{numericmultiple}
In this subsection we consider the situation of an unknown number of changepoints. Our simulations are run with parameters 
$n=200$, $p \in \left\{ 100,1000,5000\right\}$ and $J \in \left\{0,2,5\right\}$. For each simulated data set we take the changepoint locations $\eta_1, \ldots, \eta_J$ to be ordered and uniformly drawn samples from $\left\{1, \ldots, n-1\right\}$ without replacement. For each combination of $n, p$ and $J$, we consider three different sparsity regimes; \textit{dense}, \textit{sparse} and \textit{mixed}. In the dense and sparse regimes, we sample $k_1, \ldots, k_J$ independently and uniformly from $\left\{\lceil \surd{(p\log n)} \rceil, \ldots, p\right\}$ and $\left\{1, \ldots, \lfloor \surd{(p\log n)} \rfloor\right\}$, respectively. In the mixed regime we sample each $k_j$ independently from a mixture between the dense and sparse regimes, each with equal probability. For each combination of $n,p,J$ and sparsity regime, each changepoint has change-vector $\theta_j \propto (I_{j,1}, \ldots, I_{j,k} , 0, \ldots, 0)\T$, where $I_{j,1}, \ldots, I_{j,k}$ are drawn independently and uniformly from $\{-1,1\}$, scaled such that $\Delta_j\phi_j^2/{\sigma_j^2} = (7/2)^2 r(k_j)$. Notice that we have increased the signal strength slightly in comparison with the single changepoint case, as multiple changepoint estimation is more challenging than estimating the position of a single changepoint whose existence is known. For each combination of $n,p,J$ and sparsity regime we simulate $1000$ data sets. 

For both ESAC and the modified Inspect algorithm, we generate seeded intervals using Algorithm \ref{alg:seededintervals} with parameters $\alpha = 3/2$ and $K=4$. For the Pilliat method we generate intervals using Algorithm \ref{alg:seededintervals} with parameters $\alpha = 3/2$ and $K=2$, giving very similar intervals as the $a$-adic grid $\mathcal{G}_a$ defined in \citet{pilliat_optimal_2022} for $a=2/3$. Due to high computational cost, we run SUBSET with only $100$ randomly drawn intervals in its Wild Binary Segmentation step. To ensure comparability with the remaining methods, we have modified the Pilliat algorithm so that its tests for a changepoint in an integer interval $(s,e]$ are performed by testing for a changepoint at each candidate position $s<v<e$, instead of only the mid-point. In our experience, testing only at the mid-point of an interval results in lower detection power.

We choose detection thresholds using either Monte Carlo simulations (using $N=1000$ samples) or bootstrapping (using $B=100$ samples). For the ESAC algorithm, we use the penalty functions $\widetilde{\gamma}(t)$ and $\widetilde{\lambda}(t)$ given in Appendix \ref{impdetails}, using a false positive rate $\epsilon = 1/N$ to generate the former. For the Pilliat algorithm we choose detection thresholds for the Partial Sum statistic and the dense statistic by Monte Carlo simulating the leading constant in the theoretical thresholds given in \citet{pilliat_optimal_2022}, and apply a Bonferroni correction. For the modified version of Inspect we set $\lambda = \left\{{\log(p\log n)/2}\right\}^{1/2}$ and choose the detection threshold $\xi$ to be the largest sparse projection 
over all seeded intervals and over $N=1000$ data sets with no changepoints. For SUBSET we use the function for choosing thresholds provided by the author, which is based on Monte Carlo simulation. For Sparsified Binary Segmentation and Double CUSUM we use the default parameters when running the algorithms (except for setting $\phi = -1$ for the Double CUSUM algorithm), and use the default bootstrap procedures to select detection thresholds. 

For each method considered and each configuration of parameters and changepoint regimes, Table \ref{fig:tablemulti_small} displays the average Hausdorff distance and average absolute estimation error of $J$. The Double CUSUM and Sparsified Binary Segmentation methods are abbreviated as DC and SBS, respectively. For each configuration of parameters and changepoint regimes, the minimum value of each of the performance measures is indicated in boldface. In terms of average Hausdorff distance, Table \ref{fig:tablemulti_small} demonstrates that ESAC and SUBSET are the top performers across all sparsity regimes with comparable accuracy. SUBSET slightly outperforms ESAC when there are few changepoints ($J=2$), while the opposite is true when there are $J=5$ changepoints. Averaging the Hausdorff distance over all configurations, SUBSET is seen to slightly outperform ESAC. We believe this is due to ESAC using Narrowest-Over-Threshold choice of changepoints, which usually causes ESAC to use fewer observations to estimate changepoint locations. See Appendix \ref{appendixC} for a simulation where ESAC does not use Narrowest-Over-Threshold choice of changepoints locations, improving its performance. For estimating $J$, the number of changepoints, ESAC is the clear winner, with SUBSET in second place. 

Figure \ref{runningtime} displays the natural logarithm of the running times (in milliseconds) of the methods as functions of $n$ and $p$. For the left plot, we fix $p=100$, and for the right plot we fix $n=100$. When it comes to execution time, ESAC outperforms the competing methods by a significant margin for all considered values of $n$ and $p$. The running time of ESAC is smaller than that of the competitors by a factor seemingly constant in $n$ and $p$. When not applying a log transform to the running times (which is omitted for brevity), all methods are seen to have an approximately linear computational cost in both $n$ and $p$.

\input{tables/tablemulti/table_small}

\begin{figure}[ht]
\includegraphics[width=\textwidth]{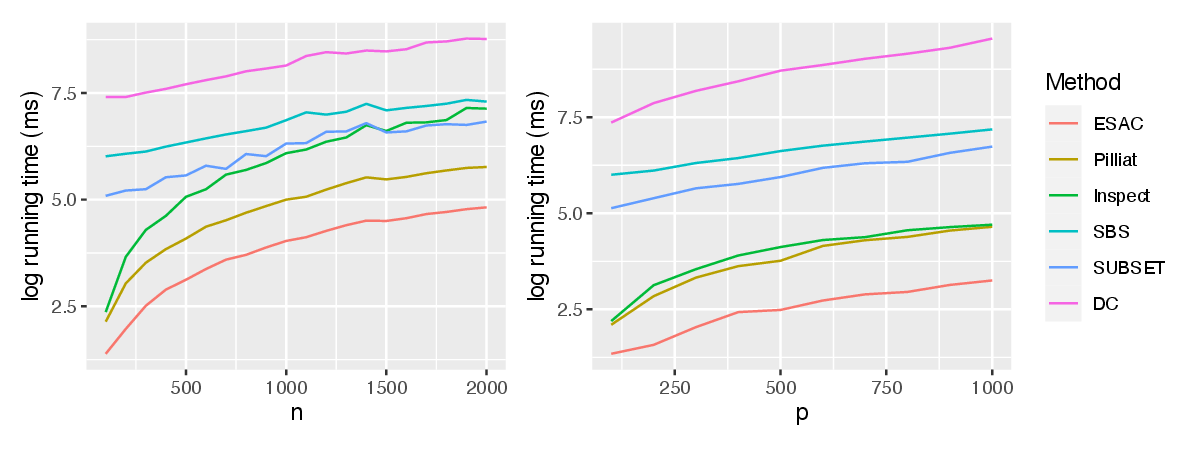}
\caption{Running times as functions of $n$ (left) and $p$ (right) on a logarithmic scale.}\label{runningtime}
\end{figure}

\subsection{Misspecified model}
ESAC is designed for data with isotropic Gaussian noise, which can be an unrealistic assumption in practice. We now investigate the empirical performance of ESAC and the competing methods in the single changepoint setting under other data generating mechanism than the model described in Section \ref{secproblemdesc}. With the changepoint location fixed at $\eta = \lceil n/5 \rceil = 40$, we consider two sparsity regimes, \textit{sparse} and \textit{dense}. We sample $k$ independently and uniformly from $\left\{1, \ldots, \lfloor \surd{(p\log n)}\rfloor\right\}$ in the sparse regime, and from $\left\{\lceil \surd{(p\log n)}\rceil \ldots, p\right\}$ in the dense regime. In both regimes, we take the change-vector $\theta$ to satisfy $\theta \propto (I_1, \ldots, I_k , 0, \ldots, 0)^{\T}$, where $I_1, \ldots, I_k$ are drawn independently and uniformly from $\{-1,1\}$. Furthermore, we scale $\theta$ such that $\Delta\phi^2 = ({n}/{5})\normm{\theta}_2^2 = 9 r(k)$.

Similar to the simulation study in \citet{wang_high_2018}, we consider the following data generating mechanisms. In model $M_0$ we take the noise vector $W_{v}$ to satisfy $W_v \sim \N_p(0, I)$ independently for $v \in [n]$. In models $M_{\text{Unif}}$ and $M_{\text{t}_{\text{d}}}$ we take $W_{i,v} \sim \text{Unif}(-\surd{3},\surd{3})$ and $\{d/(d-2)\}^{1/2} W_{i,v} \sim t_d$, respectively and independently for all $v \in[n]$ and $i\in[p]$, where $t_d$ denotes the Student t distribution with $d$ degrees of freedom. In model $M_{\text{cs, loc}}(\rho)$ we let the noise vectors $W_1, \ldots, W_n$ have short-ranged spatial correlation, taking $W_v \sim \N_p\left(0, \Sigma(\rho)\right)$ independently for all $v\in[n]$, where $\Sigma(\rho)_{j,k} = \rho^{|j-k|}$ for each $j,k\in[p]$. In the model $M_{\text{cs}}(\rho)$ we let the noise vectors $W_1, \ldots, W_n$ have global spatial correlation by taking $W_v \sim \N_p\left(0, \Delta(\rho)\right)$ independently for $v\in [n]$, where $\Delta(\rho) = (1-\rho) I_p + \rho /p I_p I_p^{\T}$. In the model $M_{\text{temp}}(\rho)$ we allow for temporal dependence between the noise vectors $W_1, \ldots, W_n$ by letting $W_1 = \widetilde{W}_1$ and $W_v = \surd{\rho} \widetilde{W}_v + \surd{(1-\rho)} W_{v-1}$ for $v = 2, \ldots, n$, where $\widetilde{W}_1, \ldots, \widetilde{W}_n \sim \N_p(0, I_p)$, independently. In the models 
$M_{\text{async}}$ and $M_{\text{gradual}}$ we allow for changes in the mean to occur asynchronous and gradual in time, respectively, with noise vectors $W_v \sim \N_p(0, I_p)$ independently for $v\in [n]$. In $M_{\text{async}}$, for each changepoint $\eta_j$, we randomly shift the position (in time) of the change in mean in the $i$th coordinate, where the shifts are drawn independently from $\text{Unif}(\eta_j - \lfloor \Delta_j/2\rfloor,\eta_j - \lfloor \Delta_j/2\rfloor +1, \ldots, \eta_j + \lfloor \Delta_j/2\rfloor)$. In $M_{\text{gradual}}$, for each changepoint $\eta_j$, any change in mean occurs linearly over time, starting at position $\eta_j - \lfloor \Delta_j/2\rfloor+1$ and ending at position $\eta_j + \lfloor \Delta_j/2 \rfloor+1$. For both the single and multiple changepoint settings, we have set $n=p=200$.

 Table \ref{tablesinglelocmisspec} displays the MSE of the competing methods using the same running parameters as in Section \ref{numericsingle}, based on $N=1000$ runs. Table \ref{tablesinglelocmisspec} indicates that ESAC (along with Inspect and SUBSET) is robust to model deviations in the form of light-tailed noise and short-ranged spatial correlation. With global spatial correlation, however, all methods degrade substantially in performance, with Inspect having a slight edge over the remaining methods. With autocorrelation, the performance of the methods also degrades markedly, with ESAC and SUBSET having a slight edge over the remaining methods. Lastly, ESAC and SUBSET seem to be slightly more robust to asynchronous and gradual occurrence of changepoints than the remaining methods.

\input{tables/tablesingleloc_misspec/table_even}

%% file: tables/tablesingleloc/table_even.tex
\begin{table}[ht] \centering
\caption{Single changepoint estimation}
\label{tablesingleloc}
\small
\begin{adjustbox}{width=\columnwidth}
\begin{tabular}{@{\extracolsep{1pt}} ccccc|ccccc|ccccc}
\hline
\multicolumn{5}{c|}{Parameters} & \multicolumn{5}{c|}{MSE} &\multicolumn{5}{c|}{Time in milliseconds} \\ \hline 
$n$ & $p$ & $k$ &  $\eta$ & $\phi$ & \text{ESAC} & \text{Inspect} & \text{SBS} & \text{SUBSET} & \text{DC} & \text{ESAC} & \text{Inspect} &\text{SBS} & \text{SUBSET}& \text{DC} \\
\hline \
200 & 100 & 1 & 40 & 1.40 & 10.4 & 25.3 & 83.0 & 54.2 & \textbf{9.8}  & \textbf{0.4}  & 2.1 & 9.7 & 1.7 & 12.9 \\
200 & 100 & 5 & 40 & 2.00 & 5.8 & 4.1 & 389.0 & \textbf{3.0}  & 15.3 & \textbf{0.3}  & 2.0 & 10.4 & 1.9 & 13.0 \\
200 & 100 & 24 & 40 & 1.90 & \textbf{96.5}  & 139.6 & 1495.8 & 251.1 & 953.4 & \textbf{0.2}  & 2.0 & 9.0 & 1.6 & 14.7 \\
200 & 100 & 100 & 40 & 1.90 & \textbf{95.1}  & 425.9 & 1520.6 & 250.8 & 2807.6 & \textbf{0.2}  & 2.0 & 9.1 & 1.6 & 11.7 \\
200 & 1000 & 1 & 40 & 1.52 & 6.9 & 105.8 & 29.5 & 33.1 & \textbf{6.5}  & \textbf{2.0}  & 40.8 & 84.5 & 12.0 & 120.6 \\
200 & 1000 & 10 & 40 & 2.93 & 5.1 & \textbf{0.8}  & 130.6 & 1.2 & 10.4 & \textbf{1.9}  & 40.8 & 84.0 & 12.6 & 122.2 \\
200 & 1000 & 73 & 40 & 3.37 & \textbf{4.6}  & 64.5 & 1478.2 & 8.2 & 276.1 & \textbf{1.9}  & 40.9 & 83.4 & 12.1 & 122.4 \\
200 & 1000 & 1000 & 40 & 3.37 & \textbf{3.5}  & 796.2 & 1534.7 & 7.1 & 317.2 & \textbf{2.0}  & 41.0 & 82.9 & 12.6 & 122.2 \\
200 & 5000 & 1 & 40 & 1.60 & 45.3 & 413.3 & \textbf{29.1}  & 81.4 & 154.2 & \textbf{11.7}  & 210.4 & 427.9 & 73.8 & 651.7 \\
200 & 5000 & 18 & 40 & 4.00 & 9.4 & \textbf{0.6}  & 65.7 & 3.3 & 94.7 & \textbf{11.6}  & 210.0 & 426.1 & 74.3 & 652.4 \\
200 & 5000 & 163 & 40 & 5.04 & \textbf{3.6}  & 60.3 & 1466.2 & \textbf{3.6}  & 7.2 & \textbf{11.8}  & 210.1 & 422.3 & 75.0 & 654.0 \\
200 & 5000 & 5000 & 40 & 5.04 & \textbf{4.4}  & 1453.9 & 1563.1 & \textbf{4.4}  & 5.3 & \textbf{11.9}  & 210.1 & 420.8 & 75.3 & 655.4 \\
500 & 100 & 1 & 100 & 0.92 & 55.4 & 97.9 & 120.8 & 216.4 & \textbf{54.6}  & \textbf{0.7}  & 5.4 & 16.2 & 3.9 & 26.3 \\
500 & 100 & 5 & 100 & 1.31 & 22.9 & 15.7 & 1060.2 & \textbf{12.6}  & 108.0 & \textbf{0.6}  & 5.3 & 15.5 & 3.7 & 24.9 \\
500 & 100 & 25 & 100 & 1.25 & \textbf{112.7}  & 323.1 & 9560.3 & 2150.5 & 6850.1 & \textbf{0.6}  & 5.3 & 16.0 & 3.6 & 25.3 \\
500 & 100 & 100 & 100 & 1.25 & \textbf{284.4}  & 1845.9 & 9768.9 & 1959.7 & 18386.6 & \textbf{0.6}  & 5.2 & 16.4 & 3.7 & 24.3 \\
500 & 1000 & 1 & 100 & 1.00 & \textbf{30.5}  & 217.7 & 79.3 & 190.5 & 31.0 & \textbf{5.2}  & 252.6 & 141.9 & 32.1 & 256.9 \\
500 & 1000 & 10 & 100 & 1.90 & 12.3 & 5.1 & 232.6 & \textbf{3.7}  & 84.3 & \textbf{5.4}  & 252.2 & 142.4 & 31.7 & 256.1 \\
500 & 1000 & 79 & 100 & 2.22 & \textbf{22.1}  & 122.5 & 9547.8 & 66.4 & 1725.7 & \textbf{5.6}  & 253.2 & 141.6 & 32.3 & 256.2 \\
500 & 1000 & 1000 & 100 & 2.22 & \textbf{15.4}  & 3895.9 & 9790.4 & 82.9 & 2401.7 & \textbf{5.7}  & 254.8 & 142.6 & 32.3 & 258.2 \\
500 & 5000 & 1 & 100 & 1.05 & \textbf{22.2}  & 1322.6 & 51.6 & 95.9 & 192.0 & \textbf{31.0}  & 1307.0 & 743.7 & 168.9 & 1536.6 \\
500 & 5000 & 18 & 100 & 2.58 & 7.9 & \textbf{1.8}  & 102.9 & 3.2 & 505.5 & \textbf{30.1}  & 1302.7 & 726.1 & 163.7 & 1434.2 \\
500 & 5000 & 177 & 100 & 3.32 & \textbf{11.2}  & 175.0 & 9438.2 & \textbf{11.2}  & 37.6 & \textbf{30.8}  & 1300.5 & 720.1 & 158.6 & 1420.4 \\
500 & 5000 & 5000 & 100 & 3.32 & \textbf{20.2}  & 8212.5 & 9799.4 & 33.7 & 27.2 & \textbf{31.0}  & 1301.1 & 717.8 & 158.3 & 1415.7 \\
\hline \multicolumn{5}{c|}{Average MSE}
 & \textbf{37.8}  & 821.9 & 2889.1 & 230.3 & 1460.9 \\
\hline \\[-1.8ex]
\end{tabular}
\end{adjustbox}
\end{table}

%% file: tables/tablemulti/table_small.tex
 \begin{table}[h] \centering
\caption{Multiple change-point estimation}
\label{fig:tablemulti_small}
\small
\begin{adjustbox}{width=\columnwidth}
\begin{tabular}{@{\extracolsep{1pt}} cccc|cccccc|cccccc}
\hline
\multicolumn{4}{c|}{Parameters} & \multicolumn{6}{c|}{Hausdorff distance} &\multicolumn{6}{c|}{$\left | \widehat{J}-J \right |$}  \\ \hline 
        $n$ & $p$ & Sparsity & J & \text{ESAC} & \text{Pilliat} & \text{Inspect} & \text{SBS} & \text{SUBSET} & \text{DC} & \text{ESAC} & \text{Pilliat} & \text{Inspect} & \text{SBS} & \text{SUBSET} & \text{DC}  \\  \hline
        200 & 100 & - & 0 & - & - & - & - & - & - & \textbf{0.00} & \textbf{0.00} & \textbf{0.00} & 0.04 & \textbf{0.00} & 0.04  \\ 
        200 & 100 & Dense & 2 & \textbf{5.27} & 23.05 & 21.61 & 69.92 & 6.52 & 76.37 & \textbf{0.06} & 0.48 & 0.29 & 1.03 & 0.11 & 1.08  \\ 
        200 & 100 & Sparse & 2 & \textbf{1.43} & 12.54 & 7.92 & 49.19 & 1.67 & 14.92 & \textbf{0.01} & 0.25 & 0.10 & 0.70 & 0.04 & 0.26  \\ 
        200 & 100 & Mixed & 2 & 4.60 & 18.57 & 14.36 & 61.15 & \textbf{3.74} & 52.93 & \textbf{0.05} & 0.37 & 0.19 & 0.87 & 0.06 & 0.74  \\ 
        200 & 100 & Dense & 5 & \textbf{5.17} & 23.74 & 16.27 & 67.29 & 6.42 & 55.71 & \textbf{0.14} & 1.53 & 0.54 & 3.12 & 0.33 & 2.76  \\ 
        200 & 100 & Sparse & 5 & \textbf{1.36} & 13.99 & 6.38 & 57.18 & 2.42 & 23.19 & \textbf{0.02} & 0.78 & 0.23 & 2.79 & 0.20 & 1.36  \\ 
        200 & 100 & Mixed & 5 & \textbf{4.04} & 18.77 & 12.56 & 60.43 & 4.90 & 41.31 & \textbf{0.10} & 1.20 & 0.42 & 3.01 & 0.28 & 2.16  \\ 
        200 & 1000 & - & 0 & - & - & - & - & - & - & \textbf{0.00} & \textbf{0.00} & \textbf{0.00} & 0.34 & 0.02 & 0.02  \\ 
        200 & 1000 & Dense & 2 & 1.45 & 13.39 & 9.49 & 55.12 & \textbf{1.08} & 84.98 & \textbf{0.01} & 0.28 & 0.12 & 0.72 & 0.03 & 1.23  \\ 
        200 & 1000 & Sparse & 2 & 0.83 & 9.07 & 9.83 & 49.72 & \textbf{0.75} & 20.38 & \textbf{0.00} & 0.18 & 0.15 & 0.62 & 0.03 & 0.29  \\ 
        200 & 1000 & Mixed & 2 & 1.78 & 11.78 & 9.59 & 50.51 & \textbf{1.48} & 53.23 & \textbf{0.01} & 0.24 & 0.13 & 0.69 & 0.04 & 0.82  \\ 
        200 & 1000 & Dense & 5 & \textbf{1.52} & 15.81 & 9.33 & 57.33 & 1.56 & 60.79 & \textbf{0.03} & 1.02 & 0.26 & 2.96 & 0.17 & 3.06  \\ 
        200 & 1000 & Sparse & 5 & \textbf{0.69} & 11.34 & 9.16 & 54.20 & 1.52 & 25.80 & \textbf{0.00} & 0.64 & 0.36 & 2.68 & 0.17 & 1.43  \\ 
        200 & 1000 & Mixed & 5 & \textbf{1.19} & 13.93 & 9.45 & 58.26 & 1.72 & 45.16 & \textbf{0.02} & 0.81 & 0.33 & 2.86 & 0.19 & 2.38  \\ 
        200 & 5000 & - & 0 & - & - & - & - & - & - & \textbf{0.00} & \textbf{0.00} & \textbf{0.00} & 1.96 & \textbf{0.00} & \textbf{0.00}  \\ 
        200 & 5000 & Dense & 2 & 0.97 & 12.99 & 13.25 & 54.75 & \textbf{0.60} & 124.34 & \textbf{0.00} & 0.29 & 0.17 & 0.53 & 0.02 & 1.68  \\ 
        200 & 5000 & Sparse & 2 & 0.76 & 9.66 & 15.81 & 58.96 & \textbf{0.51} & 70.83 & \textbf{0.00} & 0.20 & 0.24 & 0.56 & 0.03 & 1.04  \\ 
        200 & 5000 & Mixed & 2 & 0.98 & 11.36 & 13.76 & 55.92 & \textbf{0.78} & 96.95 & \textbf{0.00} & 0.25 & 0.21 & 0.50 & 0.03 & 1.36  \\ 
        200 & 5000 & Dense & 5 & \textbf{0.89} & 14.59 & 12.97 & 50.33 & 1.37 & 83.64 & \textbf{0.01} & 0.94 & 0.45 & 2.45 & 0.17 & 3.63  \\ 
        200 & 5000 & Sparse & 5 & \textbf{0.53} & 11.70 & 13.50 & 54.00 & 1.29 & 54.33 & \textbf{0.00} & 0.74 & 0.55 & 2.53 & 0.18 & 2.71  \\ 
        200 & 5000 & Mixed & 5 & \textbf{0.78} & 13.57 & 13.25 & 53.72 & 1.41 & 69.94 & \textbf{0.00} & 0.86 & 0.50 & 2.50 & 0.17 & 3.19    \\ \hline
        \multicolumn{4}{c|}{Average} & \textbf{1.90} & 14.43 & 12.13 & 56.55 & 2.21 & 58.60 & \textbf{0.02} & 0.61 & 0.29 & 1.73 & 0.12 & 1.73 \\ 
\hline
\end{tabular}
\end{adjustbox}
\end{table}

%% file: tables/tablesingleloc_misspec/table_even.tex
\begin{table}[h] \centering
\caption{Single changepoint estimation under misspecified models}
\label{tablesinglelocmisspec}
\small
\begin{tabular}{@{\extracolsep{1pt}} cc|ccccc}
\hline
\multicolumn{2}{c|}{Parameters} & \multicolumn{5}{c|}{MSE} \\ \hline 
Model & Sparsity & \text{ESAC} & \text{Inspect} & \text{SBS} & \text{SUBSET} & \text{DC} \\
\hline \
$M$ & Sparse & 1.2 & \textbf{1.0}  & 506.1 & 1.1 & 26.8 \\
$M$ & Dense & \textbf{1.6}  & 84.5 & 1507.9 & \textbf{1.6}  & 1101.2 \\
$M_{\text{Unif}}$ & Sparse & 1.0 & 1.2 & 666.7 & \textbf{0.9}  & 17.7 \\
$M_{\text{Unif}}$ & Dense & \textbf{7.5}  & 141.3 & 1506.1 & 23.4 & 982.4 \\
$M_{t_{3}}$ & Sparse & 1373.5 & \textbf{405.7}  & 3014.5 & 1379.4 & 3342.1 \\
$M_{t_{3}}$ & Dense & 1561.6 & \textbf{1184.0}  & 12195.2 & 1642.9 & 8132.5 \\
$M_{t_{10}}$ & Sparse & 2.3 & 1.5 & 649.5 & \textbf{1.4}  & 64.3 \\
$M_{t_{10}}$ & Dense & \textbf{2.1}  & 54.6 & 2045.5 & \textbf{2.1}  & 1256.2 \\
$M_{\text{cs, loc}}(\rho = 0.1)$ & Sparse & 1.1 & \textbf{0.7}  & 497.7 & 1.0 & 21.6 \\
$M_{\text{cs, loc}}(\rho = 0.1)$ & Dense & \textbf{1.9}  & 69.3 & 1501.3 & \textbf{1.9}  & 1108.3 \\
$M_{\text{cs, loc}}(\rho = 0.4)$ & Sparse & 1.3 & \textbf{1.2}  & 508.0 & \textbf{1.2}  & 17.9 \\
$M_{\text{cs, loc}}(\rho = 0.4)$ & Dense & \textbf{2.7}  & 116.6 & 1505.7 & \textbf{2.7}  & 1347.6 \\
$M_{\text{cs}}(\rho = 0.1)$ & Sparse & 126.7 & \textbf{2.0}  & 524.3 & 82.2 & 20.9 \\
$M_{\text{cs}}(\rho = 0.1)$ & Dense & 281.4 & \textbf{114.4}  & 1516.1 & 281.4 & 1111.4 \\
$M_{\text{cs}}(\rho = 0.4)$ & Sparse & 4087.5 & 66.2 & 540.1 & 3997.6 & \textbf{50.7}  \\
$M_{\text{cs}}(\rho = 0.4)$ & Dense & 5473.6 & \textbf{1296.8}  & 1503.6 & 5473.6 & 2072.9 \\
$M_{\text{AR}}(\rho = 0.1)$ & Sparse & 148.0 & \textbf{77.1}  & 194.2 & 148.0 & 175.8 \\
$M_{\text{AR}}(\rho = 0.1)$ & Dense & \textbf{40.6}  & 461.7 & 3023.6 & \textbf{40.6}  & 2557.8 \\
$M_{\text{AR}}(\rho = 0.4)$ & Sparse & \textbf{979.9}  & 1648.6 & 1270.0 & \textbf{979.9}  & 2010.2 \\
$M_{\text{AR}}(\rho = 0.4)$ & Dense & \textbf{1274.3}  & 1994.0 & 2552.5 & \textbf{1274.3}  & 4337.5 \\
$M_{\text{async}}$ & Sparse & 83.1 & 99.0 & 603.8 & \textbf{81.2}  & 233.3 \\
$M_{\text{async}}$ & Dense & \textbf{75.8}  & 288.3 & 1521.7 & 79.4 & 1644.7 \\
$M_{\text{grad}}$ & Sparse & 50.3 & 54.5 & 794.5 & \textbf{49.2}  & 223.8 \\
$M_{\text{grad}}$ & Dense & \textbf{67.7}  & 231.3 & 1524.4 & 92.2 & 2002.1 \\
$M_{\text{grad}}$ & Dense & \textbf{67.7}  & 231.3 & 1524.4 & 92.2 & 2002.1 \\
\hline \\[-1.8ex]
\end{tabular}
\end{table}

%% file: sections/dataexample.tex
\section{Real data example}\label{dataex}
To illustrate how ESAC can be applied in practice, we examine raw sensor data from a Swedish hydro power plant. The data consists of measurements from 20 sensors taken every minute for 1800 minutes, so that $p=20$ and $n=1800$. The sensors measure the magnitude of movements and vibrations (the latter measured at 1-10 and 10-1000 Hz bands) at various locations along the shaft connecting the turbine and the generator. During the 1800 minutes we consider, the mode of operation changes several times, detailed in Table \ref{hydrpowermode}. We take these changes of operation mode as the ground truth regarding the number of changepoints and their locations. 
\begin{table}[ht] \centering
\caption{Operation modes of the hydro power plant}
\label{hydrpowermode}
\small
\begin{adjustbox}{}
\begin{tabular}{@{\extracolsep{1pt}} c|c}
\hline
Time period & Operation mode \\ \hline 
1 -- 529 & running\\
530 -- 537 & stopping \\
538 -- 1307 & off \\
1308 -- 1310 & starting \\
1311 -- 2000 & running \\
\hline \\[-1.8ex]
\end{tabular}
\end{adjustbox}
\end{table}

The data generating mechanism of the data is undeniably in violation of several underlying assumptions of ESAC. Importantly, the data are highly cross-correlated and auto-regressive. Moreover, the measurements in the data set are influenced by contextual variables such as power output, guide vane opening, and other (human controlled) running conditions in a complex manner. This dependence on contextual variables should ideally be modeled carefully, although such modeling is outside the scope of this paper. As a remedy, we instead transform the observed data by right multiplying each observed data point $X_i$ by $\widehat{\Sigma}^{-1/2}$. Here, $\widehat{\Sigma}$ is the estimated variance-covariance matrix of $X_i$, estimated from an independent data set with 5992 observations, in which running conditions are stable (i.e. with no changes in operation mode). Moreover, we choose the penalty function $\gamma(t)$ empirically as described in Appendix \ref{appendixA}, using false probability rate $\epsilon = 0.01$ and letting each of the $N=1000$ Monte Carlo samples $X^{(j)}$ have independent entries following a $t_5$ distribution. This choice of penalty function ensures that ESAC is rather conservative in declaring changepoints. 

The Monte Carlo simulation for generating the penalty function $\lambda(t)$ took 2 minutes and 2 seconds. Applying ESAC to the data took 0.035 seconds, resulting in six estimated changepoints, at locations 531, 533, 974, 1067, 1308, and 1330. Figure \ref{dataexfig} displays the 20 transformed sensor measurements over the sampling period, with estimated changepoint locations indicated by red ticks on the x axis. The grey rectangle in the plot indicate the times at which the plant is either starting, stopping, or off. From the Figure, we clearly see that the first, second and fifth and sixth identified changepoint are associated with a change in operation mode of the plant. Interestingly, the other two changepoints, located at time points 974, 1067, are not associated with a change in running conditions. These changepoints are likely declared by ESAC due to the sudden shift in the yellow curves occurring at time 974 and reverting back again at time 1067. 

\begin{figure}
\includegraphics[width=\textwidth]{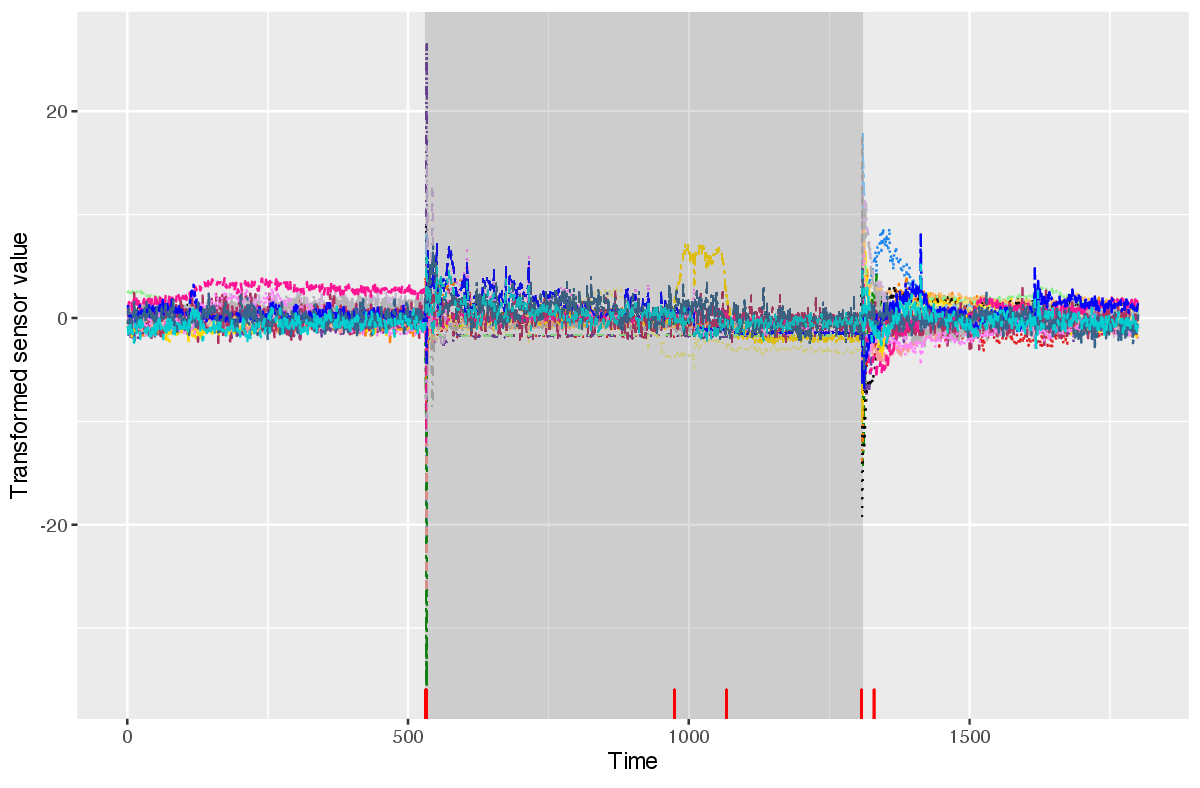}
\caption{Vibration sensor measurements with detected changepoints indicated by red ticks. Grey areas indicate time intervals in which the plant is starting, stopping or off.}\label{dataexfig}
\end{figure}

%% file: sections/proofs.tex
\subsection{Proofs of main results}\label{secproofs}

\begin{proof}[Proof of Proposition \ref{lemmalemmalemma}]
Set $c_1 = 6+2\log(8/\epsilon)/\log(2)$, $c_2 =  12+ 2(\log(1/\epsilon))^{1/2} +2 \log(1/\epsilon)$ and $\gamma_0 = 9 (c_1 + c_1^{1/2}\exp(-1)) + c_2$. Note first that $I$ has cardinality no larger than $n^3$. By a union bound, it thus suffices to show that $\PP(S_{\gamma, (s,e]} >0) \leq \eps n^{-3}$ for any $(s,e) \subseteq I$. 

So fix any $(s,e) \subseteq I$. Let $t \in \mathcal{T}\setminus\{p\}$ (the case $t=p$ is handled later), and fix $x_t>0$ (to be specified shortly). Since $(s,e)$ does not contain any changepoint, we must have ${C}^{v}_{(s,e]} (i) \iid \N(0,1)$ for all $i \in [p]$. By Lemma \ref{lemma6} we have that
\begin{align}
 \sum_{i =1}^p \left\{ {C}^{v}_{(s,e]} (i)^2  - \nu_{a(t)} \right \} \ind {\left\{\left|{C}^{v}_{(s,e]}(i)\right| > a(t) \right\}} &\geq 9 \left [  \left\{ p e^{-a(t)^2/2} x_t\right\}^{1/2} + x_t \right] \label{l21e3}
\end{align}
with probability at most $e^{-x_t}$. Now set $x_t = c_1 \left\{  \frac{p \log^2(n)}{t^2} \wedge r(t)  \right\}$ for all $t$. Then,
\begin{align}
\sum_{t\in \mathcal{T}\setminus\{p\}} e^{-x_t} &\leq \sum_{t\in \mathcal{T}\setminus\{p\}}   \exp\left\{-c_1\frac{p \log^2(n)}{t^2}\right\} + \sum_{t\in \mathcal{T}\setminus\{p\}} \exp\left\{ -c_1r(t)\right\}.
\end{align}
For the first sum, we have
\begin{align}
\sum_{t\in \mathcal{T}\setminus\{p\}}   \exp\left\{-c_1\frac{p \log^2(n)}{t^2}\right\} &\leq \sum_{k=0}^{\infty} \exp\left\{-c_1{\log(n)} 4^{k}\right\} \\
&= \sum_{k=0}^{\infty}\left( \frac{1}{n^{c_1}} \right)^{4^k}\\
&\leq n^{-c_1} + n^{-c_1}\sum_{k=1}^{\infty}\left( \frac{1}{n^{c_1}} \right)^{3k}\\
&=2 n^{-c_1}.
\end{align}
For the second sum, noting that $c_1 r(t) = c_1 \left\{ t \log \left( \frac{ep\log n}{t^2}\right) \vee \log n\right\} \geq (c_1/2) t \log \left( \frac{ep\log n}{t^2}\right) + (c_1/2)\log n$, we have
\begin{align}
\sum_{t\in \mathcal{T}\setminus\{p\}} \exp\left\{ -c_1r(t)\right\} &\leq n^{-c_1/2} \exp(-c_1/2)  
\sum_{t\in \mathcal{T}\setminus\{p\}}  \left(   \frac{t^2}{ep\log n}   \right)^{c_1t/2}\\
&\leq n^{-c_1/2} \exp(-c_1/2) \left( 1 + \sum_{k=1}^{\infty} 4^{-c_1t/2} \right)\\
&\leq 2n^{-c_1/2}.
\end{align}
With this choice of $x_t$, we thus have that
\begin{align}
\sum_{t\in \mathcal{T}\setminus\{p\}} e^{-x_t} &\leq 4n^{-c_1/2},
\end{align}
using that $c_1>1$. Moreover, using that $a^2(t) = 4\log\left( \frac{ep\log n}{t^2}  \right)$, we have that
\begin{align}
9 \left [  \left\{ p e^{-a^2(t)/2} x_t\right\}^{1/2} + x_t \right] &= 9 \left [  \left\{ p \frac{t^4}{e^2 p^2 \log^2 n} x_t\right\}^{1/2} + x_t \right] \\
&\leq 9 \left \{ \frac{tc_1^{1/2}}{{e}}  + c_1 r(t)\right\}\\
&\leq 9 \left( c_1^{1/2}\exp(-1) + c_1 \right) r(t),
\end{align}
where we used that $x_t \leq c_1r(t)$ and $x_t \leq {c_1p \log^2(n)}/{t^2}$, as well as the fact that $t \leq r(t)$ whenever $t \leq (p \log n)^{1/2}$. Recalling that $\gamma(t) = \gamma_0 r(t)$, since $\gamma_0 >  9 \left (c_1 + c_1^{1/2}\exp(-1) \right)$, a union bound gives
\begin{align}
& \ \PP \left( \exists t \in \mathcal{T}\setminus\{p\} \ ; \ S_{\gamma, (s,e]}(t) \geq 0 \right)\\
= & \ \PP \left[ \exists t \in \mathcal{T}\setminus\{p\} \ ; \ \sum_{i =1 }^p \left\{ {C}^{v}_{(s,e]} (i)^2  - \nu_{a(t)} \right \} \ind{\left\{|{C}^{v}_{(s,e]}(i)| > a(t) \right\}} \geq \gamma_0 r(t) \right]\\
\leq & 4n^{-c_1/2} \\
\leq & 4n^{-3-\log(8/\epsilon)/\log(2)}\\
\leq & 4n^{-3} \exp\left( -\log(8/\epsilon)\log (n) / \log(2)\right)\\
\leq & n^{-3} \epsilon/2,
\end{align}
where we in the last inequality used that $n\geq 2$. 

Now consider the case where $t = p$. If $p \leq ({p\log n})^{1/2}$, then similarly as above we have that
\begin{align}
\ \PP \left(  \ S_{\gamma, (s,e]}(p) \geq 0 \right) &\leq n^{-3} \epsilon/2.
\end{align}

If we instead have $p > (p \log n)^{1/2}$ (in which case $a(p) = 0$ and $\nu_{a(p)} = 1$), then
\begin{align}
\sum_{i =1 }^p \left\{ {C}^{v}_{(s,e]} (i)^2  - \nu_{a(p)} \right \} \ind{\left\{|{C}^{v}_{(s,e]}(i)| > a(p) \right\}} &= \sum_{i =1 }^p  {C}^{v}_{(s,e]} (i)^2 - p .  
\end{align}
As $\sum_{i =1 }^p {C}^{v}_{(s,e]} (i)^2   \sim \chi^2_{p}$, we obtain from Lemma \ref{lemma12} that
\begin{align}
\PP \left\{  \sum_{i =1 }^p  {C}^{v}_{(s,e]} (i)^2 - p   > 2(p \log (2n^3/\epsilon))^{1/2} + 2 \log(2n^3/\epsilon) \right\} &\leq n^{-3} \epsilon /2. 
\end{align}
Now, 
\begin{align}
2(p \log (2n^3/\epsilon))^{1/2} + 2 \log(2n^3/\epsilon) &\leq 2(p \log (n^4/\epsilon))^{1/2} + 2 \log(n^4/\epsilon) \\
&\leq 4 (p\log n)^{1/2} + 2(p \log(1/\epsilon))^{1/2} + 8\log n -2\log (\epsilon) \\
&\leq r(p) \left(12+ 2(\log(1/\epsilon))^{1/2} +2 \log(1/\epsilon) \right)\\
&= c_2 r(p)\\
&< \gamma_0 r(p),
\end{align}
using that $n\geq 2$, $r(p)\geq 1$, and $r(p) = (p\log n)^{1/2} \geq \log n$ whenever $p \geq (p\log n)^{1/2}$. Hence, 
\begin{align}
\PP \left\{  \sum_{i =1 }^p  {C}^{v}_{(s,e]} (i)^2 - p   > \gamma_0 r(p) \right\} &\leq n^{-3} \epsilon /2. 
\end{align}
We conclude that
\begin{align}
    \PP\left( \underset{(s,e) \in I}{\max} \ S_{\gamma, (s,e]} >0   \right) &\leq n^3 \PP \left( \exists t \in \mathcal{T} \ ; \ S_{\gamma, (s,e]}(t) \geq 0 \right) \\
    &\leq n^3 \left( n^{-3}\epsilon/2 + n^{-3}\epsilon/2\right) \\
    &= \epsilon,
\end{align}
and we are done. 
\end{proof}
\begin{proof}[Proof of Theorem \ref{locprop}]
Let $\lambda_0 \geq 63$, $\lambda(t) = \lambda_0 r(t)$, $C_1 = 2 \left\{ 2 \left( 4\lambda_0 + 242\right)^{1/2}+ 2\lambda_0 +123 \right\}$ and $C_0 >C_1$. Let $\widehat{\eta} \in \underset{0<v<n}{\arg \max} \ S^v_{\lambda}$, where $S^v_{\lambda}$ is defined as in \eqref{sdef}, and let $S^v(t)$ be defined as in \eqref{stdef}. 
Let the CUSUM transformation $T^v_{(s,e]} (\cdot ) $ be defined as in \eqref{c2def}, and for ease of notation, let $T^v (\cdot ) = T^v_{(0,n]} (\cdot )$ . 
Let $\mathcal{K} = \left\{   i \ ; \ \mu_{i, \eta+1} - \mu_{i, \eta} \neq 0 \right\}$ denote the set of coordinates for which there is a change in mean, and for any $0<v<n$ let $\beta_v= \sum_{i \in \mathcal{K}} \left \{ T^{\eta}(\mu\roww{i})^2 - T^v (\mu\roww{i})^2 \right\}$. Let $\overline{k}$ denote the smallest element in $\mathcal{T}$ such that $\overline{k} \geq k$. We may without loss of generality take $\sigma=1$, as we can otherwise normalize the data matrix $X$ and replace the squared norm of the change in mean $\phi^2$ by $\phi^2 / \sigma^2$.

Consider the event $\mathcal{E}= \mathcal{E}_1 \cap \mathcal{E}_2 \cap \mathcal{E}_3 \cap \mathcal{E}_4$ as defined in Lemma \ref{lemma15}, for which we know that $\PP\left(\mathcal{E}\right) \geq 1-\frac{1}{n}$. On the event $\mathcal{E}$, we will show that any $0<v<n$ such that $\left | v - \eta \right| > C_1 {h(k)}/{\phi^2}$ must satisfy $S^{\eta}_{\lambda} > S^{v}_{\lambda}$, which implies that $\left | \widehat{\eta}_{\lambda} - \eta \right| \leq C_1 {h(k)}/{\phi^2}$.

Fix some $0<v<n$ and let $t^* \in \underset{t \in \mathcal{T}}{\arg \max} \ S^{v}_{\lambda}(t)$, so that $S^{v}_{\lambda} = S^{v}_{\lambda} (t^*)$. We claim that 
\begin{align}
S^{\eta}_{\lambda} - S^{v}_{\lambda} \geq  \beta_v   - 2\left\{{2 \beta_v h(k)}\right\}^{1/2} - (119 + 2\lambda_0) h(k)\label{p1claim}.
\end{align}
To see this, suppose first that $k \leq (p\log n)^{1/2}$. We have that
\begin{align}
S^{\eta}_{\lambda} - S^{v}_{\lambda} &\geq S^{\eta}_{\lambda}(\overline{k}) - S^{v}_{\lambda}(t^*)\\
= & \sum_{i=1 }^p \left [ \left( {C}^{\eta} (i)^2  - \nu_{a(\overline{k})} \right ) \ind{\left\{\left|{C}^{\eta}(i)\right| > a(\overline{k}) \right\}} - \left( {C}^{v} (i)^2  - \nu_{a(t^*)} \right ) \ind{\left\{\left|{C}^{v}(i)\right| > a(t^*) \right\}} \right] \notag\\
&- \lambda_0 r(\overline{k}) + \lambda_0 r(t^*).
\intertext{For any $x \in \RR$ and any $t\in \mathcal{T}$, we have $x^2 - \nu_{a(t)} \leq \left( x^2  - \nu_{a(t)} \right ) \ind{\left\{\left|x\right| > a(t) \right\}} \leq x^2$, and so }
S^{\eta}_{\lambda} - S^{v}_{\lambda} \geq & \sum_{i \in \mathcal{K} } \left[ {C}^{\eta} (i)^2  - {C}^{v} (i)^2  \right ] - k \nu_{a(\overline{k})}  \\
&+ \sum_{i \in [p]\setminus \mathcal{K} } \left [ \left( {C}^{\eta} (i)^2  - \nu_{a(\overline{k})} \right ) \ind{\left\{\left|{C}^{\eta}(i)\right| > a(\overline{k}) \right\}} - \left( {C}^{v} (i)^2  - \nu_{a(t^*)} \right ) \ind{\left\{\left|C^{v}(i)\right| > a(t^*) \right\}} \right] \notag \\
&- \lambda_0 r(\overline{k}) +\lambda_0 r(t^*).
\intertext{On the event $\mathcal{E}_1 \cap \mathcal{E}_2 \cap \mathcal{E}_3 \supseteq \mathcal{E}$ we therefore have that}
S^{\eta}_{\lambda} - S^{v}_{\lambda} \geq & \beta_{v}   - 2 \left\{ 2\beta_v \log n\right\}^{1/2} - 16 r(k)  - 35 r(\overline{k}) \notag \\
&- (63 - \lambda_0) r(t^*) - \lambda_0 r(\overline{k}) -  k \nu_{a(\overline{k})}\\
\geq & \beta_v   - 2 \left \{ 2 \beta_v r(k)\right\}^{1/2} - (92 + 2\lambda_0) r(k)\label{locprope1}, 
\end{align}
where we have used that $\lambda_0 \geq 63$, $\log n \leq r(k)$, $r(\overline{k}) \leq 2 r(k)$, and for $k \leq (p\log n)^{1/2}$, we have $k \nu_{a(\overline{k})} \leq k\left( 2 + a^2(\overline{k})\right) \leq 2k + ka^2(k) \leq 6 r(k)$. Since $r(k) \leq h(k)$ for all $k \in [p]$, the claim \eqref{p1claim} holds whenever $k \leq (p\log n)^{1/2}$.

Now suppose $k > (p \log n)^{1/2}$. By the definition of $\mathcal{E}_4$, we have that
$S^{v}_{\lambda}(t^*) - S^v_{\lambda}(p) \leq 5h(p) + 63 r(t^*) - \lambda_0 r(t^*) + \lambda_0 r(p)$. Hence, 
\begin{align}
S^{\eta}_{\lambda} - S^{v}_{\lambda} &\geq S^{\eta}_{\lambda}(p) - S^{v}_{\lambda}(t^*)\\
&= S^{\eta}_{\lambda}(p) - S^{v}_{\lambda}(p) + S^v_{\lambda}(p) - S^v_{\lambda}(t^*)\\
&\geq S_{\lambda}^{\eta}(p) - S^{v}_{\lambda}(p) - 5h(p) - 63 r(t^*) + \lambda_0 r(t^*) - \lambda_0 r(p)\\
&= \sum_{i=1}^p \left\{      C^{\eta}(i)^2 - C^{v}(i)^2\right\} - 5h(p) - 63 r(t^*) + \lambda_0 r(t^*) - \lambda_0 r(p).
\intertext{On the event $\mathcal{E}$ we thus have that}
S^{\eta}_{\lambda} - S^{v}_{\lambda} &\geq \beta_v -2 \left \{2\beta_v \log n\right\}^{1/2} - 16 r(p) - 63 r(p) - 35 r(p) - 5h(p) \\&- 63 r(t^*) + \lambda_0 r(t^*) - \lambda_0 r(p)\\
&\geq \beta_v -2\left\{ \beta_v h(p)\right\}^{1/2} - (119 + \lambda_0) h(p),
\end{align}
where we in the last inequality used that $r(k)\leq h(k)$ for all $k$ and $\lambda_0 \geq 63$. Hence \eqref{p1claim} holds whenever $k > (p\log n)^{1/2}$. 
Solving the quadratic inequality \eqref{p1claim} with respect to $\beta_v$, we obtain that $S^{\eta}_{\lambda} - S^v_{\lambda} >0$ if 
\begin{align}
\beta_v > \left\{ 2 \left( 4\lambda_0 + 242\right)^{1/2} + 2\lambda_0 +123 \right\} h(k)\label{quadineq}.
\end{align}
Without loss of generality we may assume $v \geq \eta$ (the converse case is similar). By Lemma \ref{lemma4NOT} we have that 
\begin{align}
\beta_v &= \sum_{i \in \mathcal{K}} \left\{ T^{\eta}(\mu\roww{i})^2 - T^v (\mu\roww{i})^2 \right\}\\
&= \frac{|v - \eta| \eta}{|v - \eta| + \eta} \phi^2 \\
&\geq \frac{1}{2} \min\left( \left| v - \eta \right| \ , \ \eta \right) {\phi^2},
\end{align}
and therefore \eqref{quadineq} is satisfied if 
\begin{align}
\min\left( \left| v - \eta \right| \ , \ \eta \right) &> 2 \left\{ 2 \left(4\lambda_0 + 242\right)^{1/2}+ 2\lambda_0 +123 \right\} \frac{h(k)}{\phi^2}\\
&= C_1 \frac{h(k)}{\phi^2} \label{tmpeq}.
\end{align}
By the assumption $C_0 > C_1$, $\eta$ is strictly larger than the right hand side of \eqref{tmpeq}. Therefore \eqref{quadineq} is satisfied if 
\begin{align}
 \left| v - \eta \right| &> C_1 \frac{h(k)}{\phi^2}.
\end{align}
Hence, if $|\eta - v| > C_1 {h(k)}/{\phi^2}$, we must have $S^{\eta}_{\lambda} > S^v_{\lambda}$, and the proof is complete.
\end{proof}\n

 \begin{proof}[Proof of Theorem \ref{adaptivetheorem}]
Let $\gamma_0 \geq 82$, $\gamma(t) = \gamma_0 r(t)$ and define $C_1 = 32\left\{ \gamma_0 + 170 + 8 \left( 2\gamma_0 + 276\right)^{1/2}\right\}$ and $C_0 = 2C_1$. We may without loss of generality take $\sigma=1$, as we can otherwise normalize the data matrix $X$ and replace the squared norm of the change in mean $\phi^2$ by $\phi^2 / \sigma^2$. Let $\mathcal{M} =  \{(s_m, e_m] \ ; \ m\in [M]\}$ denote the (enumerated) collection of seeded intervals generated by Algorithm \ref{alg:seededintervals}. In the following, we will use the name ESAC to refer to either Algorithm \ref{alg:cap} or \ref{alg:cap3}. We work on the event $\mathcal{E} = \mathcal{E}_5 \cap \mathcal{E}_6$ as defined in Lemma \ref{lemma20}, for which we know that $\PP(\mathcal{E}) \geq 1-1/n$. The proof goes as follows. In step 1 we show that each changepoint $\eta_j$ will be detected using a seeded interval with certain properties. In step 2, by an inductive argument, we show that ESAC detects all changepoints within the given error-rate.

\textbf{Step 1}. We first claim that, for each $j$ in $[J]$, there exists a seeded interval $(s_{\widetilde{m}},e_{\widetilde{m}}] = (v-l, v+l] \in \mathcal{M}$ such that the following holds
\begin{enumerate}[label=(P\arabic*)]
    \item $ C_1 r(k_j) /(4\phi^2_j) \leq l \leq C_1 r(k_j) /(2\phi^2_j) \vee 1$; \label{it1}
    \item $|\eta_j - v| \leq l/2$;\label{it5}
    \item $s_{\widetilde{m}} \geq \eta_j - (\Delta_j /2 \vee 1)$;\label{it2}
    \item $e_{\widetilde{m}} \leq \eta_j + (\Delta_j /2 \vee 1)$;\label{it3}
    \item $S_{\gamma, (s_{\widetilde{m}}, e_{\widetilde{m}}]}^v \geq0$.\label{it4}
\end{enumerate}

To see this, fix any $j\in[J]$, and let $h = C_1 r(k_j) /(2\phi^2_j)$. Now let $(s_{\widetilde{m}},e_{\widetilde{m}}] = (v-l,v+l]$ denote the seeded interval from Lemma \ref{lemma14}. Then properties \ref{it1} and \ref{it5} follow immediately. Moreover, as ${\phi_j^2 \Delta_j}\geq C_0 r({k}_j)$ (by the signal-to-noise ratio assumption \eqref{SNR} in Theorem \ref{adaptivetheorem}) and $C_0\geq 2C_1$, we have $h \leq \Delta_j/4$. The properties \ref{it2} and \ref{it3} then follow from Lemma \ref{lemma14}. To show the last property \ref{it4}, observe first that
\begin{align}
S^v_{\gamma, (s_{\widetilde{m}},e_{\widetilde{m}}]} \geq \beta_{(s_{\widetilde{m}},e_{\widetilde{m}}]}^v - 8\left\{2\beta_{(s_{\widetilde{m}},e_{\widetilde{m}}]}^v r(k_j)\right\}^{1/2} - \left( \gamma_0 + 106 \right)r(k_j), 
\end{align}
on the event $\mathcal{E}$, where $\beta_{(s_{\widetilde{m}},e_{\widetilde{m}}]}^v= \sum_{i=1}^p T^v_{(s_{\widetilde{m}},e_{\widetilde{m}}]}(\mu_{\roww{i}})^2$. By solving the quadratic inequality, we obtain that $S^v_{\gamma, (s_{\widetilde{m}},e_{\widetilde{m}}]} \geq 0$ whenever
\begin{align}
\beta_{(s_{\widetilde{m}},e_{\widetilde{m}}]}^v &\geq \left\{\gamma_0 + 170 + 8\left(2\gamma_0 + 276\right)^{1/2}\right\} r(k_j)\\
&= C_1 / 32 r(k_j).
\end{align}
Assume without loss of generality that $\eta_j \leq v$ (the converse case is similar). By the definition of the CUSUM, and using that $|\eta_j - v| \leq l/2$, we get that
\begin{align}
\beta_{(s_{\widetilde{m}},e_{\widetilde{m}}]}^v &= \frac{v-s_{\widetilde{m}}}{(e_{\widetilde{m}}-s)(e_{\widetilde{m}}-v)}(e_{\widetilde{m}}-\eta_j)^2 \phi_j^2\\
&\geq \frac{1}{2l}(l/2)^2 \phi_j^2\\
&=l\phi^2_j /8.
\end{align}
Since $l \geq C_1 r(k_j) /(4\phi^2_j)$, we must have that $\beta_{(s_{\widetilde{m}},e_{\widetilde{m}}]}^v \geq C_1 / 32 r(k_j)$, which implies \ref{it4}.

\textbf{Step 2}. We continue the proof as follows. By induction, with some slight abuse of notation, it suffices to consider any integer interval $(s,e] \subseteq (0,n]$ such that
\begin{align}
\eta_{h-1}\leq s < \eta_{h} < \ldots < \eta_{h+q} <e \leq \eta_{q+h+1},
\end{align}
for some $q\geq -1$, and, whenever $q>-1$, 
\begin{align}
s &\leq \eta_h - \Delta_h/2 ;\\
e &\geq \eta_{h+q} + \Delta_{h+q}/2.
\end{align}
Note that $q=-1$ corresponds to there being no changepoint in the open integer interval $(s,e)$. We consider this case first. For any seeded interval $(s_m, e_m] \subseteq (s,e]$ and any $s_m<v<e_m$, we will have that $S^v_{\gamma, (s_m,e_m]}<0$, due to the definition of the event $\mathcal{E}$. Hence no changepoint will be declared by ESAC in this case.

Now consider the case where $q>-1$. Note first that a changepoint will be declared by the ESAC algorithm. Indeed, for the $h$th changepoint we may take $(s_{\widetilde{m}}, e_{\widetilde{m}}]$ as in step 1, for which the properties \ref{it2} and \ref{it3} imply that $(s_{\widetilde{m}}, e_{\widetilde{m}}] \subseteq (s,e]$, due to the inductive hypothesis. By property \ref{it4}, we know that $S^v_{\gamma, (s_{\widetilde{m}},e_{\widetilde{m}}]}\geq0$, and hence a changepoint will be detected in $(s,e)$. This implies that $\mathcal{O}_{(s,e]}$, as defined in the ESAC algorithm, satisfies $\mathcal{O}_{(s,e]} \neq \emptyset$. Now let $m^*$, $v^*$ and $l^*$ be as defined in the ESAC algorithm. Note that $(s_{m^*}, e_{m^*})$ must contain a changepoint, say $\eta_j$, as we by the definition of $\mathcal{E}$ otherwise would have had $S^v_{\gamma, (s_{m^*},e_{m^*}]}<0$ for any $s_{m^*}< v< e_{m^*}$. Further, since ESAC uses the narrowest possible seeded interval to estimate a changepoint, we must have that $l^* =  (e_{m^*}-s_{m^*})$ satisfies $l^* \leq e_{\widetilde{m}} - s_{\widetilde{m}} \leq C_1 r(k_j) /(\phi^2_j) \vee 2$, where $(s_{\widetilde{m}}, e_{\widetilde{m}}]$ is the seeded interval as in the claim for $\eta_j$. Since $s_{m^*} < v^* < e_{m^*}$, it then follows that 
\begin{align}
    |v^* - \eta_j|\leq  \left\{ C_1 r(k_j) /\phi^2_j \vee 2 \right\}- 2  \leq C_1 \frac{r(k_j)}{\phi^2_j}.
\end{align}

It remains to show that the two new segments in the recursive step, $(s,s_{m^*}+1]$ and $(e_{m^*}-1, e]$ satisfy the inductive hypothesis. Without loss of generality consider $(s,s_{m^*}+1]$ (the argument for the other interval is similar), and suppose that $j \geq h+1$ (otherwise there is nothing to show). To show that the inductive hypothesis holds for $(s,s_{m^*}+1]$, it suffices to show that $s_{m^*}+1 \geq \eta_{j-1} + \Delta_{j-1}/2$. As $\eta_j \in (s_{m^*}, e_{m^*})$, we must have $e_{m^*} \geq \eta_j +1$. Hence
\begin{align}
    s_{m^*}+1 &= e_{m^*} +1  - l^*\\
    &\geq \eta_j +2 - \left\{ C_1 r(k_j) /\phi^2_j \vee 2\right\}\\
    &= \eta_{j-1} + (\eta_j - \eta_{j-1}) - \left\{ C_1 r(k_j) /\phi^2_j -2 \vee 0\right\}\\
    &\geq \eta_{j-1} + (\eta_j - \eta_{j-1}) - \Delta_j /2\\
    &\geq \eta_{j-1} + (\eta_j - \eta_{j-1})/2\\
    &\geq \eta_{j-1} + \Delta_{j-1}/2,
\end{align}
where we in the first inequality used that $l^* = (e_{m^*} - s_{m^*}) \leq C_1 r(k_j) /\phi^2_j \vee 2$ and in the second inequality used that the signal-to-noise ratio condition \eqref{SNR} implies $C_1 r(k_j) /\phi^2_j \leq \Delta_j/2$. 
Hence the inductive hypothesis holds for $(s, s_{m^*}+1]$. 
\end{proof} 

\begin{proof}[Proof of Proposition \ref{compprop}]
Let $\mathcal{M}$ denote the set of seeded intervals generated from Algorithm \ref{alg:seededintervals}. Note first that computing and storing the cumulative sum of all rows of $X$ requires $\mathcal{O}(np)$ FLOPs. Once these are stored, the number of FLOPs required to compute $C^v_{(s,e]}(j)$ as in \eqref{c2def} for some $(s,e] \in \mathcal{M}$ and some $s<v<e$ is of order $\mathcal{O}(p)$. Hence, the number of FLOPs required to compute $S^v_{\lambda,(s,e]}$ is of order $p \left |\mathcal{T} \right| = \mathcal{O}\left\{p \log(p\log n)\right\}$. In the best case there are $n-1$ changepoints detected by the ESAC algorithm using all $n-1$ intervals $(s,e] \in \mathcal{M}$ such that $e-s = 2$. In this case, the total number of FLOPs executed before ESAC terminates is of order $\mathcal{O}\left\{np + np\log(p\log n)\right\} = \mathcal{O}\left\{np\log(p\log n)\right\}$. In the worst case there are no changepoints detected by ESAC, in which case $S^v_{\lambda,(s,e]}$ has to be computed over each triple of integers $s,v,e$ such that $s<v<e$ and $(s,e] \in \mathcal{M}$. By Lemma \ref{lemma221}, there are at most $\mathcal{O}(n\log n)$ distinct such triples. Hence the number of FLOPs executed before ESAC terminates in this case is of order $\mathcal{O}\left\{np\log n \log(p\log n)\right\}$.
\end{proof}

%% file: sections/appendixA.tex
\subsection{Implementation details}\label{impdetails}\label{appendixA}
To apply ESAC in practice, a choice must be made regarding the penalty functions $\lambda, \gamma$, estimation of $\sigma$, as well as the parameters $\alpha$ and $K$ controlling the generation of seeded intervals. 
In this subsection we discuss these issues in turn, but first, we define two variants of ESAC as well as our algorithm for generating seeded intervals.
\subsubsection{Slight modifications to ESAC}
Algorithm \ref{alg:cap} constitutes the variant of the ESAC algorithm that features interval trimming. Here, the recursive step in the algorithm (the third and second last lines) differ from those found in \ref{alg:cap2}. When Algorithm \ref{alg:cap} declares a changepoint at location $v^*$, detected in the interval $(s^*,e^*)$, the remaining elements in interval $(s^*,e^*)$ are never again used to detect or estimate changepoints. 

A faster variant of Algorithm \ref{alg:cap} is given by Algorithm \ref{alg:cap3}. Algorithm \ref{alg:cap3} reduces the execution time by modifying step 3 in Algorithm \ref{alg:cap} to only evaluate $S_{\gamma, (s_m, e_m]}^v$ at the mid-point $v_m = (s_m + e_m)/2$ of any seeded interval. Interestingly, the theoretical guarantees given by Theorem \ref{adaptivetheorem} also hold for this variant of ESAC, unlike Algorithm \ref{alg:cap2}. Note that the same modification can naturally be made to Algorithm \ref{alg:cap2} as well. 
In practice, we have experienced that Algorithm \ref{alg:cap3} has much lower power for detecting changepoints compared to \ref{alg:cap2}. We therefore only recommend using Algorithm \ref{alg:cap3} when the consequent reduction in computational cost is necessary.
\begin{algorithm}
\caption{ESAC'$\left( X, (s, e], \mathcal{M}, \mathcal{B}, \gamma, \lambda \right)$.}\label{alg:cap}
\textbf{Input: } Matrix of observations $X \in \RR^{p \times n}$, left open and right closed integer interval $(s,e]$ in which candidate changepoints are searched for, an enumerated collection $\mathcal{M} =  \{(s_m, e_m] \ ; \ m\in [M]\}$ of $M$ half open integer sub intervals of $\{0, \ldots, n\}$, a set of already detected changepoints $\mathcal{B}$, and penalty functions $\gamma(t), \lambda(t)$. \\
\textbf{Output: } Set $\mathcal{B}$ of detected changepoints.

\begin{tabbing}
\qquad \enspace if {$e-s\leq 1$}:\\
\qquad \qquad stop\\
\qquad \enspace set $\mathcal{M}_{(s,e]} = \left\{   m \in [M] \ : \ (s_m, e_m] \subset (s,e]   \right\}$\\
\qquad \enspace set $\mathcal{O}_{(s,e]} = \left\{   m \in \mathcal{M}_{(s,e]} \ : \ \underset{s_m < v <e_m}{\max}  S^{v}_{\gamma, (s_m,e_m]} >0 \right\}$ \label{alltests}\\
\qquad \enspace if {$\mathcal{O}_{(s,e]} = \emptyset$}\\
\qquad \qquad stop\\
\qquad \enspace set $l^*= \underset{m \in \mathcal{O}_{(s,e]}}{\text{min}} |e_m - s_m|$\\
\qquad \enspace set $\mathcal{O}_{l^*} =  \left\{   m \in \mathcal{O}_{(s,e]} \ : \  |e_m - s_m| = l^* \right\}$\\
\qquad \enspace set $m^*= \underset{m \in \mathcal{O}_{l^*}}{\text{argmax}} \underset{s_m < v <e_m}{\max} S^v_{\lambda, (s_m, e_m]}$\\
\qquad \enspace set $v^*= \underset{s_{m^*} <v <e_{m^*}}{\text{argmax}} S^{v}_{\lambda,(s_{m^*},e_{m^*}]}$\\
\qquad \enspace $\mathcal{B} \gets \mathcal{B} \cup \{v^*\}$\\
\qquad \enspace $\mathcal{B}  \gets \text{ESAC'} \left( X,(s, s_{m^*}+1],\mathcal{M}, \mathcal{B}, \gamma, \lambda \right)$\\
\qquad \enspace $\mathcal{B} \gets \text{ESAC'} \left( X,(e_{m^*}-1,e], \mathcal{M}, \mathcal{B}, \gamma, \lambda \right)$\\
\qquad \enspace return $\mathcal{B}$
\end{tabbing}
\end{algorithm}

\begin{algorithm}
\caption{ESAC${''}\left( X, (s, e], \mathcal{M}, \mathcal{B}, \gamma, \lambda \right)$}\label{alg:cap3}
\textbf{Input: } Matrix of observations $X \in \RR^{p \times n}$, left open and right closed integer interval $(s,e]$ in which candidate changepoints are searched for, an enumerated collection $\mathcal{M} =  \{(s_m, e_m] \ ; \ m\in [M]\}$ of $M$ half open integer sub intervals of $\{0, \ldots, n\}$, a set of already detected changepoints $\mathcal{B}$, and penalty parameters $\gamma, \lambda$.  \\
\textbf{Output: } Set $\mathcal{B}$ of detected changepoints.
\begin{tabbing}
\qquad \enspace if {$e-s\leq 1$}\\
\qquad \qquad    stop\\
\qquad \enspace set $\mathcal{M}_{(s,e]} = \left\{   m \in [M] \ : \ (s_m, e_m] \subset (s,e]   \right\}$\\
\qquad \enspace set $v_m = \left \lfloor \frac{s_m + e_m}{2}\right \rfloor$ for all $m = 1, \ldots, M$\\
\qquad \enspace set $\mathcal{O}_{(s,e]} = \left\{   m \in \mathcal{M}_{(s,e]} \ :   S^{v_m}_{\gamma, (s_m,e_m]} >0 \right\}$ \\
\qquad \enspace if {$\mathcal{O}_{(s,e]} = \emptyset$}\\
\qquad \qquad stop\\
\qquad \enspace set $l^*= \underset{m \in \mathcal{O}_{(s,e]}}{\text{min}} |e_m - s_m|$\\
\qquad \enspace set $\mathcal{O}_{l^*} =  \left\{   m \in \mathcal{O}_{(s,e]} \ : \  |e_m - s_m| = l^* \right\}$ \\
\qquad \enspace set $m^*= \underset{m \in \mathcal{O}_{l^*}}{\text{argmax}} \underset{s_m < v <e_m}{\max} S^v_{\lambda, (s_m, e_m]}$\\
\qquad \enspace set $v^*= \underset{s_{m^*} <v <e_{m^*}}{\text{argmax}} S^{v}_{\lambda,(s_{m^*},e_{m^*}]}$ \\
\qquad \enspace $\mathcal{B} \gets \mathcal{B} \cup \{v^*\}$\\
\qquad \enspace $\mathcal{B}  \gets \text{ESAC}{''} \left( X,(s, s_{m^*}+1],\mathcal{M}, \mathcal{B}, \gamma, \lambda \right)$ \\
\qquad \enspace $\mathcal{B} \gets \text{ESAC}{''} \left( X,(e_{m^*}-1,e], \mathcal{M}, \mathcal{B}, \gamma, \lambda \right)$\\
\qquad \enspace \textbf{return} $\mathcal{B}$
\end{tabbing}
\end{algorithm}
%
%
%
\subsubsection{Efficient implementation of ESAC}
The ESAC Algorithms \ref{alg:cap2}, \ref{alg:cap} and \ref{alg:cap3} are based on Narrowest-Over-Threshold selection of changepoints. Once a changepoint is detected in some seeded interval, 
say of length $l^*$, the changepoint location is estimated based only on intervals of length $l^*$. To minimize running time, any version of ESAC should therefore iterate through the seeded intervals $\left\{(s_m, e_m] \ : \ m \in [M]\right\}$ in the (increasing) order of their width. This computational trick gives significant speed improvements whenever changepoints can be detected by short seeded intervals.
\subsubsection{Choice of $\alpha$ and $K$}
The choice of $\alpha$ and $K$ entails a trade-off between computational cost and statistical performance. As either $\alpha^{-1}$ or $K$ increase, more seeded intervals are generated from Algorithm \ref{alg:seededintervals}, increasing both the chance of detecting a changepoint and the running time of ESAC. After some experimentation, we have experienced that $\alpha = 3/2$ and $K=4$ give a decent balance between running time and statistical accuracy.
\subsubsection{Variance re-scaling}
In the theoretical analysis of this paper, the noise level $\sigma$ of each time series is assumed known and common across all $p$ time series. In practice, this is an unrealistic assumption. As is common in the changepoint literature, we suggest estimating the noise level separately for each time series by the (scaled) Median Absolute Deviation (MAD) of first-order differences, as in e.g. \citet{wang_high_2018}. If it is reasonable to assume that each time series has approximately the same noise level, the common noise level $\sigma$ can be estimated for instance by taking a mean or median of the MAD estimates for each time series. Once estimates of the noise levels are obtained, the time series need only to be re-scaled by their estimated noise levels before applying ESAC.
\subsubsection{Analytical choice of penalty functions}
Recall that $\lambda(t)$ and $\gamma(t)$ are the penalty functions used in the penalized score statistic for changepoint localization and detection, respectively. The proofs of Theorems \ref{locprop} and \ref{adaptivetheorem} provide suggestions for analytical choices of these penalty functions. However, we believe the leading constants are overly conservative. To obtain more practical choices of analytical penalizing functions, we have run simulations for combinations of $n$ up to $1000$ and $p$ up to $5000$. We have experienced that replacing $n$ with $n^4$ in $a(t)$ \eqref{atdef} and $r(t)$ \eqref{rdef} gives a slightly better balance between the two terms $t \log \left(\frac{ep \log n}{t^2} \right)$ and $\log n$ in $r(t)$. As default values in our R package, as 
well as in the simulation study, we have replaced $n$ with $n^4$ in $a(t)$ and $r(t)$. 
For changepoint estimation, we recommend using the penalty function
\begin{align}
\widetilde{\lambda}(t) = \begin{cases} \frac{3}{2}\left\{ \left(p \log n^4\right)^{1/2}  + \log n^4\right\} & \text{ if } t\geq(p \log n)^{1/2}, \\
t \log \left( \frac{e p \log  n^4 }{t^2}\right) + \log n^4 & \text{ otherwise.}
\end{cases}\label{tilder}
\end{align}
This recommendation is independent of whether each time series is re-scaled by Median Absolute Deviation (MAD) estimates. The choice of the two leading constants in $\widetilde{\lambda}$ are the result of minimizing the Mean Squared Error (MSE) of the estimator \eqref{singlechangeloc} over a rough grid of $n$, $p$, $\eta$ and $k$, where $n$ has ranged from $200$ to $1000$ and $p$ has ranged from $100$ to $5000$. For changepoint detection one can also use $\gamma(t) = \widetilde{\lambda}(t)$, which in our experience gives a false positive rate of less than $1/n$. If the variance of each time series is re-scaled by MAD estimates, however, we recommend choosing the penalty function $\gamma(t)$ for changepoint detection using Monte Carlo simulation. 
\subsubsection{Empirical choice of penalty functions}
To obtain exact control over the probability of a false changepoint being detected by ESAC, one can choose the penalty function $\gamma(t)$ by Monte Carlo simulation. Consider any false positive probability $\epsilon>0$ and Monte Carlo sample size $N$. A naive choice of empirical penalty function, denoted by $\widehat{\gamma}_{\epsilon}(t)$, is given by the following. Let $\mathcal{M}$ denote the collection of seeded intervals to be used by ESAC. Simulate $N$ data sets $\left(X^{(j)}\right)_{j=1}^N$ following model \eqref{datamodel} with no changepoints, in which each row is re-scaled by MAD estimates if applicable. For each $t \in \mathcal{T}$, let $\widehat{\gamma}_{\epsilon}(t)$ denote the $\lceil N(1-\epsilon)\rceil$ largest value of $\underset{(s,e] \in \mathcal{M}}{\max} \ \underset{s < v<e}{\max} \ S_{0, (s, e]}^v \left(X^{(j)}\right)(t)$ over $j = 1,\ldots, N$, where $S_{0, (s, e]}^v \left(X^{(j)}\right)(t)$ is the sparsity-specific score statistic from \eqref{stdef} computed over the seeded interval $(s,e]$ with input matrix $X^{(j)}$ and with penalty function $0$.

Due to multiple testing, the approximate false positive probability when using the naive penalty function $\widehat{\gamma}_{\epsilon}$ can only be upper bounded by $\left| \mathcal{T}\right| \epsilon$. To adjust for multiple testing, a Bonferroni correction can easily be applied by replacing $\epsilon$ by $\epsilon / \left| \mathcal{T}\right|$ in the definition of $\widehat{\gamma}_{\epsilon}(t)$. In our experience, though, such a Bonferroni correction is too conservative. An alternative approach to handle the multiple testing is to use the empirical penalty function $\widehat{\gamma}^*(t) = r(t) \ \underset{s \in \mathcal{T}}{\max} \ \widehat{\gamma}_{\epsilon}(s) / r(s)$, in which the functional form is specified and only the leading constant is chosen by Monte Carlo simulation. In our experience, this approach also leads to an overly conservative penalty function, as the functional form of $\widehat{\gamma}_{\epsilon}(t)$ does not match the theoretical counterpart $r(t)$ exactly. We therefore recommend to use the following penalty function $\widetilde{\gamma}(t)$, in which we introduce three separate leading constants for different segments of $\mathcal{T}$ (and consequently a Bonferroni correction) for slightly more flexibility. Let $\widetilde{\gamma}(t)$ be defined by
$$
\widetilde{\gamma}(t) = \begin{cases} \widetilde{\gamma}_{1} r(t) ,& \text{ for } t \leq \log{n} \wedge (p\log n)^{1/2} \\
\widetilde{\gamma}_{2} r(t) ,& \text{ for } \log{n} < t \leq (p\log n)^{1/2}\\
\widehat{\gamma}_{\epsilon/3}(p),& \text{ for }  t = p,
\end{cases}
$$
where $\widetilde{\gamma}_{1}$ and $\widetilde{\gamma}_{2}$ are defined by 
\begin{align}
\widetilde{\gamma}_1 &= \underset{t \in \mathcal{T}; t\leq \log n}{\max} \widehat{\gamma}_{\epsilon/3}(t) / r(t),\\
\widetilde{\gamma}_2 &= \underset{t \in \mathcal{T}; \log n < t \leq \surd(p\log n)}{\max} \widehat{\gamma}_{\epsilon/3}(t) / r(t).
\end{align}
The penalty function $\widetilde{\gamma}(t)$ ensures that the approximate probability of a false positive using ESAC is at most $\epsilon$. We remark that the upper boundary of the first segment ($\log n$) is chosen somewhat ad hoc, while the second segment is the remaining region of the sparse regime, and the last segment is the dense regime. The empirical penalty function $\widetilde{\gamma}(t)$ can also be used for changepoint estimation, i.e. setting $\lambda(t) = \widetilde{\gamma}(t)$, although we have experienced that the analytical penalty function $\widetilde{\lambda}(t)$ gives better performance in terms of MSE for Gaussian data. 
\subsubsection{Generation of seeded intervals}
Given some sample size $n \geq 2$ and parameters $\alpha>1$ and $K >1$, Algorithm \ref{alg:seededintervals} generates a set of seeded intervals. 
\begin{algorithm}
\caption{Seeded Interval Generation$(\alpha, K)$}\label{alg:seededintervals}
\textbf{Input:} Parameters $\alpha$ and $K$ controlling the number of generated intervals\\
\textbf{Output:} Set of seeded intervals
\begin{tabbing}
\qquad \enspace $\text{Intervals} \gets \{\}$\\
\qquad \enspace $l \gets 1$\\ 
\qquad \enspace while {$l \leq \frac{n}{2}$}:\\
    \qquad \qquad set $s = \max \left\{ 1, \lfloor \frac{l}{K} \rfloor\right\}$ \\ 
    
    \qquad \qquad for {$i = 0, \ldots,  \frac{n-2l}{s} $}:\\
        \qquad \qquad \qquad $\text{Intervals} \gets \text{Intervals}\cup \{(is, is+2l]\}$\\
        
    \qquad \qquad $\text{Intervals} \gets \text{Intervals}\cup \{(n-2l, n]\}$\\
    \qquad \qquad $l\gets \max \left\{l+1, \left \lfloor \alpha l \right \rfloor\right\}$\\

\qquad \enspace Return Intervals.
\end{tabbing}
\end{algorithm}

%% file: sections/appendixB.tex
\subsection{A Narrowest-over-Threshold variant of Inspect}\label{appendixB}
We have modified the Inspect algorithm given by Algorithm 4 in \citet{wang_high_2018} in the following fashion. Instead of using Wild Binary Segmentation as search procedure, the methodology of \citet{kovacs_seeded_2022} is used. More specifically, the collection of integer sub-intervals is generated by Algorithm \ref{alg:seededintervals} instead of the random draws. Moreover, the location of any detected changepoint is determined using only the narrowest intervals in which a changepoint is detected. We have given this modified version of Inspect the name NOTInspect, which is short for Narrowest-Over-Threshold Inspect. Formally, NOTInspect is defined as follows. For any $0\leq s < e \leq n$, let $H^{(s,e]}$ denote the $p \times (e-s-1)$ matrix in which the $(i,j)$th element is given by 
$$
H^{(s,e]}_{i,j} = T^{s+j}_{(s,e]}(X\roww{i}),
$$
i.e. the CUSUM of the $i$th row of $X$ computed over the interval $(s,e]$ and evaluated at position $s+j$. For ease of notation, let $H^{(s,e]}_v$ denote the $(v-s)$th column of $H^{(s,e]}$. Given $\lambda >0$, let $\widehat{v}^{(s,e]}_{\lambda}$ denote the leading left singular vector of the matrix
$$
\widehat{M}_{\lambda} = \underset{M \in \mathcal{S}_2}{\arg \ \max} \left(  \left\langle H^{(s,e]}, M \right\rangle- \lambda \normm{M}_1 \right),
$$
where $\mathcal{S}_2 = \left\{   M \in \RR^{p \times (e-s-1) } \ : \ \normm{M}_F \leq 1  \right \}$. 

Given an enumerated set $\mathcal{M} = \{(s_m, e_m] \}_{m=1}^M$ of $M$ half open integer sub intervals of ${0, \ldots, n}$, observations $X \in \RR^{p \times n}$, and tuning parameters $\lambda, \xi >0$, the NOTInspect algorithm is initiated by calling NOTInspect$(X, (s,e], \mathcal{M}, \emptyset, \lambda, \xi)$, and defined by Algorithm \ref{alg:inspect}.

\begin{algorithm}
\caption{NOTInspect$\left( X, (s, e], \mathcal{M}, \mathcal{B}, \lambda, \xi \right)$}\label{alg:inspect}
 {\textbf{Input: } Matrix of observations $X \in \RR^{p \times n}$, left open and right closed integer interval $(s,e]$ in which candidate changepoints are searched for, an enumerated collection $\mathcal{M} =  \{(s_m, e_m] \ ; \ m\in [M]\}$ of $M$ half open integer sub intervals of $\{0, \ldots, n\}$, a set of already detected changepoints $\mathcal{B}$, and penalization parameters $\lambda, \xi >0$.}\\
\textbf{Output: } A set $\mathcal{B}$ of detected changepoints.
\begin{tabbing}
\qquad \enspace if {$e-s\leq 1$}:\\
\qquad \qquad stop\\
\qquad \enspace set $\mathcal{M}_{(s,e]} = \left\{   m \ : \ (s_m, e_m] \subset (s,e]   \right\}$\\
\qquad \enspace set $\mathcal{O}_{(s,e]} = \left\{   m \in \mathcal{M}_{(s,e]} \ : \ \underset{s_m < b <e_m}{\max} \left(\widehat{v}^{(s_m, e_m]}_{\lambda}\right)^{\T} H_{b_m}^{(s_m,e_m]} >\xi \right\}$\\
\qquad \enspace if {$\mathcal{O}_{(s,e]} = \emptyset$}:\\
\qquad \qquad stop\\
\qquad \enspace set $l^*= \underset{m \in \mathcal{O}_{(s,e]}}{\text{min}} |e_m - s_m|$\\
\qquad \enspace set $\mathcal{O}_{l^*} = \mathcal{O}_{(s,e]} \bigcap \left\{   m \ : \  |e_m - s_m| = l^* \right\}$\\
\qquad \enspace set $m^*= \underset{m \in \mathcal{O}_{l^*}}{\text{argmax}} \underset{s_m < b <e_m}{\max} \left(\widehat{v}^{(s_m, e_m]}_{\lambda}\right)^{\T} H_{b}^{(s_m,e_m]}$\\
\qquad \enspace set $b^*= \underset{s_{m^*} <b <e_{m^*}}{\text{argmax}} \left(\widehat{v}^{(s_{m^*}, e_{m^*}]}_{\lambda}\right)^{\T} H_{b}^{(s_{m^*},e_{m^*}]}$\\
\qquad \enspace $\mathcal{B} \gets \mathcal{B} \cup \{b^*\}$\\
\qquad \enspace $\mathcal{B}  \gets \text{NOTInspect} \left( X,(s, b^*],\mathcal{M}, \mathcal{B}, \lambda, \xi \right)$\\
\qquad \enspace $\mathcal{B} \gets \text{NOTInspect} \left( X,(b^*,e], \mathcal{M}, \mathcal{B}, \lambda, \xi \right)$\\
\qquad \enspace {return} $\mathcal{B}$
\end{tabbing}
\end{algorithm}

%% file: sections/appendixC.tex
\subsection{Empirical comparison between different variants of ESAC}\label{appendixC}
In the following we compare the empirical performance of different variants of the ESAC algorithm. In all versions, seeded intervals are generated using Algorithm \ref{alg:seededintervals} with parameters $\alpha$ and $K$ specified. The variants and configurations considered are:\n\n
ESAC A: Algorithm \ref{alg:cap3} \textit{without interval trimming} and with $\alpha = 2$, $K=4$;\n
ESAC B: Algorithm \ref{alg:cap2} with $\alpha = 2$, $K=4$;\n
ESAC C: Algorithm \ref{alg:cap2} with $\alpha = 3/2$, $K=4$;\n
ESAC D: Algorithm \ref{alg:cap} with $\alpha = 3/2$, $K=4$;\n
ESAC E: Algorithm \ref{alg:cap2} without Narrowest-Over-Threshold choice of changepoint location and $\alpha =3/2$, $K=4$;\n
ESAC F: Algorithm \ref{alg:cap2} with mid-point estimation and $\alpha =3/2$, $K=4$.\n\n
ESAC A is a mix of Algorithm \ref{alg:cap3} and Algorithm \ref{alg:cap2}, in the sense that it tests for a changepoint at the midpoint of each seeded interval, but does not trim away intervals once a changepoint is detected. With ESAC E, a changepoint location is estimated by considering all seeded intervals in which a changepoint is detected, and not only the narrowest seeded intervals. This is achieved by replacing $\mathcal{O}_{l^*}$ by $ \mathcal{O}_{(s,e]}$ in Algorithm \ref{alg:cap2}. In ESAC F, the estimated changepoint location $v^*$ is replaced by $v^* = \left \lfloor ( e_{m^*} + s_{m^*}) /2 \right \rfloor$.

We have run a simulation with the exact same configuration as in Section \ref{numericmultiple} in the main text. For changepoint detection, we have chosen the empirical penalty function $\widetilde{\lambda}(t)$ as in Section \ref{impdetails} separately for each variant of ESAC. For changepoint estimation we have used the analytical penalty function $\widetilde{\lambda}(t)$ as given in Section \ref{impdetails}. For each variant of ESAC and each configuration of parameters and changepoint regimes, Table \ref{fig:tablemultiESAC} displays the average Hausdorff distance, average absolute estimation error of $J$ and average running time in milliseconds. For each configuration of parameters and changepoint regimes, the minimum (and best) value of each of the performance measures is indicated in boldface.

Comparing ESAC A and B, one observes that testing only for a changepoint at the midpoint of a seeded interval results in a substantial improvement of running time but with a cost to statistical accuracy. The running time of ESAC B is roughly three to four times that of ESAC A, while the average Hausdorff distance and absolute estimation error of $K$ of ESAC B are generally significantly larger than those of ESAC A, independently of the model configuration. This is likely due to ESAC A having lower power in detecting changepoints than ESAC B, as is indicated by ESAC A having higher estimation error of $K$. Comparing ESAC B and C, one observes a similar effect of decreasing $\alpha$ from $2$ to $3/2$. ESAC C has a running time almost twice that of ESAC B, while the average Hausdorff distance over all simulation setups is around half that of ESAC B. Comparing ESAC C and D, one observes that interval trimming substantially reduces statistical performance, with virtually no gain in terms of computational cost. Importantly, the estimation error of $K$ is markedly higher for ESAC D, which indicates that interval trimming reduces power in detecting changepoints.

Comparing ESAC C and E, one observes that using a Narrowest-over-Threshold method to estimate changepoints (as opposed to considering seeded intervals of all widths) has a mixed effect on statistical performance and a positive effect on the computational cost. In terms of Hausdorff distance, ESAC E tends to slightly outperform ESAC C, while the converse is true when considering estimation error of $K$. In terms of running time, ESAC C slightly outperforms ESAC E, especially when there are many changepoints. Lastly, comparing ESAC C and F, one observes that estimating changepoints using the penalized score statistic improves estimation accuracy compared to estimating changepoints by taking a mid-point of a seeded interval. For ESAC C, The average Hausdorff distance over all simulations is around one third that of ESAC F. Somewhat less pronounced is the difference in estimation error of $K$, where ESAC C also outperforms ESAC F.  
%
%
%
%
%
\input{tables/tablemultiFAST/table}

%% file: tables/tablemultiFAST/table.tex
 \begin{table}[h] \centering
\caption{Multiple changepoint estimation with different variants of ESAC}
\label{fig:tablemultiESAC}
\small
\begin{adjustbox}{width=\columnwidth}
\begin{tabular}{@{\extracolsep{1pt}} cccc|cccccc|cccccc|cccccc}
\hline
\multicolumn{4}{c|}{Parameters} & \multicolumn{6}{c|}{Hausdorff distance} &\multicolumn{6}{c|}{$\left | \widehat{J}-J \right |$}  &\multicolumn{6}{c}{Time in miliseconds} \\ \hline 
$n$ & $p$ & Sparsity & K & \text{A} & \text{B} & \text{C} & \text{D} & \text{E} & \text{F}  & \text{A} & \text{B} & \text{C} & \text{D} & \text{E} & \text{F} &\text{A} & \text{B} & \text{C} & \text{D} & \text{E} & \text{F} \\
\hline \
200 & 100 & - & 0  & -  & -  & -  & -  & -  & -  & \textbf{0.000}  & \textbf{0.000}  & \textbf{0.000}  & \textbf{0.000}  & \textbf{0.000}  & \textbf{0.000}  & \textbf{2.494}  & 8.122 & 13.202 & 13.291 & 13.429 & 13.254 \\
200 & 100 & Dense & 2  & 29.557 & 7.248 & 7.197 & 19.046 & \textbf{7.046}  & 12.934 & 0.453 & 0.104 & \textbf{0.088}  & 0.414 & 0.093 & 0.110 & \textbf{2.720}  & 8.088 & 12.539 & 12.329 & 13.826 & 12.549 \\
200 & 100 & Sparse & 2  & 5.172 & 1.615 & 1.658 & 6.885 & \textbf{1.525}  & 5.886 & 0.085 & 0.020 & \textbf{0.012}  & 0.150 & 0.019 & 0.025 & \textbf{2.822}  & 7.838 & 12.347 & 12.036 & 13.794 & 12.580 \\
200 & 100 & Mixed & 2  & 17.593 & 5.016 & \textbf{5.006}  & 13.105 & 5.087 & 9.727 & 0.274 & 0.065 & \textbf{0.054}  & 0.280 & 0.067 & 0.068 & \textbf{2.615}  & 7.857 & 12.385 & 12.130 & 13.766 & 12.474 \\
200 & 100 & Dense & 5  & 24.954 & 10.901 & \textbf{6.476}  & 20.446 & 6.747 & 10.947 & 1.107 & 0.345 & \textbf{0.194}  & 1.246 & 0.225 & 0.254 & \textbf{2.695}  & 7.133 & 11.089 & 10.747 & 14.609 & 11.043 \\
200 & 100 & Sparse & 5  & 7.042 & 3.584 & \textbf{1.585}  & 9.954 & 1.748 & 5.156 & 0.236 & 0.088 & \textbf{0.028}  & 0.483 & 0.058 & 0.052 & \textbf{2.699}  & 7.056 & 10.710 & 10.719 & 14.382 & 10.906 \\
200 & 100 & Mixed & 5  & 17.453 & 8.030 & 4.609 & 16.320 & \textbf{4.565}  & 8.727 & 0.734 & 0.231 & \textbf{0.128}  & 0.882 & 0.137 & 0.157 & \textbf{2.774}  & 7.083 & 10.964 & 10.580 & 14.536 & 10.942 \\
200 & 1000 & - & 0  & -  & -  & -  & -  & -  & -  & 0.002 & 0.001 & \textbf{0.000}  & \textbf{0.000}  & \textbf{0.000}  & \textbf{0.000}  & \textbf{21.873}  & 80.721 & 132.741 & 132.657 & 132.787 & 132.845 \\
200 & 1000 & Dense & 2  & 8.523 & 2.200 & 1.731 & 11.782 & \textbf{1.482}  & 6.582 & 0.144 & 0.029 & \textbf{0.016}  & 0.285 & 0.021 & 0.025 & \textbf{23.498}  & 75.455 & 121.002 & 118.612 & 134.114 & 121.129 \\
200 & 1000 & Sparse & 2  & 1.772 & 1.166 & 0.972 & 4.981 & \textbf{0.751}  & 4.916 & 0.024 & 0.011 & \textbf{0.004}  & 0.121 & 0.009 & 0.008 & \textbf{23.417}  & 73.991 & 119.314 & 116.363 & 133.431 & 118.794 \\
200 & 1000 & Mixed & 2  & 4.570 & 2.151 & 1.779 & 9.092 & \textbf{1.564}  & 6.183 & 0.081 & 0.024 & \textbf{0.015}  & 0.210 & 0.019 & 0.021 & \textbf{23.233}  & 74.204 & 119.315 & 116.722 & 133.449 & 119.462 \\
200 & 1000 & Dense & 5  & 11.205 & 3.319 & 2.201 & 16.493 & \textbf{1.929}  & 6.409 & 0.412 & 0.090 & \textbf{0.043}  & 0.919 & 0.059 & 0.053 & \textbf{23.379}  & 65.516 & 103.870 & 100.549 & 139.538 & 103.803 \\
200 & 1000 & Sparse & 5  & 3.561 & 1.982 & 0.933 & 9.515 & \textbf{0.736}  & 4.507 & 0.092 & 0.045 & \textbf{0.011}  & 0.429 & 0.025 & 0.021 & \textbf{23.647}  & 65.101 & 101.906 & 98.008 & 136.001 & 101.782 \\
200 & 1000 & Mixed & 5  & 8.422 & 3.099 & 1.754 & 14.095 & \textbf{1.373}  & 5.419 & 0.274 & 0.072 & \textbf{0.033}  & 0.682 & 0.040 & 0.052 & \textbf{23.416}  & 65.203 & 102.504 & 99.180 & 138.009 & 102.491 \\
200 & 5000 & - & 0  & -  & -  & -  & -  & -  & -  & \textbf{0.000}  & 0.002 & 0.003 & 0.003 & 0.003 & 0.003 & \textbf{110.148}  & 425.285 & 700.419 & 700.541 & 698.719 & 702.664 \\
200 & 5000 & Dense & 2  & 6.299 & 1.293 & 1.150 & 5.607 & \textbf{0.747}  & 5.082 & 0.107 & 0.010 & \textbf{0.005}  & 0.130 & 0.009 & 0.007 & \textbf{118.414}  & 385.540 & 619.304 & 599.710 & 702.220 & 619.071 \\
200 & 5000 & Sparse & 2  & 1.586 & 0.909 & 0.812 & 2.787 & \textbf{0.406}  & 4.261 & 0.022 & 0.005 & \textbf{0.001}  & 0.049 & 0.003 & 0.009 & \textbf{117.881}  & 380.432 & 607.773 & 591.900 & 702.059 & 608.943 \\
200 & 5000 & Mixed & 2  & 5.353 & 1.105 & 1.081 & 5.124 & \textbf{0.668}  & 5.067 & 0.085 & 0.007 & \textbf{0.005}  & 0.115 & 0.007 & 0.016 & \textbf{119.099}  & 382.842 & 615.231 & 595.623 & 698.241 & 613.451 \\
200 & 5000 & Dense & 5  & 8.858 & 2.516 & 0.836 & 9.285 & \textbf{0.544}  & 4.386 & 0.290 & 0.054 & \textbf{0.006}  & 0.448 & 0.016 & 0.021 & \textbf{118.777}  & 334.106 & 516.330 & 498.218 & 712.661 & 517.062 \\
200 & 5000 & Sparse & 5  & 3.073 & 1.200 & 0.680 & 4.043 & \textbf{0.341}  & 3.737 & 0.080 & 0.019 & \textbf{0.001}  & 0.160 & 0.007 & 0.015 & \textbf{118.780}  & 328.529 & 507.473 & 486.259 & 703.117 & 507.189 \\
200 & 5000 & Mixed & 5  & 6.073 & 1.582 & 0.706 & 7.221 & \textbf{0.453}  & 4.124 & 0.190 & 0.032 & \textbf{0.004}  & 0.321 & 0.014 & 0.019 & \textbf{118.770}  & 330.622 & 511.869 & 489.977 & 705.463 & 511.883 \\
500 & 100 & - & 0  & -  & -  & -  & -  & -  & -  & 0.001 & \textbf{0.000}  & 0.001 & 0.001 & 0.001 & 0.001 & \textbf{6.205}  & 24.287 & 40.742 & 40.820 & 40.817 & 40.840 \\
500 & 100 & Dense & 2  & 67.900 & 29.573 & \textbf{10.152}  & 42.356 & 10.293 & 22.459 & 0.384 & 0.148 & \textbf{0.047}  & 0.379 & 0.062 & 0.053 & \textbf{6.580}  & 23.749 & 38.286 & 37.800 & 41.399 & 38.353 \\
500 & 100 & Sparse & 2  & 21.915 & 8.268 & 2.361 & 9.141 & \textbf{2.005}  & 11.876 & 0.119 & 0.035 & \textbf{0.001}  & 0.075 & 0.007 & 0.008 & \textbf{6.589}  & 23.215 & 37.525 & 36.811 & 41.112 & 37.503 \\
500 & 100 & Mixed & 2  & 49.745 & 17.831 & \textbf{5.829}  & 27.407 & 6.041 & 17.855 & 0.277 & 0.087 & \textbf{0.021}  & 0.240 & 0.032 & 0.036 & \textbf{6.526}  & 23.268 & 37.893 & 37.271 & 41.240 & 37.939 \\
500 & 100 & Dense & 5  & 48.052 & 21.964 & \textbf{11.901}  & 50.218 & 12.038 & 22.616 & 0.820 & 0.272 & \textbf{0.125}  & 1.265 & 0.148 & 0.176 & \textbf{6.665}  & 21.533 & 34.640 & 33.852 & 43.034 & 34.689 \\
500 & 100 & Sparse & 5  & 18.105 & 7.235 & \textbf{2.312}  & 17.084 & 2.506 & 10.364 & 0.266 & 0.071 & \textbf{0.007}  & 0.334 & 0.046 & 0.032 & \textbf{6.652}  & 21.166 & 33.671 & 32.940 & 42.543 & 33.844 \\
500 & 100 & Mixed & 5  & 34.329 & 14.720 & 6.449 & 33.383 & \textbf{6.254}  & 15.993 & 0.567 & 0.174 & \textbf{0.056}  & 0.755 & 0.086 & 0.086 & \textbf{6.708}  & 21.385 & 34.159 & 33.330 & 43.037 & 34.114 \\
500 & 1000 & - & 0  & -  & -  & -  & -  & -  & -  & \textbf{0.000}  & 0.003 & \textbf{0.000}  & \textbf{0.000}  & \textbf{0.000}  & \textbf{0.000}  & \textbf{58.259}  & 241.550 & 408.466 & 410.328 & 408.340 & 411.164 \\
500 & 1000 & Dense & 2  & 35.257 & 6.585 & 4.104 & 24.037 & \textbf{3.594}  & 14.977 & 0.203 & 0.026 & \textbf{0.014}  & 0.254 & 0.018 & 0.021 & \textbf{62.356}  & 231.143 & 378.515 & 371.903 & 410.607 & 379.165 \\
500 & 1000 & Sparse & 2  & 6.619 & 2.473 & 1.603 & 7.812 & \textbf{1.050}  & 10.553 & 0.034 & 0.004 & \textbf{0.000}  & 0.078 & 0.006 & 0.001 & \textbf{62.345}  & 227.051 & 372.265 & 363.301 & 411.828 & 371.272 \\
500 & 1000 & Mixed & 2  & 21.890 & 6.052 & 2.034 & 15.490 & \textbf{1.473}  & 12.472 & 0.117 & 0.023 & \textbf{0.002}  & 0.150 & 0.008 & 0.006 & \textbf{62.723}  & 228.743 & 372.856 & 366.669 & 410.545 & 374.385 \\
500 & 1000 & Dense & 5  & 28.396 & 7.049 & 2.960 & 39.687 & \textbf{2.806}  & 12.167 & 0.438 & 0.066 & \textbf{0.017}  & 0.912 & 0.040 & 0.025 & \textbf{63.259}  & 207.844 & 336.695 & 327.587 & 427.493 & 335.335 \\
500 & 1000 & Sparse & 5  & 8.535 & 3.037 & 1.769 & 14.935 & \textbf{1.211}  & 9.094 & 0.096 & 0.020 & \textbf{0.003}  & 0.278 & 0.025 & 0.011 & \textbf{63.527}  & 204.245 & 328.075 & 317.877 & 416.863 & 326.975 \\
500 & 1000 & Mixed & 5  & 17.612 & 4.813 & 3.147 & 26.652 & \textbf{2.371}  & 11.192 & 0.274 & 0.034 & \textbf{0.017}  & 0.572 & 0.033 & 0.025 & \textbf{63.396}  & 205.635 & 333.017 & 323.492 & 423.150 & 332.866 \\
500 & 5000 & - & 0  & -  & -  & -  & -  & -  & -  & \textbf{0.000}  & \textbf{0.000}  & \textbf{0.000}  & \textbf{0.000}  & \textbf{0.000}  & \textbf{0.000}  & \textbf{316.528}  & 1328.522 & 2242.270 & 2245.298 & 2249.252 & 2236.030 \\
500 & 5000 & Dense & 2  & 16.733 & 5.984 & 2.294 & 13.965 & \textbf{1.390}  & 11.295 & 0.088 & 0.023 & \textbf{0.002}  & 0.132 & 0.007 & 0.003 & \textbf{336.784}  & 1262.177 & 2023.982 & 1994.241 & 2267.927 & 2031.555 \\
500 & 5000 & Sparse & 2  & 4.798 & 1.625 & 1.951 & 3.990 & \textbf{0.859}  & 10.402 & 0.024 & \textbf{0.001}  & 0.002 & 0.022 & 0.003 & 0.008 & \textbf{334.060}  & 1241.024 & 2009.202 & 1957.339 & 2269.248 & 1999.724 \\
500 & 5000 & Mixed & 2  & 14.071 & 5.259 & 1.946 & 8.262 & \textbf{1.092}  & 11.044 & 0.068 & 0.018 & \textbf{0.002}  & 0.068 & 0.009 & 0.009 & \textbf{331.958}  & 1250.821 & 2026.691 & 1978.321 & 2270.342 & 2019.991 \\
500 & 5000 & Dense & 5  & 13.681 & 5.718 & 2.015 & 20.418 & \textbf{1.283}  & 10.137 & 0.182 & 0.054 & \textbf{0.003}  & 0.408 & 0.029 & 0.018 & \textbf{335.042}  & 1118.323 & 1771.210 & 1701.822 & 2302.847 & 1761.599 \\
500 & 5000 & Sparse & 5  & 3.990 & 2.460 & 1.726 & 6.917 & \textbf{0.877}  & 8.494 & 0.038 & 0.012 & \textbf{0.006}  & 0.098 & 0.019 & 0.020 & \textbf{334.404}  & 1085.380 & 1720.665 & 1656.142 & 2260.874 & 1714.863 \\
500 & 5000 & Mixed & 5  & 9.622 & 4.317 & 1.925 & 14.145 & \textbf{0.908}  & 9.686 & 0.120 & 0.036 & \textbf{0.001}  & 0.255 & 0.021 & 0.021 & \textbf{331.365}  & 1095.025 & 1731.004 & 1666.874 & 2258.874 & 1733.001 \\
\hline \multicolumn{4}{c|}{Average}
 & 16.453  & 5.941 & 2.990 & 15.602 & \textbf{2.660} & 9.631 & 0.210 & 0.056 & \textbf{0.023} & 0.324 & 0.034 & 0.036 \\
\hline \\[-1.8ex]
\end{tabular}
\end{adjustbox}
\end{table}

%% file: sections/appendixD.tex
\subsection{Some more simulations}\label{appendixD}
\subsubsection{Extended table from Section \ref{numericmultiple}}
For each method considered and each configuration of parameters (now $n \in \{200,500\}$) and changepoint regimes, Table \ref{fig:tablemulti} displays the average Hausdorff distance and average absolute estimation error of $J$. 
\input{tables/tablemulti/table}

\subsubsection{Simulations with randomly drawn changes in mean}
Tables \ref{tablesinglelocuneven} and \ref{fig:tablemultiuneven} respectively display the results of re-running the simulations in Sections \ref{numericsingle} and \ref{numericmultiple} with the modification that changes in the mean-vector drawn randomly. More specifically, for each changepoint we have taken the change in mean $\theta$ to satisfy $\theta_{1:k} \propto \left(Z^{\T}, 0_{p-k}^{\T} \right)$ where $Z\sim \N_k(0, 1)$. Apart from this single modification, the simulation setups are identical to the ones in Sections \ref{numericsingle} and \ref{numericmultiple}, including the norms of the changes in mean. 

In the single changepoint case, ESAC displays a slightly larger variability in performance compared to the simulation where changes in mean are evenly spread across the affected coordinates. Averaging over all values of $n$, $p$ and $k$, one observes that the MSE of ESAC has increased slightly in comparison with table \ref{tablesingleloc}. Meanwhile, the opposite is true for the competing methods. Still, Tables \ref{tablesingleloc} and \ref{tablesinglelocuneven} are quite similar, and ESAC has highly competitive performance in both. 
\input{tables/tablesingleloc_uneven/table_even}
\input{tables/tablemulti_uneven/table}

\subsubsection{Single changepoint detection}
Here we investigate the power of each method when testing for the presence of a single changepoint. It is assumed known that there is at most one changepoint in the simulated data, and thus no multiple changepoint search method like Binary Segmentation or Seeded Binary Segmentation is used for any of the methods. Instead, we have for each method computed the corresponding test statistic for a single changepoint on the whole generated data set $X$, using e.g. $S^v_{\gamma, (0,n]}$ in \eqref{stdef} for ESAC. Our simulations are run with the same setup as in Section \ref{numericsingle}, with the exception of a slightly lower signal strength to avoid $0\%$ testing error. We adjust $\phi$ such that $\Delta\phi^2 = {n}\normm{\theta}_2^2/5 = 1.8^2 r(k)$ for each combination of $n,p$ and $k$. 

Similar to the version of ESAC given in Algorithm \ref{alg:cap3}, the Pilliat method only tests for a changepoint in the midpoint of any seeded interval $(s,e)$. This time saving trick does not affect the theoretical guarantees of neither ESAC nor Pilliat in the multiple changepoint situation because intervals for both methods are generated such that any changepoint will be close to a midpoint of some interval $(s,e)$. In this section, however, we are concerned with testing for a changepoint over a single interval (i.e. $(s,e) = (0,n)$), in which case testing only for a changepoint in the midpoint can lead to great efficiency losses whenever the true changepoint is far from the midpoint. To obtain fair and meaningful power comparisons with the remaining methods in the simulation study, we have modified the test statistic from the Pilliat algorithm to test for a changepoint in all time points $v = 1, \ldots, n-1$.

For any testing procedure, there is a trade-off between Type I and Type II errors. In order to have precise control over the Type I error of each method, we have run the competing methods with empirically chosen penalty parameters. Each method is calibrated to have Type I error at most $1\%$ based on $N=1000$ Monte Carlo simulations. The methods ESAC and Pilliat, unlike the remaining methods, combine several test statistics to test for a changepoint, resulting in a multiple testing situation. For ESAC we have adjusted for the multiple testing by using the empirical penalty function $\widetilde{\gamma}$ as defined in Section \ref{impdetails}. 
Similarly, we have for the Pilliat algorithm chosen thresholds for two of its three constituent tests by Monte Carlo simulating the leading constant in the theoretical thresholds and applied a Bonferroni correction. For the last test statistic used in the Pilliat algorithm (the Berk Jones statistic), we have used the theoretical threshold provided in the paper. 
For Inspect, we have chosen the detection threshold $\xi$ to be the $10$th largest sparse projection $\underset{0 < b <n}{\max} \left(\widehat{v}^{(0, n]}_{\lambda}\right)^{\T} T_{b}^{(0,n]}$, where $\lambda = \left\{\log \left( p \log n\right)/2\right\}^{1/2}$ (see Appendix \ref{appendixB}), over $N=1000$ data sets with no changepoints. For SUBSET we have used the function for choosing the penalty parameter $\beta$ provided by the author, with the remaining penalty parameters at their recommended values, also using the $10$th largest value out of $N=1000$ Monte Carlo samples. For Sparsified Binary Segmentation we have chosen the threshold $\pi_T$ in the same way as in Section \ref{numericsingle}, also using the $10$th largest among $N=1000$ Monte Carlo samples. For the Double CUSUM algorithm we have used the input parameter $\phi =-1$ and chosen the threshold value to be the $10$th largest double CUSUM statistic over $N=1000$ Monte Carlo simulated data sets without any changepoints.

For each method considered and each configuration of parameters, Table \ref{tablesingletest} displays the average detection rate and average running time in milliseconds. For each configuration of parameters, the best value of the detection rate and the running time is indicated in boldface (when there are no changepoints, boldface indicates the detection rate closest to 1\% from below). In terms of statistical power, Table \ref{tablesingletest} demonstrates that ESAC, Pilliat and SUBSET are the only methods with competitive power across all sparsity regimes and combinations of $n$ and $p$. Pilliat has the highest power in seven out of the $24$ different combinations of parameters with a changepoint, while the same number is three for ESAC and one for SUBSET. Averaging over the 24 combinations of parameters, Pilliat and ESAC have the highest over-all power. The Pilliat algorithm has a slight edge over ESAC, and SUBSET in third place. In comparison, Inspect has high detection power only for $k = \lceil p^{1/3} \rceil$, and with performance seemingly deteriorating when $p$ grows. Double CUSUM has excellent power for detecting dense changepoints, but fails to detect sparse changepoints, especially when $k=1$ or $p$ is large. Sparsified Binary Segmentation has high power for sparse changepoints (especially when $k=1$), but fails completely to detect dense changepoints. In terms of running time, ESAC is again the clear winner, with Pilliat as the runner-up. We remark again that SUBSET is the only method not implemented in C or C++, giving the other methods an advantage when comparing running times. We also remark that the running time of the noise level scaling by MAD estimates is not included in the running times of ESAC, Inspect, Pilliat and SUBSET, as the running time of the scaling dominates the running time of ESAC, SUBSET and Pilliat. The running time of the MAD scaling is however included in the running times of the Double CUSUM and Sparsified Binary Segmentation algorithms, as the implementations of these algorithms do not offer an option to disable the MAD scaling.

It is interesting to note that the power of ESAC, Pilliat and SUBSET seems to grow with $n$. This might be due to the Signal-to-Noise ratio of the simulated changepoints being proportional to the detection boundary for multiple changepoints, which grows faster with $n$ than the minimax testing rate for a single changepoint, see \citet{liu_minimax_2021}.
\input{tables/tablesingletest/table_even}

%% file: tables/tablemulti/table.tex
 \begin{table}[h] \centering
\caption{Multiple changepoint estimation}
\label{fig:tablemulti}
\small
\begin{adjustbox}{width=\columnwidth}
\begin{tabular}{@{\extracolsep{1pt}} cccc|cccccc|cccccc}
\hline
\multicolumn{4}{c|}{Parameters} & \multicolumn{6}{c|}{Hausdorff distance} &\multicolumn{6}{c|}{$\left | \widehat{J}-J \right |$}  \\ \hline 
        $n$ & $p$ & Sparsity & J & \text{ESAC} & \text{Pilliat} & \text{Inspect} & \text{SBS} & \text{SUBSET} & \text{DC} & \text{ESAC} & \text{Pilliat} & \text{Inspect} & \text{SBS} & \text{SUBSET} & \text{DC}  \\  \hline
        200 & 100 & - & 0 & - & - & - & - & - & - & \textbf{0.000} & \textbf{0.000} & \textbf{0.000} & 0.035 & 0.001 & 0.043  \\ 
        200 & 100 & Dense & 2 & \textbf{5.273} & 23.045 & 21.610 & 69.918 & 6.516 & 76.367 & \textbf{0.059} & 0.477 & 0.286 & 1.028 & 0.108 & 1.076  \\ 
        200 & 100 & Sparse & 2 & \textbf{1.427} & 12.539 & 7.920 & 49.193 & 1.667 & 14.918 & \textbf{0.009} & 0.245 & 0.103 & 0.695 & 0.035 & 0.260  \\ 
        200 & 100 & Mixed & 2 & 4.598 & 18.566 & 14.355 & 61.153 & \textbf{3.741} & 52.926 & \textbf{0.051} & 0.356 & 0.188 & 0.867 & 0.062 & 0.742  \\ 
        200 & 100 & Dense & 5 & \textbf{5.172} & 23.738 & 16.271 & 67.285 & 6.424 & 55.707 & \textbf{0.137} & 1.530 & 0.540 & 3.198 & 0.330 & 2.759  \\ 
        200 & 100 & Sparse & 5 & \textbf{1.359} & 13.986 & 6.378 & 57.178 & 2.424 & 23.193 & \textbf{0.023} & 0.799 & 0.228 & 2.786 & 0.204 & 1.361  \\ 
        200 & 100 & Mixed & 5 & \textbf{4.042} & 18.768 & 12.559 & 60.433 & 4.901 & 41.305 & \textbf{0.098} & 1.197 & 0.421 & 3.008 & 0.281 & 2.164  \\ 
        200 & 1000 & - & 0 & - & - & - & - & - & - & \textbf{0.000} & \textbf{0.000} & \textbf{0.000} & 0.340 & 0.024 & 0.017  \\ 
        200 & 1000 & Dense & 2 & 1.452 & 13.393 & 9.490 & 55.116 & \textbf{1.079} & 84.983 & \textbf{0.008} & 0.279 & 0.119 & 0.721 & 0.031 & 1.232  \\ 
        200 & 1000 & Sparse & 2 & 0.830 & 9.074 & 9.830 & 49.716 & \textbf{0.751} & 20.376 & \textbf{0.003} & 0.182 & 0.151 & 0.621 & 0.029 & 0.291  \\ 
        200 & 1000 & Mixed & 2 & 1.776 & 11.782 & 9.585 & 50.506 & \textbf{1.480} & 53.226 & \textbf{0.013} & 0.243 & 0.125 & 0.690 & 0.041 & 0.817  \\ 
        200 & 1000 & Dense & 5 & \textbf{1.524} & 15.808 & 9.328 & 57.330 & 1.558 & 60.791 & \textbf{0.025} & 1.018 & 0.261 & 2.960 & 0.173 & 3.057  \\ 
        200 & 1000 & Sparse & 5 & \textbf{0.685} & 11.342 & 9.161 & 54.198 & 1.520 & 25.800 & \textbf{0.002} & 0.639 & 0.357 & 2.684 & 0.167 & 1.428  \\ 
        200 & 1000 & Mixed & 5 & \textbf{1.189} & 13.932 & 9.446 & 58.255 & 1.720 & 45.158 & \textbf{0.017} & 0.813 & 0.332 & 2.864 & 0.191 & 2.377  \\ 
        200 & 5000 & - & 0 & - & - & - & - & - & - & \textbf{0.000} & \textbf{0.000} & 0.002 & 1.957 & \textbf{0.000} & \textbf{0.000}  \\ 
        200 & 5000 & Dense & 2 & 0.965 & 12.988 & 13.254 & 54.753 & \textbf{0.600} & 124.337 & \textbf{0.003} & 0.294 & 0.170 & 0.534 & 0.019 & 1.679  \\ 
        200 & 5000 & Sparse & 2 & 0.763 & 9.655 & 15.806 & 58.961 & \textbf{0.512} & 70.825 & \textbf{0.001} & 0.203 & 0.235 & 0.562 & 0.028 & 1.039  \\ 
        200 & 5000 & Mixed & 2 & 0.977 & 11.364 & 13.762 & 55.915 & \textbf{0.782} & 96.945 & \textbf{0.003} & 0.254 & 0.212 & 0.500 & 0.027 & 1.361  \\ 
        200 & 5000 & Dense & 5 & \textbf{0.891} & 14.586 & 12.967 & 50.328 & 1.373 & 83.639 & \textbf{0.009} & 0.942 & 0.449 & 2.448 & 0.173 & 3.635  \\ 
        200 & 5000 & Sparse & 5 & \textbf{0.534} & 11.697 & 13.501 & 53.998 & 1.291 & 54.328 & \textbf{0.000} & 0.735 & 0.548 & 2.529 & 0.176 & 2.712  \\ 
        200 & 5000 & Mixed & 5 & \textbf{0.776} & 13.568 & 13.252 & 53.718 & 1.413 & 69.941 & \textbf{0.004} & 0.855 & 0.501 & 2.503 & 0.174 & 3.191  \\ 
        500 & 100 & - & 0 & - & - & - & - & - & - & \textbf{0.000} & \textbf{0.000} & \textbf{0.000} & 0.057 & 0.004 & 0.046  \\ 
        500 & 100 & Dense & 2 & 15.310 & 71.171 & 50.125 & 130.610 & \textbf{7.423} & 140.433 & 0.078 & 0.612 & 0.299 & 0.784 & \textbf{0.048} & 0.838  \\ 
        500 & 100 & Sparse & 2 & \textbf{2.739} & 30.053 & 15.911 & 66.668 & 3.201 & 21.826 & \textbf{0.003} & 0.284 & 0.096 & 0.363 & 0.035 & 0.161  \\ 
        500 & 100 & Mixed & 2 & 11.309 & 54.753 & 35.821 & 93.890 & \textbf{5.418} & 88.264 & 0.063 & 0.464 & 0.208 & 0.533 & \textbf{0.038} & 0.530  \\ 
        500 & 100 & Dense & 5 & 20.245 & 66.890 & 50.548 & 120.816 & \textbf{8.226} & 102.263 & 0.229 & 1.939 & 0.726 & 2.543 & \textbf{0.207} & 1.988  \\ 
        500 & 100 & Sparse & 5 & \textbf{3.591} & 36.375 & 20.258 & 70.980 & 4.767 & 19.457 & \textbf{0.021} & 0.923 & 0.301 & 1.605 & 0.191 & 0.513  \\ 
        500 & 100 & Mixed & 5 & 12.290 & 53.954 & 39.507 & 104.507 & \textbf{6.803} & 77.078 & \textbf{0.150} & 1.489 & 0.582 & 2.138 & 0.193 & 1.441  \\ 
        500 & 1000 & - & 0 & - & - & - & - & - & - & \textbf{0.000} & 0.002 & \textbf{0.000} & 0.471 & \textbf{0.000} & 0.025  \\ 
        500 & 1000 & Dense & 2 & \textbf{2.217} & 30.635 & 18.859 & 102.138 & 2.790 & 99.626 & \textbf{0.002} & 0.295 & 0.094 & 0.531 & 0.044 & 0.591  \\ 
        500 & 1000 & Sparse & 2 & 2.052 & 17.323 & 23.773 & 71.741 & \textbf{1.325} & 18.185 & \textbf{0.002} & 0.148 & 0.135 & 0.382 & 0.023 & 0.116  \\ 
        500 & 1000 & Mixed & 2 & 2.181 & 25.431 & 21.361 & 92.648 & \textbf{1.907} & 68.446 & \textbf{0.002} & 0.223 & 0.114 & 0.496 & 0.061 & 0.426  \\ 
        500 & 1000 & Dense & 5 & \textbf{2.893} & 36.449 & 25.510 & 97.098 & 3.660 & 86.127 & \textbf{0.015} & 0.964 & 0.282 & 2.143 & 0.174 & 1.717  \\ 
        500 & 1000 & Sparse & 5 & \textbf{1.808} & 22.483 & 22.461 & 66.221 & 4.089 & 20.893 & \textbf{0.001} & 0.547 & 0.363 & 1.493 & 0.191 & 0.498  \\ 
        500 & 1000 & Mixed & 5 & \textbf{2.526} & 31.025 & 24.077 & 90.958 & 4.682 & 60.749 & \textbf{0.011} & 0.768 & 0.315 & 1.847 & 0.213 & 1.216  \\ 
        500 & 5000 & - & 0 & - & - & - & - & - & - & \textbf{0.000} & 0.002 & 0.006 & 2.508 & 0.025 & \textbf{0.000}  \\ 
        500 & 5000 & Dense & 2 & 2.011 & 34.047 & 21.082 & 139.285 & \textbf{1.647} & 282.956 & \textbf{0.002} & 0.329 & 0.105 & 1.106 & 0.030 & 1.562  \\ 
        500 & 5000 & Sparse & 2 & 2.033 & 16.974 & 32.140 & 144.237 & \textbf{1.408} & 132.128 & \textbf{0.002} & 0.148 & 0.188 & 1.159 & 0.083 & 0.778  \\ 
        500 & 5000 & Mixed & 2 & 2.277 & 26.553 & 24.275 & 141.006 & \textbf{1.990} & 207.933 & \textbf{0.003} & 0.241 & 0.133 & 1.146 & 0.027 & 1.184  \\ 
        500 & 5000 & Dense & 5 & \textbf{1.937} & 42.354 & 24.293 & 84.334 & 3.570 & 187.785 & \textbf{0.006} & 1.182 & 0.292 & 1.133 & 0.186 & 3.405  \\ 
        500 & 5000 & Sparse & 5 & \textbf{1.671} & 21.168 & 29.794 & 79.665 & 2.654 & 104.455 & \textbf{0.002} & 0.560 & 0.500 & 1.073 & 0.150 & 2.029  \\ 
        500 & 5000 & Mixed & 5 & \textbf{1.957} & 31.474 & 25.265 & 84.857 & 3.864 & 143.185 & \textbf{0.006} & 0.860 & 0.380 & 1.154 & 0.190 & 2.772  \\ \hline
        \multicolumn{4}{c|}{Average} & 3.480 & 25.248 & 20.098 & 77.767 & \textbf{3.033} & 81.015 & \textbf{0.025} & 0.549 & 0.246 & 1.385 & 0.104 & 1.263 \\ 
\hline
\end{tabular}
\end{adjustbox}
\end{table}

%% file: tables/tablesingleloc_uneven/table_even.tex
\begin{table}[h] \centering
\caption{Single changepoint estimation with changes in mean randomly drawn}
\label{tablesinglelocuneven}
\small
\begin{adjustbox}{width=\columnwidth}
\begin{tabular}{@{\extracolsep{1pt}} ccccc|ccccc|ccccc}
\hline
\multicolumn{5}{c|}{Parameters} & \multicolumn{5}{c|}{MSE} &\multicolumn{5}{c|}{Time in milliseconds} \\ \hline 
$n$ & $p$ & $k$ &  $\eta$ & $\phi$ & \text{ESAC} & \text{Inspect} & \text{SBS} & \text{SUBSET} & \text{DC} & \text{ESAC} & \text{Inspect} &\text{SBS} & \text{SUBSET}& \text{DC} \\
\hline \
200 & 100 & 1 & 40 & 1.40 & 10.1 & 25.7 & 70.4 & 39.2 & \textbf{9.4}  & \textbf{0.3}  & 2.2 & 11.0 & 2.1 & 13.8 \\
200 & 100 & 5 & 40 & 2.00 & 4.6 & \textbf{2.4}  & 40.0 & 2.6 & 6.6 & \textbf{0.4}  & 2.2 & 11.2 & 2.4 & 14.5 \\
200 & 100 & 24 & 40 & 1.90 & 46.6 & \textbf{20.3}  & 977.3 & 125.1 & 168.4 & \textbf{0.3}  & 2.1 & 10.6 & 1.8 & 17.5 \\
200 & 100 & 100 & 40 & 1.90 & \textbf{53.0}  & 147.2 & 1507.2 & 307.4 & 1614.6 & \textbf{0.3}  & 2.2 & 10.5 & 1.9 & 13.4 \\
200 & 1000 & 1 & 40 & 1.52 & \textbf{5.1}  & 95.2 & 34.9 & 27.5 & 5.2 & \textbf{2.1}  & 44.3 & 96.1 & 13.4 & 136.0 \\
200 & 1000 & 10 & 40 & 2.93 & 1.4 & 0.6 & 2.5 & \textbf{0.5}  & 3.4 & \textbf{2.2}  & 44.2 & 96.0 & 13.9 & 138.2 \\
200 & 1000 & 73 & 40 & 3.37 & 6.1 & \textbf{2.0}  & 510.3 & 4.9 & 15.0 & \textbf{2.2}  & 44.3 & 95.4 & 13.7 & 137.7 \\
200 & 1000 & 1000 & 40 & 3.37 & \textbf{3.7}  & 544.8 & 1546.1 & \textbf{3.7}  & 357.0 & \textbf{2.1}  & 44.5 & 94.8 & 13.4 & 138.9 \\
200 & 5000 & 1 & 40 & 1.60 & 38.9 & 481.5 & \textbf{19.3}  & 47.6 & 165.3 & \textbf{12.4}  & 226.8 & 486.2 & 80.6 & 730.1 \\
200 & 5000 & 18 & 40 & 4.00 & 1.1 & \textbf{0.5}  & 0.9 & 0.6 & 1.4 & \textbf{12.6}  & 226.8 & 483.2 & 81.7 & 731.0 \\
200 & 5000 & 163 & 40 & 5.04 & 3.4 & \textbf{0.7}  & 281.0 & 3.4 & 18.2 & \textbf{12.6}  & 226.4 & 479.2 & 82.3 & 733.8 \\
200 & 5000 & 5000 & 40 & 5.04 & \textbf{3.8}  & 1299.5 & 1548.5 & \textbf{3.8}  & 6.0 & \textbf{13.0}  & 226.6 & 478.3 & 83.1 & 734.2 \\
500 & 100 & 1 & 100 & 0.92 & \textbf{48.8}  & 86.9 & 130.1 & 242.1 & 48.9 & \textbf{0.6}  & 6.0 & 17.9 & 4.6 & 29.7 \\
500 & 100 & 5 & 100 & 1.31 & 14.2 & 13.3 & 49.7 & \textbf{12.5}  & 37.2 & \textbf{0.6}  & 6.0 & 17.9 & 5.1 & 29.0 \\
500 & 100 & 25 & 100 & 1.25 & 223.3 & \textbf{63.7}  & 5044.6 & 746.3 & 671.9 & \textbf{0.6}  & 6.0 & 17.7 & 4.1 & 28.4 \\
500 & 100 & 100 & 100 & 1.25 & \textbf{446.9}  & 727.1 & 9526.7 & 1953.2 & 8663.5 & \textbf{0.6}  & 6.0 & 18.6 & 3.9 & 28.1 \\
500 & 1000 & 1 & 100 & 1.00 & 30.1 & 200.6 & 30.0 & 125.3 & \textbf{29.3}  & \textbf{5.7}  & 275.3 & 160.7 & 35.7 & 288.9 \\
500 & 1000 & 10 & 100 & 1.90 & 4.7 & 4.3 & 11.6 & \textbf{2.7}  & 13.9 & \textbf{5.7}  & 274.5 & 161.5 & 35.2 & 288.5 \\
500 & 1000 & 79 & 100 & 2.22 & 14.6 & \textbf{10.5}  & 1854.5 & 13.3 & 168.5 & \textbf{5.8}  & 275.2 & 161.1 & 35.4 & 288.2 \\
500 & 1000 & 1000 & 100 & 2.22 & \textbf{22.4}  & 2083.8 & 9755.1 & 63.5 & 2677.1 & \textbf{6.0}  & 277.1 & 160.6 & 35.4 & 289.5 \\
500 & 5000 & 1 & 100 & 1.05 & \textbf{26.4}  & 1290.1 & 36.0 & 75.3 & 265.8 & \textbf{32.8}  & 1401.2 & 812.4 & 178.8 & 1634.8 \\
500 & 5000 & 18 & 100 & 2.58 & 2.2 & 1.4 & 7.0 & \textbf{1.2}  & 9.3 & \textbf{31.7}  & 1399.8 & 804.8 & 178.9 & 1564.4 \\
500 & 5000 & 177 & 100 & 3.32 & 13.9 & \textbf{1.9}  & 812.9 & 13.9 & 111.0 & \textbf{32.8}  & 1399.2 & 803.7 & 177.0 & 1563.2 \\
500 & 5000 & 5000 & 100 & 3.32 & \textbf{14.9}  & 6546.6 & 9921.2 & \textbf{14.9}  & 25.2 & \textbf{33.6}  & 1400.3 & 800.8 & 175.1 & 1560.7 \\
\hline \multicolumn{5}{c|}{Average MSE}
 & \textbf{43.331}  & 568.765 & 1821.581 & 159.599 & 628.846 \\
\hline \\[-1.8ex]
\end{tabular}
\end{adjustbox}
\end{table}

%% file: tables/tablemulti_uneven/table.tex
 \begin{table}[h] \centering
\caption{Multiple changepoint estimation with randomly drawn changes in mean}
\label{fig:tablemultiuneven}
\small
\begin{adjustbox}{width=\columnwidth}
\begin{tabular}{@{\extracolsep{1pt}} cccc|cccccc|cccccc|cccccc}
\hline
\multicolumn{4}{c|}{Parameters} & \multicolumn{6}{c|}{Hausdorff distance} &\multicolumn{6}{c|}{$\left | \widehat{J}-J \right |$}  &\multicolumn{6}{c}{Time in miliseconds} \\ \hline 
$n$ & $p$ & Sparsity & K & \text{ESAC}  & \text{Pilliat} & \text{Inspect} & \text{SBS} & \text{SUBSET}  & \text{DC} & \text{ESAC}  & \text{Pilliat} & \text{Inspect} & \text{SBS} & \text{SUBSET}  & \text{DC} &\text{ESAC} & \text{Pilliat} & \text{Inspect} & \text{SBS} & \text{SUBSET}  & \text{DC}  \\
\hline \
200 & 100 & - & 0  & -  & -  & -  & -  & -  & -  & \textbf{0.000}  & \textbf{0.000}  & \textbf{0.000}  & 0.036 & 0.001 & 0.056 & \textbf{12.966}  & 30.647 & 74.424 & 216.435 & 210.336 & 267.195 \\
200 & 100 & Dense & 2  & 6.343 & 24.101 & 16.779 & 54.275 & \textbf{5.489}  & 40.699 & \textbf{0.068}  & 0.494 & 0.217 & 0.811 & 0.087 & 0.595 & \textbf{13.286}  & 27.903 & 63.091 & 215.905 & 369.398 & 283.611 \\
200 & 100 & Sparse & 2  & \textbf{1.041}  & 9.866 & 9.135 & 51.509 & 1.341 & 13.402 & \textbf{0.005}  & 0.168 & 0.121 & 0.688 & 0.041 & 0.185 & \textbf{12.863}  & 26.726 & 59.260 & 217.011 & 376.712 & 285.572 \\
200 & 100 & Mixed & 2  & 4.530 & 18.165 & 12.471 & 49.626 & \textbf{3.579}  & 29.501 & \textbf{0.044}  & 0.350 & 0.156 & 0.712 & 0.056 & 0.403 & \textbf{12.998}  & 27.244 & 60.739 & 215.607 & 368.155 & 281.791 \\
200 & 100 & Dense & 5  & \textbf{5.051}  & 23.400 & 13.704 & 60.773 & 7.000 & 36.150 & \textbf{0.134}  & 1.571 & 0.446 & 2.984 & 0.372 & 1.990 & \textbf{11.560}  & 23.256 & 44.433 & 218.781 & 406.915 & 286.411 \\
200 & 100 & Sparse & 5  & \textbf{1.144}  & 11.798 & 7.383 & 54.875 & 2.197 & 21.282 & \textbf{0.014}  & 0.628 & 0.272 & 2.793 & 0.215 & 1.227 & \textbf{11.062}  & 21.775 & 40.321 & 216.358 & 403.685 & 286.231 \\
200 & 100 & Mixed & 5  & \textbf{3.868}  & 19.370 & 11.640 & 57.086 & 4.795 & 32.205 & \textbf{0.102}  & 1.135 & 0.390 & 2.825 & 0.259 & 1.687 & \textbf{11.252}  & 22.338 & 42.175 & 218.172 & 398.624 & 287.708 \\
200 & 1000 & - & 0  & -  & -  & -  & -  & -  & -  & \textbf{0.001}  & 0.002 & \textbf{0.001}  & 0.334 & 0.049 & 0.011 & \textbf{133.635}  & 417.508 & 676.153 & 1383.229 & 1644.789 & 1541.410 \\
200 & 1000 & Dense & 2  & 1.621 & 13.373 & 7.469 & 46.956 & \textbf{0.914}  & 65.967 & \textbf{0.010}  & 0.288 & 0.085 & 0.652 & 0.025 & 0.966 & \textbf{123.439}  & 349.062 & 504.772 & 1396.731 & 2937.718 & 1620.628 \\
200 & 1000 & Sparse & 2  & \textbf{0.580}  & 5.056 & 9.926 & 45.467 & \textbf{0.580}  & 9.819 & \textbf{0.001}  & 0.061 & 0.150 & 0.601 & 0.025 & 0.160 & \textbf{112.348}  & 332.217 & 493.877 & 1405.696 & 2935.402 & 1680.447 \\
200 & 1000 & Mixed & 2  & 1.160 & 9.583 & 9.511 & 48.964 & \textbf{1.000}  & 44.287 & \textbf{0.007}  & 0.171 & 0.118 & 0.625 & 0.034 & 0.629 & \textbf{113.235}  & 340.377 & 500.685 & 1395.207 & 2947.915 & 1646.245 \\
200 & 1000 & Dense & 5  & \textbf{1.676}  & 15.654 & 6.579 & 52.049 & 2.015 & 52.562 & \textbf{0.025}  & 0.967 & 0.181 & 2.737 & 0.184 & 2.694 & \textbf{98.949}  & 283.893 & 339.906 & 1403.437 & 3137.914 & 1689.102 \\
200 & 1000 & Sparse & 5  & \textbf{0.654}  & 6.088 & 10.978 & 56.695 & 1.380 & 20.461 & \textbf{0.004}  & 0.261 & 0.409 & 2.797 & 0.180 & 1.196 & \textbf{92.722}  & 266.748 & 332.063 & 1406.943 & 3189.067 & 1704.795 \\
200 & 1000 & Mixed & 5  & \textbf{0.972}  & 11.230 & 7.933 & 54.369 & 1.532 & 40.742 & \textbf{0.014}  & 0.587 & 0.283 & 2.757 & 0.181 & 2.090 & \textbf{94.930}  & 274.330 & 336.437 & 1396.920 & 3169.231 & 1693.709 \\
200 & 5000 & - & 0  & -  & -  & -  & -  & -  & -  & 0.001 & \textbf{0.000}  & 0.003 & 2.045 & \textbf{0.000}  & \textbf{0.000}  & \textbf{691.022}  & 2483.067 & 3438.853 & 6474.674 & 7839.738 & 45460.281 \\
200 & 5000 & Dense & 2  & \textbf{0.833}  & 12.928 & 9.419 & 53.485 & 0.881 & 125.673 & \textbf{0.002}  & 0.253 & 0.113 & 0.503 & 0.029 & 1.675 & \textbf{608.734}  & 2068.055 & 2597.533 & 6476.473 & 13875.649 & 46029.927 \\
200 & 5000 & Sparse & 2  & 0.680 & 5.005 & 21.564 & 56.211 & \textbf{0.372}  & 20.694 & \textbf{0.002}  & 0.078 & 0.311 & 0.544 & 0.021 & 0.288 & \textbf{576.760}  & 1980.902 & 2657.543 & 6486.496 & 13803.731 & 45776.696 \\
200 & 5000 & Mixed & 2  & 0.819 & 8.515 & 17.548 & 54.106 & \textbf{0.476}  & 76.468 & \textbf{0.005}  & 0.142 & 0.233 & 0.535 & 0.021 & 1.027 & \textbf{588.131}  & 2024.048 & 2643.031 & 6478.838 & 14013.301 & 46452.491 \\
200 & 5000 & Dense & 5  & \textbf{0.838}  & 14.882 & 10.696 & 51.242 & 1.505 & 82.618 & \textbf{0.006}  & 0.965 & 0.332 & 2.432 & 0.191 & 3.589 & \textbf{503.523}  & 1679.471 & 1755.036 & 6466.005 & 14882.984 & 44220.589 \\
200 & 5000 & Sparse & 5  & \textbf{0.564}  & 5.720 & 17.985 & 52.833 & 1.438 & 28.486 & \textbf{0.004}  & 0.268 & 0.716 & 2.542 & 0.167 & 1.557 & \textbf{469.414}  & 1583.551 & 1810.660 & 6480.760 & 14912.785 & 45156.719 \\
200 & 5000 & Mixed & 5  & \textbf{0.915}  & 10.379 & 13.680 & 51.268 & 1.290 & 55.573 & \textbf{0.011}  & 0.549 & 0.531 & 2.485 & 0.173 & 2.667 & \textbf{483.432}  & 1628.287 & 1787.464 & 6447.313 & 14864.080 & 43521.559 \\
500 & 100 & - & 0  & -  & -  & -  & -  & -  & -  & 0.001 & \textbf{0.000}  & \textbf{0.000}  & 0.046 & 0.001 & 0.035 & \textbf{38.700}  & 94.066 & 304.918 & 246.704 & 357.170 & 358.181 \\
500 & 100 & Dense & 2  & 14.852 & 71.176 & 36.840 & 76.220 & \textbf{7.096}  & 50.032 & 0.072 & 0.598 & 0.209 & 0.443 & \textbf{0.043}  & 0.298 & \textbf{40.440}  & 88.531 & 270.099 & 256.006 & 638.682 & 396.802 \\
500 & 100 & Sparse & 2  & \textbf{2.493}  & 18.588 & 19.554 & 59.445 & 3.716 & 16.005 & \textbf{0.003}  & 0.122 & 0.117 & 0.340 & 0.043 & 0.110 & \textbf{38.181}  & 83.009 & 253.752 & 255.744 & 636.948 & 398.321 \\
500 & 100 & Mixed & 2  & 11.450 & 49.265 & 31.611 & 74.956 & \textbf{6.349}  & 38.182 & 0.054 & 0.390 & 0.183 & 0.422 & \textbf{0.046}  & 0.230 & \textbf{39.096}  & 85.391 & 260.994 & 254.795 & 641.740 & 396.331 \\
500 & 100 & Dense & 5  & 19.669 & 65.912 & 41.700 & 86.378 & \textbf{8.573}  & 50.159 & \textbf{0.235}  & 1.918 & 0.552 & 1.913 & 0.245 & 0.936 & \textbf{36.228}  & 76.931 & 206.456 & 258.670 & 679.120 & 419.395 \\
500 & 100 & Sparse & 5  & \textbf{2.497}  & 23.396 & 20.207 & 70.891 & 5.215 & 14.485 & \textbf{0.009}  & 0.518 & 0.296 & 1.595 & 0.217 & 0.325 & \textbf{34.096}  & 71.251 & 188.874 & 260.845 & 681.193 & 423.444 \\
500 & 100 & Mixed & 5  & 12.201 & 48.611 & 32.186 & 78.728 & \textbf{7.454}  & 36.218 & \textbf{0.148}  & 1.241 & 0.460 & 1.762 & 0.238 & 0.754 & \textbf{34.799}  & 73.492 & 198.098 & 257.864 & 690.167 & 418.055 \\
500 & 1000 & - & 0  & -  & -  & -  & -  & -  & -  & \textbf{0.000}  & \textbf{0.000}  & \textbf{0.000}  & 0.492 & 0.016 & 0.017 & \textbf{398.208}  & 1285.509 & 3847.094 & 1662.366 & 2824.899 & 2383.389 \\
500 & 1000 & Dense & 2  & 2.771 & 31.180 & 13.583 & 78.957 & \textbf{2.315}  & 74.251 & \textbf{0.006}  & 0.276 & 0.062 & 0.399 & 0.038 & 0.461 & \textbf{369.555}  & 1118.043 & 2773.639 & 1698.756 & 5014.257 & 2678.245 \\
500 & 1000 & Sparse & 2  & 1.666 & 9.834 & 25.831 & 73.875 & \textbf{1.250}  & 8.877 & \textbf{0.001}  & 0.041 & 0.168 & 0.383 & 0.027 & 0.072 & \textbf{345.645}  & 1072.310 & 2701.346 & 1699.812 & 5020.515 & 2720.316 \\
500 & 1000 & Mixed & 2  & 2.211 & 20.015 & 19.778 & 78.736 & \textbf{2.062}  & 45.206 & \textbf{0.003}  & 0.152 & 0.107 & 0.392 & 0.039 & 0.286 & \textbf{352.894}  & 1094.544 & 2712.138 & 1698.072 & 4973.244 & 2699.230 \\
500 & 1000 & Dense & 5  & \textbf{3.369}  & 36.761 & 17.330 & 80.512 & 4.553 & 64.838 & \textbf{0.021}  & 1.004 & 0.200 & 1.751 & 0.221 & 1.272 & \textbf{318.164}  & 964.578 & 1829.702 & 1719.044 & 5337.658 & 2886.657 \\
500 & 1000 & Sparse & 5  & \textbf{1.404}  & 9.821 & 25.099 & 74.378 & 3.735 & 12.836 & \textbf{0.001}  & 0.135 & 0.397 & 1.512 & 0.209 & 0.314 & \textbf{300.161}  & 905.299 & 1768.004 & 1727.178 & 5324.184 & 2912.666 \\
500 & 1000 & Mixed & 5  & \textbf{2.193}  & 24.535 & 19.081 & 74.682 & 3.797 & 47.089 & \textbf{0.009}  & 0.511 & 0.280 & 1.585 & 0.184 & 0.971 & \textbf{308.808}  & 933.998 & 1801.129 & 1717.245 & 5316.684 & 2874.535 \\
500 & 5000 & - & 0  & -  & -  & -  & -  & -  & -  & 0.001 & 0.001 & 0.005 & 2.548 & 0.001 & \textbf{0.000}  & \textbf{2190.188}  & 7671.521 & 21419.785 & 7984.972 & 13603.143 & 80867.898 \\
500 & 5000 & Dense & 2  & 2.521 & 32.488 & 15.256 & 138.242 & \textbf{1.923}  & 276.838 & \textbf{0.004}  & 0.319 & 0.066 & 1.124 & 0.025 & 1.536 & \textbf{1986.577}  & 6730.429 & 14858.876 & 7973.077 & 23835.000 & 73071.458 \\
500 & 5000 & Sparse & 2  & 1.292 & 5.778 & 44.258 & 130.900 & \textbf{0.847}  & 19.421 & \textbf{0.001}  & 0.010 & 0.273 & 1.030 & 0.019 & 0.110 & \textbf{1871.097}  & 6313.414 & 15406.145 & 7977.799 & 23782.197 & 77183.655 \\
500 & 5000 & Mixed & 2  & 2.227 & 18.370 & 26.340 & 142.507 & \textbf{1.243}  & 155.575 & \textbf{0.003}  & 0.132 & 0.153 & 1.089 & 0.019 & 0.917 & \textbf{1934.793}  & 6529.652 & 15159.541 & 7965.174 & 23783.834 & 75432.805 \\
500 & 5000 & Dense & 5  & \textbf{1.884}  & 39.501 & 16.271 & 80.997 & 3.996 & 176.776 & \textbf{0.004}  & 1.178 & 0.185 & 1.091 & 0.221 & 3.332 & \textbf{1720.410}  & 5776.899 & 9325.082 & 7991.455 & 24947.670 & 73536.119 \\
500 & 5000 & Sparse & 5  & \textbf{1.299}  & 5.986 & 37.767 & 78.567 & 3.635 & 22.392 & \textbf{0.002}  & 0.045 & 0.621 & 1.127 & 0.203 & 0.596 & \textbf{1590.057}  & 5281.750 & 9756.378 & 8003.694 & 25304.561 & 71966.096 \\
500 & 5000 & Mixed & 5  & \textbf{2.164}  & 22.478 & 27.356 & 81.147 & 3.183 & 116.785 & \textbf{0.005}  & 0.438 & 0.407 & 1.092 & 0.166 & 2.236 & \textbf{1646.648}  & 5509.990 & 9507.831 & 8003.054 & 24760.199 & 69574.033 \\
\hline \multicolumn{4}{c|}{Average}
 & 3.374  & 21.078 & 19.012 & 69.775 & \textbf{3.020} & 56.182 & \textbf{0.025} & 0.428 & 0.234 & 1.299 & 0.108 & 0.940 \\
\hline \\[-1.8ex]
\end{tabular}
\end{adjustbox}
\end{table}

%% file: tables/tablesingletest/table_even.tex
\begin{table}[h] \centering
\caption{Single changepoint detection}
\label{tablesingletest}
\small
\begin{adjustbox}{width=\columnwidth,center}
\begin{tabular}{@{\extracolsep{1pt}} ccccc|cccccc|cccccc}
\hline
\multicolumn{5}{c|}{Parameters} & \multicolumn{6}{c|}{Detection rate} &\multicolumn{6}{c|}{Time in milliseconds} \\ \hline 
$n$ & $p$ & $k$ &  $\eta$ & $\phi$ & \text{ESAC} & \text{Pilliat} & \text{Inspect} & \text{SBS} & \text{SUBSET} & \text{DC} & \text{ESAC}  & \text{Pilliat} & \text{Inspect} &\text{SBS} & \text{SUBSET}& \text{DC} \\
\hline \
200 & 100 & - & - & - & 0.016 & 0.013 & \textbf{0.008} & 0.015 & 0.017 & 0.007  & \textbf{0.2}  & 0.8 & 2.2 & 11.2 & 2.1 & 9.5 \\
200 & 100 & 1 & $n/5$ & 1.12 & 0.849 & 0.815 & 0.232 & \textbf{0.892}  & 0.886 & 0.082 & \textbf{0.1}  & 0.3 & 2.1 & 9.8 & 1.7 & 8.6 \\
200 & 100 & 5 & $n/5$ & 1.60 & 0.952 & \textbf{0.965}  & 0.911 & 0.690 & 0.962 & 0.775 & \textbf{0.1}  & 0.2 & 2.0 & 9.4 & 1.6 & 8.6 \\
200 & 100 & 24 & $n/5$ & 1.52 & 0.696 & 0.749 & 0.666 & 0.075 & 0.567 & \textbf{0.813}  & \textbf{0.1}  & 0.4 & 2.0 & 12.6 & 1.6 & 8.5 \\
200 & 100 & 100 & $n/5$ & 1.52 & 0.675 & 0.740 & 0.550 & 0.038 & 0.543 & \textbf{0.819}  & \textbf{0.1}  & 0.4 & 2.1 & 9.4 & 1.7 & 8.8 \\
200 & 1000 & - & - & - & 0.008 & 0.009 & \textbf{0.010} & 0.006  & \textbf{0.010} & 0.008 & \textbf{1.8}  & 10.6 & 42.8 & 88.2 & 12.5 & 87.5 \\
200 & 1000 & 1 & $n/5$ & 1.22 & 0.809 & 0.796 & 0.003 & \textbf{0.853}  & 0.841 & 0.018 & \textbf{1.1}  & 4.0 & 42.5 & 88.4 & 13.8 & 89.0 \\
200 & 1000 & 10 & $n/5$ & 2.35 & 0.994 & \textbf{0.995}  & 0.498 & 0.649 & 0.987 & 0.299 & \textbf{0.9}  & 2.2 & 42.4 & 88.1 & 12.7 & 90.3 \\
200 & 1000 & 73 & $n/5$ & 2.70 & 0.823 & 0.880 & 0.645 & 0.046 & 0.788 & \textbf{0.900}  & \textbf{1.1}  & 3.3 & 42.4 & 88.0 & 13.0 & 89.6 \\
200 & 1000 & 1000 & $n/5$ & 2.70 & 0.805 & 0.868 & 0.390 & 0.014 & 0.783 & \textbf{0.918}  & \textbf{1.1}  & 3.4 & 42.5 & 87.7 & 13.3 & 89.1 \\
200 & 5000 & - & - & - & 0.007 & \textbf{0.008} & 0.003  & 0.012 & 0.012 & 0.012 & \textbf{11.6}  & 65.4 & 216.7 & 449.8 & 78.0 & 458.9 \\
200 & 5000 & 1 & $n/5$ & 1.28 & 0.782 & 0.775 & 0.000 & \textbf{0.856}  & 0.797 & 0.018 & \textbf{7.1}  & 26.0 & 215.9 & 447.3 & 77.7 & 455.9 \\
200 & 5000 & 18 & $n/5$ & 3.20 & 0.996 & \textbf{1.000}  & 0.011 & 0.698 & 0.997 & 0.182 & \textbf{5.6}  & 17.4 & 215.0 & 446.0 & 78.0 & 454.6 \\
200 & 5000 & 163 & $n/5$ & 4.03 & 0.911 & \textbf{0.925}  & 0.249 & 0.038 & 0.898 & 0.911 & \textbf{6.7}  & 18.5 & 214.9 & 441.5 & 79.1 & 457.1 \\
200 & 5000 & 5000 & $n/5$ & 4.03 & 0.897 & 0.893 & 0.119 & 0.013 & 0.868 & \textbf{0.934}  & \textbf{8.8}  & 20.8 & 214.2 & 440.8 & 78.6 & 457.9 \\
500 & 100 & - & - & - & 0.007 & 0.005  & 0.007 & 0.015 & 0.015 & \textbf{0.008} & \textbf{0.5}  & 2.1 & 5.6 & 16.0 & 4.2 & 16.0 \\
500 & 100 & 1 & $n/5$ & 0.74 & 0.944 & 0.897 & 0.288 & \textbf{0.973}  & 0.968 & 0.069 & \textbf{0.2}  & 0.6 & 5.4 & 16.5 & 3.8 & 15.7 \\
500 & 100 & 5 & $n/5$ & 1.04 & 0.978 & \textbf{0.988}  & 0.973 & 0.882 & 0.987 & 0.821 & \textbf{0.2}  & 0.5 & 5.4 & 16.7 & 3.6 & 15.5 \\
500 & 100 & 25 & $n/5$ & 1.00 & 0.809 & 0.785 & 0.755 & 0.120 & 0.718 & \textbf{0.852}  & \textbf{0.3}  & 0.8 & 5.4 & 15.9 & 3.6 & 15.5 \\
500 & 100 & 100 & $n/5$ & 1.00 & 0.796 & 0.776 & 0.627 & 0.062 & 0.699 & \textbf{0.853}  & \textbf{0.3}  & 0.8 & 5.4 & 16.0 & 3.7 & 15.8 \\
500 & 1000 & - & - & - & 0.004  & 0.004  & 0.005 & 0.014 & \textbf{0.010} & 0.008 & \textbf{4.4}  & 26.4 & 261.9 & 147.0 & 32.8 & 163.2 \\
500 & 1000 & 1 & $n/5$ & 0.80 & 0.924 & 0.876 & 0.000 & \textbf{0.963}  & 0.942 & 0.026 & \textbf{2.5}  & 8.8 & 259.7 & 147.4 & 32.4 & 163.4 \\
500 & 1000 & 10 & $n/5$ & 1.52 & 0.994 & \textbf{0.998}  & 0.551 & 0.860 & 0.994 & 0.357 & \textbf{2.3}  & 5.7 & 259.0 & 148.1 & 32.4 & 164.7 \\
500 & 1000 & 79 & $n/5$ & 1.78 & \textbf{0.936}  & 0.926 & 0.711 & 0.073 & 0.868 & 0.925 & \textbf{2.8}  & 7.6 & 260.5 & 148.6 & 32.5 & 165.2 \\
500 & 1000 & 1000 & $n/5$ & 1.78 & 0.927 & 0.920 & 0.466 & 0.023 & 0.877 & \textbf{0.965}  & \textbf{2.5}  & 8.1 & 260.4 & 149.1 & 32.3 & 165.6 \\
500 & 5000 & - & - & - & 0.006 & 0.002  & \textbf{0.010} & 0.008 & 0.006 & 0.011 & \textbf{25.9}  & 158.2 & 1343.5 & 769.4 & 182.0 & 955.7 \\
500 & 5000 & 1 & $n/5$ & 0.84 & 0.939 & 0.890 & 0.000 & \textbf{0.958}  & 0.947 & 0.017 & \textbf{13.7}  & 46.6 & 1341.9 & 762.7 & 184.3 & 908.5 \\
500 & 5000 & 18 & $n/5$ & 2.06 & \textbf{1.000}  & \textbf{1.000}  & 0.010 & 0.907 & \textbf{1.000}  & 0.204 & \textbf{12.0}  & 28.3 & 1341.4 & 763.1 & 185.4 & 908.1 \\
500 & 5000 & 177 & $n/5$ & 2.66 & \textbf{0.964}  & 0.959 & 0.455 & 0.045 & 0.951 & 0.956 & \textbf{12.8}  & 37.7 & 1342.3 & 762.1 & 184.1 & 908.7 \\
500 & 5000 & 5000 & $n/5$ & 2.66 & 0.953 & 0.956 & 0.243 & 0.009 & 0.956 & \textbf{0.982}  & \textbf{13.0}  & 37.5 & 1342.8 & 757.8 & 179.5 & 907.7 \\
\hline \multicolumn{5}{c|}{Average detection rate}
 & 0.890 & \textbf{0.891}  & 0.390 & 0.447 & 0.868 & 0.571 \\
\hline \\[-1.8ex]
\end{tabular}
\end{adjustbox}
\end{table}

%% file: sections/lemmas.tex
\subsection{Auxiliary Lemmas}\label{appendixE}
\begin{lemma}\label{lemma10}
For any $a\geq 0$, define $\nu_a = \EE\left( Z^2 \ \mid \ |Z|\geq a\right)$ where $Z \sim \N(0,1)$. Then
$$
a^2 + 1 \leq \nu_a \leq a^2+2.
$$
\end{lemma}
\begin{proof}
The second inequality follows from Lemma 4 in \citet{liu_minimax_2021}. For the first inequality, let $\overline{\Phi}(x)= \int_{x}^{\infty} \phi(t) dt$, where $\phi(\cdot)$ denotes the density function of a standard normal distribution. If $a>0$, we have that
\begin{align}
\nu_a - 1- a^2 &= a \frac{\phi(a)}{\overline{\Phi}(a)} - a^2\\
&=a\left\{ \frac{\phi(a)}{\overline{\Phi}(a)} - a\right\}\\
&\geq 0,
\end{align}
using that ${\phi(a)} / {\overline{\Phi}(a)} > a$ for all $a>0$ \citep[see e.g.][]{sampford_inequalities_1953}. For $a=0$, we have that $\nu_a = \EE(Z^2) = 1$, and the claim follows.
\end{proof}

The following Lemma is due to \citet{liu_minimax_2021}. 
\begin{lemma}[\citealt{liu_minimax_2021}, Lemma 5]\label{lemma6}
Let $Z_i \iid \N(0,1)$ for $i \in [p]$, where $p \in \NN$. Let $a\geq 0$ and define $\nu_a = \EE\left( Z^2 \ \mid \ |Z|\geq a\right)$. Then for all $x>0$, 
$$
\PP \left[ \sum_{i=1}^p (Z_i^2 - \nu_a) \ind{\left( |Z_i| \geq a\right)} \geq 9 \left \{  \left( p e^{-a^2/2}x\right)^{1/2} + x \right\}\right] \leq e^{-x}.
$$
\end{lemma}

The following Lemma is analogous to Lemma \ref{lemma6}, and gives a corresponding lower bound.
\begin{lemma}\label{lemma5}
Let $Z_i \iid \N(0,1)$ for $i \in [p]$, where $p \in \NN$. Let $a\geq 1$ and define $\nu_a = \EE\left( Z_1^2 \ \mid \ |Z_1|\geq a\right)$. Then for all $x>0$, 
$$
\PP \left[  \sum_{i=1}^p (Z_i^2 - \nu_a) \ind{\left( |Z_i| \geq a\right)} \leq - 5 \left \{  \left( p e^{-a^2/2}x\right)^{1/2} + x \right\}\right) \leq e^{-x}.
$$
\end{lemma}
\begin{proof}
The proof is similar to the proof of Lemma 5 in \citet{liu_minimax_2021}. Let $X = \left( Z^2 - \nu_a \right) \ind{\left( |Z| \geq a\right)}$, where $Z \sim \N(0,1)$. Let $ \lambda \in (0,\frac{1}{2}]$. Then, as $\EE(X) = 0$, we have that
\begin{align}
\EE \left( e^{-\lambda X} \right) &= 1 + \EE \left( e^{-\lambda X} - 1 + \lambda X\right),
\end{align}
By the deterministic bound
$$
e^{-y} -1 + y \leq \begin{cases} 
y^2, & \text{ if } y>0,\\
y^2, & \text{ if } -1 \leq y \leq 0,\\
e^{-y},& \text{ if } y \leq -1,
\end{cases}
$$
we obtain that
\begin{align}
\EE \left(e^{- \lambda X}\right) \leq & 1+ \lambda ^2 \EE \left \{  X^2 \ind{\left( X>0 \right)}   \right\} + \lambda ^2 \EE \left \{ X^2 \ind{\left( - \frac{1}{\lambda} \leq X \leq 0 \right)}   \right\} \notag \\
+ & \EE \left\{ e^{-\lambda X} \ind {\left(  X < -\frac{1}{\lambda}\right)}\right\}.
\end{align}
We bound each term separately. Let $p(x)$ denote the density function of the $\chi_1^2$ distribution. For the second term, we have that
\begin{align}
\EE \left \{  X^2 \ind {\left( X>0 \right)}   \right\} &= \int_{\nu_a}^{\infty} (x-\nu_a)^2 p(x) dx\\
&= \int_{\nu_a}^{\infty} (x-\nu_a)^2 \frac{1}{(2\pi x)^{1/2}} e^{-x/2} dx\\
&\leq \frac{16}{(2\pi \nu_a)^{1/2}} e^{-\nu_a/2},
\end{align}
using that $1+ a^2 \leq \nu_a \leq a^2+2$ (Lemma \ref{lemma10}) and $a \geq 1$. For the third term, using that $X \geq a^2 - \nu_a \geq -2$ whenever $X\leq 0$, we have that
\begin{align}
\EE \left \{  X^2 \ind {\left( - \frac{1}{\lambda} \leq X \leq 0 \right)}   \right\} &\leq \EE \left \{ 2^2 \ind {\left( - \frac{1}{\lambda} \leq X \leq 0 \right)}   \right\}\\
&\leq 4 \EE \left \{ \ind {\left( |Z| \geq a\right)}   \right\}\\
&\leq 8  e^{-a^2/2} / (2\pi a^2)^{1/2}\\
&\leq 8 e^{-a^2/2} / (2\pi)^{1/2},
\end{align}
where we in the penultimate step used the standard bound $\PP(Z >a) \leq e^{-a^2/2} /(2\pi a^2)^{1/2}$ for all $a>0$.
For the last term, as $\lambda \leq 1/2$, we have that $\PP(X < -\frac{1}{\lambda}) \leq \PP(X <-2) = 0$, because $X \geq a^2 - \nu_a \geq -2$. Therefore, $\EE \left\{ e^{-\lambda X} \ind {\left(  X < - 1/\lambda \right)}\right\}=0$. Hence, 
\begin{align}
\EE \left( e^{- \lambda X} \right) &\leq 1 + \lambda^2 \left\{ \frac{8}{(2\pi)^{1/2}} + \frac{16 e^{-\frac{1}{2}}}{2 \surd{\pi}}\right\} e^{-a^2/2}\\
&\leq 1 + 6 \lambda^2  e^{-a^2/2}\\
&\leq \exp \left( 6 \lambda^2  e^{-a^2/2}\right).
\end{align}
By a Chernoff Bound we obtain that, for any $t>0$, 
\begin{align}
\PP \left\{  \sum_{i=1}^n (Z_i^2 - \nu_a) \ind {\left( |Z_i| \geq a\right)} < - t \right\} &= \PP \left\{ - \sum_{i=1}^n (Z_i^2 - \nu_a) \ind {\left( |Z_i| \geq a\right)}> t \right\}\\
&\leq \underset{0<\lambda \leq \frac{1}{2}}{\inf} e^{-\lambda t} \left \{ \EE \left( e^{-\lambda X} \right) \right\}^p\\
&\leq  \underset{0<\lambda \leq \frac{1}{2}}{\inf} \exp\left( -\lambda t + 6 \lambda^2 p e^{-a^2/2}\right)\\
&\leq \exp\left\{ - \left(   \frac{t^2 e^{a^2/2}}{24 p} \wedge \frac{t}{4}  \right)\right\}. 
\end{align}
Now take $t = 5 \left\{ \left( pe^{-a^2/2} x\right)^{1/2} + x \right\}$ to obtain the result.
\end{proof}

The following Lemma is due to \citet{birge_alternative_2001}.
\begin{lemma}[\citealt{birge_alternative_2001}, Lemma 8.1]\label{lemma12}
Let $Y \sim \chi_p^2(\Psi)$ have a non-central Chi Square distribution with $p$ degrees of freedom and non-centrality parameter $\Psi\geq0$. Then, for any $x>0$, we have that
\begin{align}
\PP\left[     Y\geq p + \Psi  + 2 \left\{x(p + 2\Psi)\right\}^{1/2} +2x      \right] &\leq e^{-x},
\intertext{and,}
\PP\left[     Y\leq p + \Psi  - 2\left\{x(p + 2\Psi)\right\}^{1/2}       \right] &\leq e^{-x},
\end{align}
\end{lemma}

\begin{lemma}\label{lemma15}
Consider the model from Section \ref{secproblemdesc}, with one and only one changepoint $\eta$, and suppose $n\geq 3$ and $\sigma = 1$. Let $\mathcal{K} = \left\{   i \ : \ \mu_{i, \eta+1} - \mu_{i, \eta} \neq 0 \right\}$ denote the set of coordinates for which there is a change in mean, let $r(t)$ be defined as in \eqref{rdef},  and let $h(t)$ be defined as in \eqref{hdef}. Let the CUSUM transformation $T^v_{(s,e]}(\cdot)$ be defined as in \eqref{c2def}, and for ease of notation, let $T^v(\cdot) =T^v_{(0,n]}(\cdot)$.  Let $k = \normm{\mu_{\eta+1} - \mu_{\eta}}_0$, and define $\beta_v = \sum_{i \in \mathcal{K}} \left\{ T^{\eta}(\mu\roww{i})^2 - T^v (\mu\roww{i})^2 \right\}$. Define the events
\begin{align}
\mathcal{E}_1 &= \left\{    \forall 0<v < n, \ \sum_{i \in \mathcal{K}} \left\{ C^{\eta}(i)^2 - C^{v}(i)^2 \right\} \geq \beta_v   - 2\left( 2\beta_v \log n \right)^{1/2} - 16 r(k) \right\},\\
  \mathcal{E}_2 &= \left\{   \forall 0<v<n, \forall t \in \mathcal{T},  \sum_{i \in [p] \setminus \mathcal{K}}  \left\{ {C}^{v} (i)^2  - \nu_{a(t)} \right \} \ind {\left\{\left|{C}^{v}(i)\right| > a(t) \right\}} \leq  63 r(t) \right\}\\
 \mathcal{E}_3 &= \left\{ \forall 0<v<n, \forall t \in \mathcal{T},  \sum_{i \in [p] \setminus \mathcal{K}} \left\{ {C}^{v} (i)^2  - \nu_{a({t})} \right \} \ind {\left\{\left|{C}^{v}(i)\right| > a({t}) \right\}}  \geq - 35 r({t}) \right\}\\
\mathcal{E}_4 &= \left\{   \forall 0<v<n, \forall t  \in \mathcal{T}, t < (p \log n)^{1/2} \ ; \right. \notag\\
 & \quad \quad  \sum_{i =1 }^p \left[  \left\{ {C}^{v} (i)^2  - \nu_{a(t)} \right \} \ind {\left\{\left|{C}^{v}(i)\right| > a(t) \right\}} - C^v(i)^2  +1 \right] \leq \left.   \vphantom{\sum_{i \in \mathcal{K}^c } \left( T^{\eta} (X\roww{i})^2  - \nu_{c(k)} \right ) \ind_{\left\{\left|T^{\eta}(X\roww{i})\right| > c(k) \right\}} - \left( T^{v} (X\roww{i})^2  - \nu_{a(t)} \right ) \ind_{\left\{\left|T^{v}(X\roww{i})\right| > a(t) \right\}}} 5 h(p) +63r(t) \right\}
\end{align}
Then $\PP \left( \mathcal{E}_1 \cap \mathcal{E}_2 \cap \mathcal{E}_3 \cap \mathcal{E}_4 \right) \geq 1-\frac{1}{n}$.
\end{lemma}
\begin{proof}

By a union bound it suffices to consider each event separately.

\textbf{Step 1}.
We first show that $\PP\left(\mathcal{E}_1^c \right) \leq \frac{1}{3n}$. As the CUSUM is a linear operation and $X = \mu + W$, we have for any $0<v<n$ and $i \in [p]$ that $C^v(i) = T^v(\mu\roww{i})  + T^v(W\roww{i})$. Hence, for any $v$, we have that 
\begin{align}
\sum_{i \in \mathcal{K}} \big\{{C}^{\eta}(i)^2 - {C}^v (i)^2\big\} &= 
\beta_v + \sum_{i \in \mathcal{K}} \big\{T^{\eta}(W\roww{i})^2 - T^v (W\roww{i})^2\big\}\\
&+ 2 \sum_{i \in \mathcal{K}} \big\{T^{\eta}(W\roww{i})T^{\eta}(\mu\roww{i}) - T^{v}(W\roww{i})T^{v}(\mu\roww{i})\big\}. 
\end{align}
We construct high-probability bounds on the two last terms separately. For the first term, note that since $W_{i,j}\iid \N(0,1)$, for any fixed $v$ we have $T^v(W\roww{i}) {\sim} \N(0,1)$ independently for all $i \in [p]$. Hence $\sum_{i \in \mathcal{K}} T^{\eta}(W\roww{i})^2 \sim \chi^2_k$ and $\sum_{i \in \mathcal{K}} T^{v}(W\roww{i})^2 \sim \chi^2_k$. By Lemma \ref{lemma12} and a union bound we therefore have that
\begin{align}
\PP  \left[ \sum_{i \in \mathcal{K}}  \left\{ T^{\eta}(W\roww{i})^2 - T^v (W\roww{i})^2  \right\} \leq - 4\{\log (9n^2)k\}^{1/2} -2 \log(9n^2) \right] \leq \frac{2}{9n^2}.
 \end{align}
Using that $n\geq 3$ and the definition of $r(k)$, we obtain that
\begin{align}
\PP  \left[ \sum_{i \in \mathcal{K}}  \left\{ T^{\eta}(W\roww{i})^2 - T^v (W\roww{i})^2   \right\} \leq - 16r(k) \right] \leq \frac{2}{9n^2}\label{hpbound1}.
 \end{align}
 To see this, consider first the case $k<({p\log n})^{1/2}$. Then $k \leq r(k)$ and $\log n \leq r(k)$, so $4\left\{ \log (9n^2) k\right\}^{1/2}\leq 4 \left\{4 \log(n) k\right\}^{1/2}\leq  8r(k)$ and $2 \log (9n^2)\leq 8 \log n \leq 8r(k)$. For the case $k\geq (p\log n)^{1/2}$, we must have that $p \geq \log n$, and hence $4\left\{\log(9n^2)k\right\}^{1/2} \leq 8 \left(p \log n\right)^{1/2} = 8r(k)$ and $2 \log (9n^2)\leq 8 \log(n) \leq 8 (p\log n)^{1/2} = 8r(k)$.

For the second term, we make use of the fact that the CUSUM transformation $T^v(y)$ of any vector $y$ can be expressed as an inner product. More precisely, define the $n$-dimensional vector $\Psi^v \in \RR^n$ to have $l$th element given by
\begin{align}
\Psi^v(l)= \begin{cases}
\ \left(\frac{n-v}{nv}\right)^{1/2} & \text{ for } l = 1, \ldots, v,\\
-\left(\frac{v}{n(n-v)}\right)^{1/2} & \text{ for } l = v+1\ldots, n. \\
\end{cases}\label{psidef}
\end{align}
Then for any vector $y \in \RR^n$, we have that
\begin{align}
T^v(y) &= \langle y, \Psi^{v} \rangle,
\end{align}
see \citet{baranowski_narrowest-over-threshold_2019}. Hence, for any $i \in \mathcal{K}$,
 \begin{align}
T^{\eta}(\mu_{\cdot,i})T^{\eta}(W_{\cdot,i}) - T^{v}(\mu_{\cdot,i})T^{v}(W_{\cdot,i}) &=\langle \mu_{\roww{i}}, \Psi^{\eta} \rangle\langle W_{\roww{i}}, \Psi^{\eta} \rangle - \langle \mu_{\roww{i}}, \Psi^{v} \rangle\langle W_{\roww{i}}, \Psi^{v} \rangle \\
&= \left\langle W_{\roww{i}}, \langle \mu_{\roww{i}}, \Psi^{\eta} \rangle \Psi^{\eta} \right\rangle - \left \langle W_{\roww{i}}, \langle \mu_{\roww{i}}, \Psi^{v} \rangle\Psi^{v} \right \rangle\\
&= \left \langle W_{\roww{i}}, \langle \mu_{\roww{i}}, \Psi^{\eta} \rangle \Psi^{\eta}  -  \langle \mu_{\roww{i}}, \Psi^{v} \rangle\Psi^{v} \right\rangle.
\end{align}
As $W\roww{i}\overset{\text{ind}}{\sim} \N_n\left(0, I\right)$ for all $i \in \mathcal{K}$, we get that 
\begin{align}
    T^{\eta}(\mu_{\cdot,i})T^{\eta}(W_{\cdot,i}) - T^{v}(\mu_{\cdot,i})T^{v}(W_{\cdot,i}) \overset{\text{ind}}{\sim} \N \left( 0, \normm{\langle \mu_{\roww{i}}, \Psi^{\eta} \rangle \Psi^{\eta}  -  \langle \mu_{\roww{i}}, \Psi^{v} \rangle\Psi^{v} }_2^2 \right),
\end{align}
for $i \in \mathcal{K}$. 
By Lemma \ref{lemma4NOT}, we have that $\normm{\langle \mu_{\roww{i}}, \Psi^{\eta} \rangle \Psi^{\eta}  -  \langle \mu_{\roww{i}}, \Psi^{v} \rangle\Psi^{v} }_2^2 = T^{\eta}(\mu\roww{i})^2 - T^v (\mu\roww{i})^2$. We therefore have that
\begin{align}
 \sum_{i \in \mathcal{K}} \left\{ T^{\eta}(X\roww{i})T^{\eta}(\mu\roww{i}) - T^{v}(X\roww{i})T^{v}(\mu\roww{i}) \right\} \sim \N\left( 0 , \beta_v\right).
\end{align}
By the standard Gaussian tail bound $\PP(Z>t)\leq e^{-t^2/2}$ for $Z \sim \N(0,1)$ and $t>0$, we obtain
\begin{align}
\PP \left[ \sum_{i \in \mathcal{K}}  \left\{ T^{\eta}(X\roww{i})T^{\eta}(\mu\roww{i}) - T^{v}(X\roww{i})T^{v}(\mu\roww{i}) \right\} < - 2(2\beta_v \log n )^{1/2} \right] \leq \frac{1}{9n^2} \label{l15e2},
\end{align}
again using that $n\geq 3$. Combining \eqref{hpbound1} and \eqref{l15e2} by a union bound, we have for any $0<v<n$ that
\begin{align}
\PP \left[    \ \sum_{i \in \mathcal{K}} \left\{ C^{\eta}(i)^2 - C^{v}(i)^2 \right\} \geq \beta_v   - 2 ( 2 \beta_v \log n)^{1/2} - 16 r(k) \right] \leq \frac{1}{3n^2}.
\end{align}
By another union bound (over $v$), we obtain that $\PP(\mathcal{E}_1^c)\leq \frac{1}{3n}$.\n

\textbf{Step 2}. We now show that $\PP\left( \mathcal{E}_2^c \right) \leq {1}/{6n}$. Fix $0<v<n$ and any $t \in \mathcal{T}$ such that $t \leq (p\log n)^{1/2}$. Fix $x_t>0$, to be determined later. As ${C}^{v} (i) \iid \N(0,1)$ for all $i \in \mathcal{K}^c$, we have by Lemma \ref{lemma6} that
\begin{align}
 \sum_{i \in \mathcal{K}^c } \left\{ {C}^{v} (i)^2  - \nu_{a(t)} \right \} \ind {\left\{\left|{C}^{v}(i)\right| > a(t) \right\}} &\leq 9 \left [  \left\{ p e^{-a(t)^2/2} x_t\right\}^{1/2} + x_t \right] \label{l15e3}
\end{align}
with probability at least $1-e^{-x_t}$. By a union bound, \eqref{l15e3} holds for all $0<v<n$ and $t\in \mathcal{T}$ such that $t \leq (p\log n)^{1/2}$, with probability at least $1 - n \sum_{t\in \mathcal{T} \setminus\{p\}} e^{-x_t}$. Now set $x_t = 6 \left\{  {p \log^2(n)}/{t^2} \wedge r(t)  \right\}$. Then
\begin{align}
\sum_{t\in \mathcal{T} \setminus\{p\}} e^{-x_t} &\leq \sum_{t\in \mathcal{T}\setminus\{p\}}   \exp\left\{-6\frac{p \log^2(n)}{t^2}\right\} + \sum_{t\in \mathcal{T}\setminus\{p\}} \exp\left\{ -6r(t)\right\}.
\end{align}
For the first sum, we have
\begin{align}
\sum_{t\in \mathcal{T}\setminus\{p\}}   \exp\left\{-6\frac{p \log^2(n)}{t^2}\right\} &\leq \sum_{k=0}^{\infty} \exp\left\{-6{\log(n)} 4^{k}\right\} \\
&= \sum_{k=0}^{\infty}\left( \frac{1}{n^6} \right)^{4^k}\\
&\leq \frac{1}{n^6} + \frac{1}{n^6} \sum_{k=1}^{\infty}\left( \frac{1}{n^6} \right)^{3k}\\
&=\frac{1}{n^6} \left( 1 + \frac{1}{n^{18}-1}\right).
\end{align}
For the second sum, noting that $6 r(t) = 6 \left\{ t \log \left( {ep\log n}/{t^2}\right) \vee \log n\right\} \geq 3 t \log \left( {ep\log n}/{t^2}\right) + 3\log n$, we have
\begin{align}
\sum_{t\in \mathcal{T}\setminus\{p\}} \exp\left\{ -6r(t)\right\} &\leq \frac{1}{n^3}  
\sum_{t\in \mathcal{T}\setminus\{p\}}  \left(   \frac{t^2}{ep\log n}   \right)^{3t}\\
&\leq \frac{1}{n^3 e^3} \left( 1 + \sum_{k=1}^{\infty} 4^{-3k} \right).
\end{align}
With our choice of $x_t$, using that $a^2(t) = 4\log\left( {ep\log n}/{t^2}  \right)$, we have that
\begin{align}
9 \left [  \left\{ p e^{-a^2(t)/2} x_t\right\}^{1/2} + x_t \right] &= 9 \left [ \left\{ p \frac{t^4}{e^2 p^2 \log^2 n} x_t\right\}^{1/2} + x_t \right] \\
&\leq 9 \left \{ \frac{t\surd{6}}{e}  + 6 r(t)\right\}\\
&\leq 63 r(t),
\end{align}
where we used that $x_t \leq 6r(t)$ and $x_t \leq {6p \log^2(n)}/{t^2}$, as well as the fact that $t \leq r(t)$ whenever $t \leq (p \log n)^{1/2}$. Hence, using that $n\geq 3$,
\begin{align}
&\PP \left[ \exists 0<v<n, \exists t \in \mathcal{T}, t \leq (p\log n)^{1/2} \ ; \ \sum_{i \in \mathcal{K}^c } \left\{ {C}^{v} (i)^2  - \nu_{a(t)} \right \} \ind {\left\{\left|{C}^{v}(i)\right| > a(t) \right\}} > 63 r(t) \right] \notag\\
\leq & \frac{1}{n^6} \left( 1 + \frac{1}{n^{18}-1}\right) + \frac{1}{n^3 e^3} \left( 1 + \sum_{k=1}^{\infty} 4^{-3k} \right) \\
\leq & \frac{1}{18n^2}\label{hpbound2}.
\end{align}
Now consider the case where $t = p$. If $p \leq (p\log n)^{1/2}$, then $a(p)>0$ and similarly as above we have that
\begin{align}
&\PP \left[ \exists 0<v<n,  \ ; \ \sum_{i \in \mathcal{K}^c } \left\{ {C}^{v} (i)^2  - \nu_{a(p)} \right \} \ind {\left\{\left|{C}^{v}(i)\right| > a(p) \right\}} > 63 r(p) \right]\\
\leq & \frac{1}{18n^2} \label{hpbound3}.
\end{align}
If we instead have $p > (p \log n)^{1/2}$, in which case $a(p) = 0$ and $\nu_{a(p)} = 1$, then for any $0<v<n$, we have
\begin{align}
\sum_{i \in \mathcal{K}^c } \left\{ {C}^{v} (i)^2  - \nu_{c(p)} \right \} \ind {\left\{\left|{C}^{v}(i)\right| > c(p) \right\}} &= \sum_{i \in \mathcal{K}^c } \left\{ {C}^{v} (i)^2 \right\}- p + k.  
\end{align}
As $\sum_{i \in \mathcal{K}^c } {C}^{v} (i)^2   \sim \chi^2_{p-k}$, we obtain from Lemma \ref{lemma12} that
\begin{align}
\sum_{i \in \mathcal{K}^c } \left\{  {C}^{v} (i)^2 \right\} -p +k & >  2 \left\{p \log(12n^2)\right\}^{1/2} + 2\log(12n^2),
\end{align}
with probability at most ${1}/{(12n^2)}$. Using that $n\geq 3$ and $r(p) \geq \log n$, we obtain by a union bound that
\begin{align}
\PP \left[\exists 0<v<n \ ; \ \sum_{i \in \mathcal{K}^c } \left\{{C}^{v} (i)^2  - \nu_{c(p)} \right \} \ind {\left\{\left|{C}^{v}(i)\right| > c(p) \right\}} >  15 r(t) \right] \leq \frac{1}{12n} \label{l15e4},
\end{align}
Combining \eqref{hpbound2}, \eqref{hpbound3} and \eqref{l15e4} by a union bound, we have that $\PP(\mathcal{E}_2^c) \leq {1}/{(6n)}$.



\textbf{Step 3}. We show that $\PP\left(   \mathcal{E}_3^c \right) \leq {1}/{6n}$. Fix $0<v<n$ and any $t \in \mathcal{T}$ such that $t \leq (p\log n)^{1/2}$. Fix $x_t>0$, to be determined later. As ${C}^{v} (i) \iid \N(0,1)$ for all $i \in \mathcal{K}^c$, we have by Lemma \ref{lemma5} that
\begin{align}
 \sum_{i \in \mathcal{K}^c } \left\{ {C}^{v} (i)^2  - \nu_{a(t)} \right \} \ind {\left\{\left|{C}^{v}(i)\right| > a(t) \right\}} &\geq -5 \left [  \left\{ p e^{-a(t)^2/2} x_t \right\}^{1/2} + x_t \right] \label{l15e32}
\end{align}
with probability at least $1-e^{-x_t}$. By a union bound, \eqref{l15e32} holds for all $0<v<n$ and $t\in \mathcal{T}$ such that $t \leq (p\log n)^{1/2}$, with probability at least $1 - n \sum_{t\in \mathcal{T} \setminus\{p\}} e^{-x_t}$. Now set $x_t = 6 \left\{  {p \log^2(n)}/{t^2} \wedge r(t)  \right\}$. Similar to Step 2, we obtain that
\begin{align}
&\PP \left[ \exists 0<v<n, \exists t \in \mathcal{T}, t \leq (p\log n)^{1/2} \ ; \ \sum_{i \in \mathcal{K}^c } \left\{ {C}^{v} (i)^2  - \nu_{a(t)} \right \} \ind {\left\{\left|{C}^{v}(i)\right| > a(t) \right\}} < - 35 r(t) \right]\notag\\
\leq & \frac{1}{18n^2}\label{hpbound5}.
\end{align}
Also similar to Step 2, we have that 
\begin{align}
&\PP \left [ \exists 0<v<n \ ; \ \sum_{i \in \mathcal{K}^c } \left\{ {C}^{v} (i)^2  - \nu_{a(p)} \right \} \ind {\left\{\left|{C}^{v}(i)\right| > a(p) \right\}} < - 35 r(p) \right]\notag\\
\leq & \frac{1}{12n^2}\label{hpbound6}.
\end{align}
It then follows that $\PP\left(\mathcal{E}_3^c\right) \leq {1}/{(6n)}$ by a union bound.

\textbf{Step 4}. Lastly we show that $\PP\left( \mathcal{E}_4^c \right) \leq {1}/{(3n)}$. Fix any $0<v<n$ and $t  \in \mathcal{T}$ such that $t < (p \log n)^{1/2}$. By Lemma \ref{lemma16} and Theorem 1.A.3(b) in \citet{shaked_stochastic_2007}, we have that
\begin{align}
&\quad \quad \sum_{i =1 }^p \left[\left\{ {C}^{v} (i)^2  - \nu_{a(t)} \right \} \ind {\left\{\left|{C}^{v}(i)\right| > a(t) \right\}} - {C}^{v} (i)^2  +1 \right] \\
&\leqst \sum_{i =i }^p \left[  \left\{ Y_i^2  - \nu_{a(t)} \right \} \ind {\left\{\left|Y_i\right| > a(t) \right\}} - Y_i^2  +1 \right],
\end{align}
where $Y_i \iid \N(0,1)$ for $i \in [p]$. Hence, 
\begin{align}
\PP\left(\mathcal{E}_4^c\right) &\leq \sum_{0<v<n} \quad \sum_{t \in \mathcal{T}\setminus\{p\}} \PP\left[ \sum_{i =1 }^p  \left\{ Y_i^2  - \nu_{a(t)} \right \} \ind {\left\{\left|Y_i\right| > a(t) \right\}} \geq 63 r(t) \right] \\
&+ \sum_{0<v<n} \quad \sum_{t \in \mathcal{T}\setminus\{p\}} \PP\left\{ \sum_{i =1 }^p  \left\{ Y_i^2  - 1\right \} \leq - 5 h(p) \right\}\\
&\leq \frac{1}{18n}  + n \log_2(p) \PP\left\{ \sum_{i =1 }^p  \left( Y_i^2  - 1\right ) \leq - 5 h(p) \right\},
\end{align}
where we for the first sum used the same arguments as in Step 2. For the second sum, we have by Lemma \ref{lemma12} that
\begin{align}
\PP \left[    \sum_{i=1}^p (Y_i^2 -1) \leq -2 p^{1/2} \left\{\log(6n^2) + \log \log_2 p \right\}^{1/2}    \right] &\leq \frac{1}{6n^2 \log_2 p }.
\end{align}
Now, 
\begin{align}
    2 p^{1/2} \left\{\log(6n^2) + \log \log_2 p \right\}^{1/2} &\leq 2 \left[ 6 \left\{ \log n \vee \log\log(ep)\right\} \right]^{1/2}\\
    &\leq 5 h(p).
\end{align}Hence $\PP\left(\mathcal{E}_4^c\right) \leq {1}/{(18n)} +{1}/{(6n)} \leq {1}/{(3n)}$, and the proof is complete.
\end{proof}

The following Lemma gives high-probability control over the score statistic $S^v_{\gamma,(s,e]}$ used as a test statistic.
\begin{lemma}\label{lemma20}
Consider the model from Section \ref{secproblemdesc}, and assume $\sigma=1$. Let $r(t)$ be defined as in \eqref{rdef}.
 For any integer $v$ such that $s<v<e$, let $T^b_{(s,e]}(\cdot)$ be defined as in \eqref{c2def}, and define 
\begin{align}
\beta_{(s,e]}^v=  \sum_{i=1}^p T^v_{(s,e]}(\mu_{\roww{i}})^2,
\intertext{and}
k_{(s,e]} = \sum_{i=1}^p \ind \left\{ T^v_{(s,e]}(\mu_{\roww{i}})^2 = 0\right\}.
\end{align}
Note that if $\beta_{(s,e]}^v= 0$ for some $v$, then the open integer interval $(s,e)$ contains no changepoint. Define the events
\begin{align}
\mathcal{E}_5 :&= \left\{          
\forall 0\leq s <v<e\leq n ,\ \beta_{(s,e]}^v =0 \  ; \ S^v_{\gamma,(s,e]} < 0
\right\},\\
\mathcal{E}_6 :&=  \left\{
\forall 0\leq s <v<e\leq n \ ; \ S^v_{\gamma, (s,e]} \geq \beta_{(s,e]} - 8\left\{2\beta_{(s,e]}^v r(k_{(s,e]})\right\}^{1/2} - \left( \gamma + 106 \right)r(k_{(s,e]})
\right\}.
\end{align}
If $\gamma \geq 82$, then $\PP\left(\mathcal{E}_5 \cap \mathcal{E}_6 \right) \geq 1-{1}/{n}$. 
\end{lemma}
\begin{proof}\n
\textbf{Step 1.} We first show that $\PP\left(\mathcal{E}_5^c\right) \leq 1/(2n)$. 
Consider any integer triple of $s,e,v$ such that $0\leq s <v<e\leq n$ and $\beta_{(s,e]}^v = 0$. Fix any $t \in \mathcal{T}\setminus\{p\}$ (the case $t=p$ is handled later), and fix $x_t>0$, to be specified later. As $\beta_{(s,e]}^v = 0$, the open integer interval $(s,e)$ contains no changepoint, and thus ${C}^{v}_{(s,e]} (i) \iid \N(0,1)$ for all $i \in [p]$. By Lemma \ref{lemma6} we have that
\begin{align}
 \sum_{i =1}^p \left\{ {C}^{v}_{(s,e]} (i)^2  - \nu_{a(t)} \right \} \ind {\left\{\left|{C}^{v}_{(s,e]}(i)\right| > a(t) \right\}} &\geq 9 \left [  \left\{ p e^{-a(t)^2/2} x_t\right\}^{1/2} + x_t \right] \label{l20e3}
\end{align}
occurs with probability at most $e^{-x_t}$. Note that there are at most $n^3$ unique choices of the triple $(s,e,v)$. By a union bound, \eqref{l20e3} holds for some $0\leq s <v < e \leq n$ and some $t\in \mathcal{T}\setminus\{p\}$ with probability at most $n^3 \sum_{t\in \mathcal{T} \setminus\{p\}} e^{-x_t}$. Now set $x_t = 8 \left\{  \frac{p \log^2(n)}{t^2} \wedge r(t)  \right\}$ for all $t$. Then,
\begin{align}
\sum_{t\in \mathcal{T}\setminus\{p\}} e^{-x_t} &\leq \sum_{t\in \mathcal{T}\setminus\{p\}}   \exp\left\{-8\frac{p \log^2(n)}{t^2}\right\} + \sum_{t\in \mathcal{T}\setminus\{p\}} \exp\left\{ -8r(t)\right\}.
\end{align}
For the first sum, we have
\begin{align}
\sum_{t\in \mathcal{T}\setminus\{p\}}   \exp\left\{-8\frac{p \log^2(n)}{t^2}\right\} &\leq \sum_{k=0}^{\infty} \exp\left\{-8{\log(n)} 4^{k}\right\} \\
&= \sum_{k=0}^{\infty}\left( \frac{1}{n^8} \right)^{4^k}\\
&\leq \frac{1}{n^8} + \frac{1}{n^8} \sum_{k=1}^{\infty}\left( \frac{1}{n^8} \right)^{3k}\\
&=\frac{1}{n^8} \left( 1 + \frac{1}{n^{24}-1}\right).
\end{align}
For the second sum, noting that $8 r(t) = 8 \left\{ t \log \left( \frac{ep\log n}{t^2}\right) \vee \log n\right\} \geq 4 t \log \left( \frac{ep\log n}{t^2}\right) + 4\log n$, we have
\begin{align}
\sum_{t\in \mathcal{T}\setminus\{p\}} \exp\left\{ -8r(t)\right\} &\leq \frac{1}{n^4}  
\sum_{t\in \mathcal{T}\setminus\{p\}}  \left(   \frac{t^2}{ep\log n}   \right)^{4t}\\
&\leq \frac{1}{n^4 e^4} \left( 1 + \sum_{k=1}^{\infty} 4^{-4k} \right).
\end{align}
Hence, using that $n\geq 2$,
\begin{align}
n^3 \sum_{t\in \mathcal{T} \setminus\{p\}} e^{-x_t} \leq & \frac{1}{n^5} \left( 1 + \frac{1}{n^{24}-1}\right) + \frac{1}{n e^4} \left( 1 + \sum_{k=1}^{\infty} 4^{-4k} \right) \\
\leq & \frac{1}{10n}.
\end{align}
With this choice of $x_t$, using that $a^2(t) = 4\log\left( \frac{ep\log n}{t^2}  \right)$, we have that
\begin{align}
9 \left [  \left\{ p e^{-c^2(t)/2} x_t\right\}^{1/2} + x_t \right] &= 9 \left [  \left\{ p \frac{t^4}{e^2 p^2 \log^2 n} x_t\right\}^{1/2} + x_t \right] \\
&\leq 9 \left \{ \frac{t\surd{8}}{e}  + 8 r(t)\right\}\\
&\leq 82 r(t),
\end{align}
where we used that $x_t \leq 8r(t)$ and $x_t \leq {8p \log^2(n)}/{t^2}$, as well as the fact that $t \leq r(t)$ whenever $t \leq (p \log n)^{1/2}$. Hence, 
\begin{align}
&\PP \left[ \exists 0\leq s<v<e\leq n, \exists t \in \mathcal{T}\setminus\{p\} \ ; \ \sum_{i =1 }^p \left\{ {C}^{v}_{(s,e]} (i)^2  - \nu_{a(t)} \right \} \ind{\left\{|{C}^{v}_{(s,e]}(i)| > a(t) \right\}} \geq 82 r(t) \right]\notag\\
\leq & \frac{1}{10n}\label{hp20bound2}.
\end{align}
Now consider the case where $t = p$. If $p \leq ({p\log n})^{1/2}$, then similarly as above we have that
\begin{align}
&\PP \left[ \exists 0\leq s<v<e\leq n,  \ ; \ \sum_{i =1 }^p \left\{ {C}^{v}_{(s,e]} (i)^2  - \nu_{a(p)} \right \} \ind{\left\{|{C}^{v}_{(s,e]}(i)| > a(p) \right\}} \geq 82 r(p) \right]\notag\\
\leq & \frac{1}{10n} \label{hp20bound3}.
\end{align}
If we instead have $p > (p \log n)^{1/2}$ (in which case $a(p) = 0$ and $\nu_{a(p)} = 1$), then for any $0\leq s<v<e\leq n$, we have
\begin{align}
\sum_{i =1 }^p \left\{ {C}^{v}_{(s,e]} (i)^2  - \nu_{c(p)} \right \} \ind{\left\{|{C}^{v}_{(s,e]}(i)| > c(p) \right\}} &= \sum_{i =1 }^p  {C}^{v}_{(s,e]} (i)^2 - p .  
\end{align}
As $\sum_{i =1 }^p {C}^{v}_{(s,e]} (i)^2   \sim \chi^2_{p}$, we obtain from Lemma \ref{lemma12} that
\begin{align}
\sum_{i =1 }^p \left\{  {C}^{v}_{(s,e]} (i)^2 \right\} -p  & \geq  2 \left\{p \log(4n^4)\right\}^{1/2} + 2\log(4n^4),
\end{align}
occurs with probability at most ${1}/{(4n^4)}$. Using that $n\geq 2$ and $r(p) \geq \log n$, we obtain by a union bound that
\begin{align}
&\PP \left[\exists 0\leq s <v<e \leq n \ ; \ \sum_{i =1 }^p \left\{ {C}^{v}_{(s,e]} (i)^2  - \nu_{c(p)} \right \} \ind{\left\{|{C}^{v}_{(s,e]}(i)| > c(p) \right\}} \geq  17 r(t) \right]\notag\\ &\leq \frac{1}{4n} \label{l20e4},
\end{align}
Combining \eqref{hp20bound2}, \eqref{hp20bound3} and \eqref{l20e4} by a union bound, and using that $\gamma \geq 82$, we get that $\PP(\mathcal{E}_5^c) \leq 1/(2n)$.\n

\textbf{Step 2}. Now we show that $\PP\left( \mathcal{E}_6^c \right) \leq 1/(2n)$. Consider any $0\leq s<v<e\leq n$. Without loss of generality, assume that $T^v_{(s,e]}(\mu\roww{1})^2 \geq T^v_{(s,e]}(\mu\roww{2})^2 \geq \ldots T^v_{(s,e]}(\mu\roww{p})^2$. Let $z$ denote the smallest integer in $\mathcal{T}$ no smaller than $k_{(s,e]}$, where we suppress the dependence of $z$ on $s,e$ in the notation. For $z\leq (p\log n)^{1/2}$, observe that
\begin{align}
&\sum_{i =1 }^p \left\{ {C}^{v}_{(s,e]} (i)^2  - \nu_{a(z)} \right \} \ind{\left\{|{C}^{v}_{(s,e]}(i)| > a(z) \right\}} \\
\geq &\sum_{i=1}^{k_{(s,e]}} \left\{ {C}^{v}_{(s,e]} (i)^2  - \nu_{a(z)} \right \} + \sum_{i =k_{(s,e]}+1 }^p \left\{ {C}^{v}_{(s,e]} (i)^2  - \nu_{a(z)} \right \} \ind{\left\{|{C}^{v}_{(s,e]}(i)| > a(z) \right\}}\label{thezum}.
\end{align}
We lower bound the two sums separately. For each $i\in [p]$ we have that ${C}^{v}_{(s,e]} = T^v_{(s,e]}(\mu\roww{i}) + T^v_{(s,e]}(W\roww{i}) \overset{\text{ind}}{\sim} \N (  T^v_{(s,e]}(\mu\roww{i}) \ , \ 1   )$. Let $x_t = 8 \left\{  {p \log^2(n)}/{t^2} \wedge r(t)  \right\}$ for all $t \in [p]$, as in Step 1. For the first sum, noting that $\sum_{i=1}^{k_{(s,e]}} {C}^{v}_{(s,e]} (i)^2 \sim \chi_{k_{(s,e]}}^2 (\beta^v_{(s,e]})$ (a non-central Chi Square distribution with $k_{(s,e]}$ degrees of freedom and non-centrality parameter $\beta^v_{(s,e]}$), we have by Lemma \ref{lemma12} that
\begin{align}
&\PP\left[   \sum_{i=1}^{k_{(s,e]}} \left\{ {C}^{v}_{(s,e]} (i)^2 -\nu_{a(z)} \right\}< k_{(s,e]} - k_{(s,e]}\nu_{a(z)} + \beta^v_{(s,e]} -2 \left\{x_{z}\left(z + 2 \beta^v_{(s,e]}\right)\right\}^{1/2}  \right] \notag\\
&\leq  e^{-x_{z}}.
\end{align}
Note that, since $k_{(s,e]}\leq z \leq (p\log n)^{1/2}$, we have $z \leq r(z)$ and $k_{(s,e]}\leq r(k_{(s,e]})$. Moreover, by Lemma \ref{lemma10}, we have
$\nu_{a(z)}^2 \leq 2+ a^2(z)\leq 2+a^2(k_{(s,e]})$, where we for the last inequality used that $z\geq k_{(s,e]}$ and that $t \mapsto a^2(t)$ is decreasing. Since $z \leq 2 k_{(s,e]}$, it also holds that $r(z)\leq 2 r(k_{(s,e]})$. Hence, 
\begin{align}
\PP\left[   \sum_{i=1}^{k_{(s,e]}} \left\{ {C}^{v}_{(s,e]} (i)^2 -\nu_{a(z)} \right\}< \beta^v_{(s,e]} -8 \left\{2\beta^v_{(s,e]} r(k_{(s,e]})\right\} - 14 r(k_{(s,e]})  \right] \leq e^{-x_z}.
\end{align}
For the second sum, we obtain from Lemma \ref{lemma5} that
\begin{align}
&\PP \left[  \sum_{i =k_{(s,e]}+1 }^p \left\{ {C}^{v}_{(s,e]} (i)^2  - \nu_{a(t)} \right \} \ind{\left\{|{C}^{v}_{(s,e]}(i)| > a(z) \right\}} \leq - 5 \left [  \left\{ p e^{-b^2(z)/2}x_z\right\}^{1/2} + x_z \right]\right]\notag\\ &\leq  e^{-x_t}\label{lll134}.
\end{align}
By the definition of $x_z$, we have that 
\begin{align}
5 \left [  \left\{ p e^{-b^2(z)/2}x_z\right\}^{1/2} + x_z \right]&\leq 46 r(z)\\
&\leq 92 r(k_{(s,e]}).
\end{align} By a union bound over the two sums in \eqref{thezum}, we have that
\begin{align}
&\sum_{i =1 }^p \left\{ {C}^{v}_{(s,e]} (i)^2  - \nu_{a(z)} \right \} \ind{\left\{|{C}^{v}_{(s,e]}(i)| > a(z) \right\}} \notag\\ 
&<\beta^v_{(s,e]} -8 \left\{2\beta^v_{(s,e]}(z) r(k_{(s,e]})\right\}^{1/2} - 106 r(k_{(s,e]}) \label{thebound1}
\end{align}
occurs with probability at most $2e^{-x_t}$. 

Now suppose that $z> (p \log n)^{1/2}$. Then,
\begin{align}
\sum_{i =1 }^p \left\{ {C}^{v}_{(s,e]} (i)^2  - \nu_{a(z)} \right \} \ind{\left\{|{C}^{v}_{(s,e]}(i)| > a(z) \right\}} &= \sum_{i =1 }^p \left\{ {C}^{v}_{(s,e]} (i)^2  \right \}  -p.
\end{align}
Using that $\sum_{i =1 }^p  {C}^{v}_{(s,e]} (i)^2 \sim \chi_p^2(\beta^v_{(s,e]})$, we have by Lemma \ref{lemma12} that
\begin{align}
\sum_{i =1 }^p \left\{ {C}^{v}_{(s,e]} (i)^2  \right \} -p < \beta^v_{(s,e]} - 2 \left\{\log (4n^4)(p + 2 \beta^v_{(s,e]})\right\}^{1/2}
\end{align}
occurs with probability at most $1/(4n^4)$. In particular, since $\log n\leq r(t)$ for all $t$, we have $r(z) = r(\surd({p\log n})) \leq 2r(k_{(s,e]})$ and $n\geq2$, this implies that \eqref{thebound1} occurs probability at most $1/(4n^4)$ whenever $z \geq (p\log n)^{1/2}$. By a union bound over $0\leq s <v<e\leq n$, we obtain that 
\begin{align}
\PP\left(\mathcal{E}_6^c\right) &\leq n^3 \left(\frac{1}{4n^4} + 2\sum_{t \in \mathcal{T} \setminus\{p\}} e^{-x_t}\right)\\
&\leq \frac{1}{4n} + \frac{1}{5n}\\
&\leq \frac{1}{2n},\end{align}
where we used the same approach as in Step 1 to bound $\sum_{t \in \mathcal{T}\setminus\{p\}} e^{-x_t}$.
The proof is complete. 
\end{proof}

\begin{lemma}\label{lemma14}
Let $\mathcal{M}$ denote the collection of seeded intervals generated by Algorithm \ref{alg:seededintervals} with parameters $\alpha \in (1,2]$ and $K\geq 2$.
Then for all real numbers $h>0$ such that $h\leq n/2$, and all integers $\eta$ such that $3h/2 \vee 1\leq \eta \leq n- \left( 3h/2 \vee 1\right)$, there exists integers $l\geq1$ and $v$ such that the following holds.
\begin{enumerate}[label=(P\arabic*)]
\item $(v-l, v+l] \in \mathcal{M}$\label{l14p1};
\item $h/2 \leq l\leq h\vee 1$\label{l14p2};
\item $|v-\eta| \leq l/K\leq l/2$\label{l14p3}.
\end{enumerate}
In particular, $(v-l, v+l] \subseteq (\eta - (3/2 h \vee 1), \eta+(3/2 h \vee 1)]$.
\end{lemma} 
\begin{proof}
Define the recursive sequence $(l_j)_{j \in \NN}$ by $l_1 = 1$, and $l_{j+1} = \max \left\{l_j+1,  \lfloor \alpha l_j\rfloor \right\}$ for $j\in \NN$. Let $H = \max \{ j \in \NN \ : \ l_j \leq n/2\}$. Formally, the set $\mathcal{M}$ of seeded intervals generated by Algorithm \ref{alg:seededintervals} is given by
\begin{align}
\mathcal{S} &= \bigcup_{l \in \{l_1, \ldots, l_H\}}  \mathcal{I}_l,
\intertext{where}
\mathcal{I}_l &= \left\{  (n-2l,n]   \right\}  \cup \bigcup_{i=0}^{ \left \lfloor \frac{n-2l}{s_l}\right \rfloor} \left\{(i s_l \ , \ i s_l +2l ]\right\},\\
s_l &= \max \left\{ 1, \left\lfloor \frac{l}{K} \right\rfloor\right\}.
\end{align}
Note that, for any $j \in [H-1]$, it holds that $l_{j+1}/l_j \leq \max\{2, \alpha\}=2$. Hence, there must exist an integer $j \in [H]$ such that $h/2 \leq l_j \leq h\vee 1$. Moreover, by the definition of $\mathcal{I}_{l_j}$, there must exist an integer $v$ such that $|v - \eta| \leq \lfloor l_j/K \rfloor \leq l_j/2$ and $(v-l_j, v+l_j] \in \mathcal{I}_{l_j}$. This proves the first three claims. For the last claim, note that
\begin{align}
v - l_j \ &= v - \eta + \eta - l_j\\
&\geq - \lfloor l_j/K \rfloor + \eta - l_j \\
&\geq \eta - (3/2h \vee 1).
\end{align}
Similarly, we have that $v +l_j \leq \eta + (3/2 h \vee 1)$.
\end{proof}

\begin{lemma}\label{lemma16}
Let $Y \sim \N(\theta,1)$, $\theta \in \RR$, $a>0$ and $\nu_a = \EE\left(  Y^2 \ \mid \ |Y|\geq a  \right)$. 
Let $ A= \left(    Y^2 - \nu_a \right) \ind {\left(\left|Y \right| \geq a \right)}$ and $B= Y^2  -1$. Then $A-B$ is stochastically decreasing in $|\theta|$.
\end{lemma}
\begin{proof}
It is equivalent to show that $B-A$ is stochastically increasing in $|\theta|$. Note first that $Y^2$ has a Chi Square distribution with non-centrality parameter $\theta^2$, which is stochastically increasing in $|\theta|$. Further, we have that
$B-A = f(Y^2),$
where the function $f$ is given by
\begin{align}
f(x) = \begin{cases} x - 1 , &\text{ if } x < a^2\\
\nu_a - 1,  &\text{ otherwise.}\end{cases}
\end{align}
Since $\nu_a \geq a^2$, $f$ is an increasing function. By \citet[Theorem 1.A.3(a)]{shaked_stochastic_2007} , $B-A$ must be stochastically increasing in $|\theta|$.
\end{proof}

\begin{lemma}\label{lemma221}
Let $\mathcal{M}$ denote the set of candidate intervals generated from Algorithm \ref{alg:seededintervals} with fixed input parameters $\alpha >1$, $K>1$ and $n \in \NN$. Then the number of distinct triples of integers $s,e,v$ such that $(s,e] \in \mathcal{M}$ and $s<v<e$ is of order $\mathcal{O}(n \log n)$.
\end{lemma}
\begin{proof}
Let $\alpha$ and $K$ be given, and define the recursive sequence $(l_j)_{j \in \NN}$ by $l_1 = 1$, and $l_{j+1} = \max \left(l_j+1,  \lfloor \alpha l_j\rfloor \right)$ for $j\in \NN$. Let $H = \sup \{ j \in \NN \ : \ l_j \leq n/2\}$. Formally, the set $\mathcal{M}$ of seeded intervals generated by Algorithm \ref{alg:seededintervals} is given by
\begin{align}
\mathcal{S} &= \bigcup_{l \in \{l_1, \ldots, l_H\}}  \mathcal{I}_l,
\intertext{where}
\mathcal{I}_l &= \left\{  (n-2l,n]   \right\}  \cup \bigcup_{i=0}^{ \left \lfloor \frac{n-2l}{s_l}\right \rfloor} \left\{(i s_l \ , \ i s_l +2l ]\right\},\intertext{and}
s_l &= \max \left\{ 1, \left\lfloor \frac{l}{K} \right\rfloor\right\}.
\end{align}
For any $(s,e] \in \mathcal{I}_l$, there are precisely $2l -1< 2l$ integers $v$ such that $s<e<v$. Hence, the number $N$ of distinct triples of integers $s,e,v$ such that $(s,e] \in \mathcal{M}$ and $s<v<e$ therefore satisfies
\begin{align}
N & < \sum_{l \in \{l_1, \ldots, l_J\}} 2l |\mathcal{I}_l|.
\end{align}
For all $l<K$, we have that $|\mathcal{I}_l| \leq n$, and so $2l |\mathcal{I}_l|\leq 2ln< 2Kn$. For $l \geq K$ we have that $\left\lfloor l/K \right\rfloor \geq l/(2K)$, and so $2l|\mathcal{I}_l| \leq  4Kn$.
Therefore, 
\begin{align}
N &< \sum_{l \in \{l_1, \ldots, l_J\}} 4Kn \\
&= 4H Kn.
\intertext{Noting that $H \leq \lceil \frac{1}{\alpha -1} \rceil  + \log_{\alpha}n$, we thus get}
N&< 4\left( \lceil \frac{1}{\alpha -1} \rceil  + \log_{\alpha}n \right) Kn\\
&= \mathcal{O}(n\log n),
\end{align}
which gives the desired result.
\end{proof}

In the following we restate some useful Lemmas from \citet{baranowski_narrowest-over-threshold_2019}.
\begin{lemma}[\citealt{baranowski_narrowest-over-threshold_2019}, Lemma 2]\label{lemma2NOT}\n
Consider the model from Section \ref{secproblemdesc}, assuming that $p=1$. Let the CUSUM transformation $T^v_{(s,e]}(\cdot)$ be defined as in \eqref{c2def}. Suppose $s<e$ are such that $\eta_{j-1} \leq s < \eta_j <e \leq \eta_{j+1}$ for some $j \in [J]$. Then, 
\begin{align}\underset{s<v<e}{\max} \ T^{v}_{(s,e]}(\mu)^2 = T^{\eta_j}_{(s,e]}(\mu)^2 \ \ \begin{cases}
\geq \frac{1}{2} {\delta} \theta_j^2\\
\leq {\delta} \theta_j^2.\end{cases}\end{align}
\end{lemma}





Given an $n \in \NN$ and any integer $0<v<n$, define the $n$-dimensional vector $\Psi^v \in \RR^n$ to have $l$th element given by
\begin{align}
\Psi^v(l)= \begin{cases}
\ \left({\frac{n-v}{nv}}\right)^{1/2}, & \text{ for } l = 1, \ldots, v,\\
-\left({\frac{v}{n(n-v)}}\right)^{1/2}, & \text{ for } l = v+1\ldots, n. \\
\end{cases}\label{psi2def}
\end{align}

\begin{lemma}[\citealt{baranowski_narrowest-over-threshold_2019}, Lemma 4]\label{lemma4NOT}
Consider the model from Section \ref{secproblemdesc}, assuming that $p=1$. Let the CUSUM transformation $T^v_{(s,e]}(\cdot)$ be defined as in \eqref{c2def}. Pick any interval $(s,e] \subset (0,n]$ such that the open integer interval $(s, e)$ contains precisely one changepoint $\eta_j$. Pick any integer $v$ such that $s<v<e$. Define $\rho = |\eta_j - v|$,  $\delta_{\text{L}} = \eta_j - s$, and $\delta_\text{R} = e - \eta_j$. Then, 
\begin{align} \normm{ \Psi^{\eta_j}_{(s,e]} \langle \mu, \Psi^{\eta_j}_{(s,e]} \rangle  - \Psi^v_{(s,e]} \langle \mu, \Psi^v_{(s,e]}\rangle }_2 =   T^{\eta_j}_{(s,e]}(\mu) ^2 -    T^{v}_{(s,e]}(\mu) ^2.\end{align} 
Moreover, 
\begin{align}
    &(1) \ \text{for any } \eta_j \leq v < e, \    T^{\eta_j}_{(s,e]}(\mu)  ^2 -    T^{v}_{(s,e]}(\mu)  ^2 = \frac{\rho \delta_{\text{L}}}{\rho + \delta_{\text{L}}} \theta_j^2;\\
    &(2) \ \text{ for any } s<v\leq \eta_j, \    T^{\eta_j}_{(s,e]}(\mu)  ^2 -    T^{v}_{(s,e]}(\mu)  ^2 = \frac{\rho \delta_{\text{R}}}{\rho + \delta_{\text{R}}} \theta_j^2.
\end{align}
\end{lemma}